\documentclass[11pt]{article}
\usepackage{fullpage}
\usepackage{url}
\usepackage[colorlinks=true,linkcolor=Emerald,citecolor=RoyalBlue]{hyperref}
\usepackage{amsmath,nccmath, amsfonts,amsthm,multirow,mdframed,amssymb}
\usepackage{thmtools,thm-restate}

\usepackage{times}
\usepackage{paralist}
\usepackage{graphicx}
\usepackage{floatpag}
\usepackage{float}
\DeclareMathOperator{\supp}{supp}
\usepackage{bbold}
\usepackage{enumitem}
\usepackage{subcaption} 
\usepackage{soul}
\usepackage{comment}
\usepackage{calc}
\usepackage{upgreek}
\usepackage{mathtools}
\usepackage{color}
\usepackage[dvipsnames]{xcolor}
\usepackage{xspace}
\usepackage{tcolorbox}
\usepackage{cleveref}
\usepackage{stackengine}
\usepackage{tocloft}
\usepackage{physics}
\usepackage{quantikz}

\stackMath
\allowdisplaybreaks

\newtheorem{theorem}{Theorem}
\newtheorem*{theorem*}{Theorem}
\newtheorem{lemma}[theorem]{Lemma}
\newtheorem{corollary}[theorem]{Corollary}
\newtheorem{fact}[theorem]{Fact}

\newtheorem{remark}{Remark}[section]

\newcommand{\poly}[1]{\mathrm{poly}\left(#1\right)}
\newcommand{\pr}[1]{\mathrm{Pr}\left[#1\right]}

\newcommand{\lr}[1]{\left(#1\right)}
\newcommand{\EE}[1]{\mathbb{E}\left[#1\right]}
\newcommand{\E}{\mathbb{E}}
\newcommand{\Var}{\mathrm{Var}}
\newcommand{\C}{\mathbb{C}}

\newcommand{\R}{\mathbb{R}}

\newcommand{\cG}{\mathcal{G}}

\newcommand{\cJ}{\mathcal{J}}

\newcommand{\cM}{\mathcal{M}}
\newcommand{\cN}{\mathcal{N}}

\newcommand{\cP}{\mathcal{P}}

\newcommand{\cS}{\mathcal{S}}

\newcommand{\wh}{\widehat}

\newcommand{\sep}{\mathrm{sep}}
\newcommand{\mix}{\mathrm{mix}}

\newcommand{\dist}{\mathrm{distinct}}
\newcommand{\rep}{\mathrm{rep}}
\newcommand{\conn}{\mathrm{conn}}

\makeatletter
\let\c@fconjecture\c@conjecture
\makeatother

\makeatletter
\let\c@fconj\c@conj
\makeatother

\title{\fontsize{15.5}{19}\selectfont A rigorous quasipolynomial-time classical algorithm \\for SYK thermal expectations}
\author{Alexander Zlokapa\\
\textit{Center for Theoretical Physics --- a Leinweber Institute, MIT}}
\date{}

\begin{document}
\maketitle
\thispagestyle{empty}
\begin{abstract}
    One of the most natural tasks in quantum simulation is to estimate a local observable of a system in thermal equilibrium, i.e., in its Gibbs state. This task is known to be BQP-complete at asymptotically small temperatures due to Rouz\'e, Fran\c ca and Alhambra (STOC'25), but at constant temperatures it remains unknown if quantum advantage is possible. The Sachdev-Ye-Kitaev (SYK) model, a local random Hamiltonian family with strongly interacting fermions, has been considered a promising candidate for such quantum advantage. A recent line of work initiated by Hastings and O'Donnell (STOC'22) showed that constant-temperature Gibbs states are poorly described by Gaussian states and require polynomial quantum circuit complexity. Beyond these rigorous results, common approaches such as tensor networks, quantum Monte Carlo path integration, and message-passing-like algorithms are obstructed by the model's large entanglement, sign problem, and mean-field interactions. Despite these barriers, we show the following results.
    \begin{enumerate}[label=(\roman*)]
        \item We give a deterministic classical algorithm to estimate SYK local thermal expectations to arbitrarily high accuracy for sufficiently large constant temperature. Our algorithm is quasipolynomial-time at any inverse polynomial error, which suffices to resolve instance-to-instance fluctuations.
        \item We rule out any phase transition above this temperature by controlling the locations of complex zeros of the partition function. This makes progress towards proving the SYK phase diagram conjectured from the replica trick in physics.
    \end{enumerate}
    Our main technical contribution is a new cluster expansion for disordered quantum many-body systems. The use of cluster expansions to obtain classical algorithms for quantum many-body systems began with Harrow, Mehraban and Soleimanifar (STOC'20), but such expansions fail for all-to-all interactions. Classically, sufficient control over a mean-field cluster expansion to show a result analogous to ours was only recently achieved by Bencs et al. (STOC'26) for the Sherrington-Kirkpatrick model; such expansions fail for noncommuting models. Similarly, standard tools from random matrix theory have so far failed to rigorously compute basic quantities such as the annealed SYK free energy. We thus expect our new cluster expansion, based on Wick pairs, to be broadly useful for disordered quantum models.
\end{abstract}

% arXiv abstract
% Estimating local observables in Gibbs states is a central problem in quantum simulation. While this task is BQP-complete at asymptotically low temperatures, the possibility of quantum advantage at constant temperature remains open. The Sachdev-Ye-Kitaev (SYK) model is a natural candidate: at any constant temperature, its Gibbs states have polynomial quantum circuit complexity and are not described by Gaussian states. Rigorous analyses of the SYK model are difficult due to the failure of known techniques using random matrix theory, cluster expansions, and rigorous formulations of the quantum path integral and replica trick. Despite this, we give a rigorous proof of a quasipolynomial-time classical algorithm that estimates SYK local thermal expectations at sufficiently high constant temperature. Our result introduces a new Wick-pair cluster expansion that we expect to be broadly useful for disordered quantum many-body systems.

\maketitle

\section{Introduction}

One of the most widely anticipated applications of quantum computers is quantum simulation, measuring either static or dynamic properties of systems found in chemistry and physics. The theoretical study of individual models, such as the Fermi-Hubbard model, is often beyond the reach of analytical tools. Nonetheless, one would like evidence that quantum computers can provide advantage on ``typical'' problem instances. Understanding the average-case complexity of a distribution over problem instances is a subtle task, as it can differ greatly from the worst-case complexity; classically, the hardness of random $K$-SAT instances has a rich history~\cite{kirkpatrick1994critical,monasson1999determining,achlioptas2006solution,mezard2005clustering,coja2010better,hetterich2016analysing}. In the quantum setting, most known average-case complexity results are for tasks (e.g., sampling measurement outcomes) over problem instances (e.g., random circuits) that do not necessarily resemble the natural tasks and problem instances that would arise in quantum simulation~\cite{hangleiter2023computational}.

Consequently, most evidence of classical hardness for physical systems lies in heuristic qualities that obstruct common classical algorithms. If a system is highly entangled, it may be difficult to evaluate via matrix product states or tensor networks~\cite{vidal2003efficient}; if a path integral has a sign problem, it may be difficult to apply Monte Carlo methods~\cite{troyer2005computational}. In this work, we study a random Hamiltonian ensemble that has both of these properties but unusually has \emph{far stronger} evidence of being classically nontrivial as well, due to a line of work beginning with~\cite{hastings2022optimizing}.

We study the Sachdev-Ye-Kitaev (SYK) model, a $q$-local random Hamiltonian ensemble defined in terms of Majorana fermions, operators $\{\psi_i\}_{i=1}^n$ satisfying $\{\psi_i, \psi_j\} = \delta_{ij}$:
\begin{align}\label{eq:hsyk}
    H = \sum_{I \in \binom{[n]}{q}} J_I \psi_I, \quad J_I \sim \cN\lr{0, \frac{(q-1)!}{n^{q-1}}}, \quad \psi_I = i^{q/2} \psi_{i_1}\cdots\psi_{i_q}, \quad q \geq 4 \text{ even}.
\end{align}
The SYK model has been extensively studied in condensed matter and high-energy physics~\cite{sachdev1993gapless,parcollet1999non,georges2000mean,georges2001quantum,kitaev2015SYK,maldacena2016remarks,maldacena2016bound,kitaev2018soft}. Moreover, its Gibbs state is expected to be quantumly easy to prepare due to non-rigorous arguments from physics~\cite{almheiri2024universal,zhang2019evaporation,maldacena2021syk,schuster2025cooling}. The SYK model has been considered a candidate for quantum advantage~\cite{anschuetz2025strongly} due to properties such as a sign problem~\cite{pan2021yukawa}, large entanglement~\cite{liu2018quantum} and magic~\cite{bettaque2026magic}. More formally, \cite{hastings2022optimizing} showed the failure of Gaussian states to obtain good energies (i.e., typical energies of constant-temperature Gibbs states). The SYK thermal state thus lives in a regime where Hartree-Fock methods fail to give good approximations~\cite{szabo2012modern}, which has been considered a promising avenue for quantum advantage~\cite{mcardle2020quantum}. Subsequent work showed the following properties that suggest classical hardness at constant temperature. 
(Precise statements are proven in Appendix~\ref{sec:obstruct}; they all follow straightforwardly from known results, although we report more robust versions of these results.)

\begin{remark}[Obstructions to classical algorithms for SYK]\label{rem:obstruct}
    For any constant inverse temperature $\beta > 0$, there exists a constant $\epsilon > 0$ such that the SYK Gibbs state $\rho_\beta$ has the following properties:
    \begin{enumerate}
        \item any state within $\epsilon$ trace distance of a purification of $\rho_\beta$ requires $\Omega(n^{1+q/2}/\log n)$ two-qubit gates to prepare~\cite{anschuetz2025strongly},
        \item no state within $\epsilon$ trace distance of $\rho_\beta$ can be expressed as a convex combination of superpositions of fewer than $\exp[o(n^{1/4})]$ orthogonal Gaussian states~\cite{hastings2023field}, and
        \item local thermal expectations have instance-to-instance fluctuations of magnitude $\Omega(1/\poly n)$.
    \end{enumerate}
\end{remark}

In this work, we study the task of estimating local observables of the SYK model in thermal equilibrium. We normalize the $n$-fermion Hamiltonian such that $\norm{H} = \Theta(n)$, ensuring each site has a constant energy density. Weak interactions between this system and an environment will eventually equilibrate to a constant-temperature Gibbs state, $\rho_\beta = e^{-\beta H} / \Tr e^{-\beta H}$ for $\beta = \Theta(1)$. The independence of $\beta$ from $n$ ensures that temperature does not change with system size, yielding a thermodynamically reasonable Gibbs state. The physically relevant static equilibrium properties of the system are then described by local expectations: for Hermitian $O$ acting on at most $O(1)$ sites, one wishes to estimate $\Tr(O \rho_\beta)$ up to inverse polynomial error. Understanding if this task admits quantum advantage is a central open question in quantum simulation, in part due to the recent development of quantum algorithms to prepare Gibbs states~\cite{temme2011quantum,zhang2023dissipative,rouze2024optimal,jiang2024quantum,gilyen2024quantum,ding2025end,ding2025efficient,chen2025efficient,chen2025quantum}. Adjacent results, corresponding to less physically motivated tasks, suggest the answer is in the affirmative: it is known that sampling global measurement outcomes of local Hamiltonians at constant temperature can simulate IQP circuit sampling~\cite{bergamaschi2024quantum,rajakumar2024gibbs}, and that estimating local expectations of local Hamiltonians at $\beta=\poly{n}$ is BQP-complete~\cite{rouze2025efficient}. However, quantum advantage via constant-temperature local expectations has remained elusive.

The SYK model---traditionally analyzed non-rigorously in the physics literature with quantum path integrals and the replica trick~\cite{sachdev2023quantum}---is difficult to study rigorously for a variety of reasons outlined in Sec.~\ref{sec:prior}. Nonetheless, beyond the rigorous obstructions described in Remark~\ref{rem:obstruct}, physics arguments have also fallen short of suggesting efficient classical algorithms for SYK thermal expectations. The sign problem, all-to-all connectivity, and inability to control errors in expansions have obstructed quantum message-passing-like algorithms~\cite{bapst2013quantum}. Only very recently~\cite{zlokapa2026syk} has a classical algorithm similar to this work been studied, albeit using non-rigorous replica computations and path integrals. In contrast, \emph{classical} mean-field models were characterized by a long history of algorithms that developed initially heuristically from physics arguments and ultimately rigorously by formalizing the physics intuition. Message-passing algorithms to estimate thermal expectations were non-rigorously studied as early as~\cite{thouless1977solution,plefka2002modified} and rigorously in~\cite{talagrand2010mean,bolthausen2014iterative,chen2021convergence,el2022sampling}.

Independent of motivations within quantum computing, the rigorous study of disordered mean-field models has been an active topic in spin glass theory and average-case complexity. In recent decades, a deep connection has emerged between phase transitions in spin glasses and the computational complexity of tasks such as sampling, counting, and estimating observables. Perhaps the most studied model is the Sherrington-Kirkpatrick model~\cite{sherrington1975solvable}, $H_\mathrm{SK} = \sum_{i < j} A_{ij} Z_i Z_j$ for $A_{ij} \sim \cN(0, 1/n)$. The model was understood at a non-rigorous level via the replica trick in statistical physics~\cite{parisi1979infinite,mezard1987spin}; rigorous upper and lower bounds for tasks including optimization and Gibbs sampling were shown only much later~\cite{montanari2025optimization,el2022sampling}. Only very recently, in the technically impressive work of~\cite{bencs2025zeros}, has an efficient algorithm been provided to estimate the Sherrington-Kirkpatrick partition function in the replica symmetric phase. As this task is related to estimating local thermal expectations, our work is closely related to~\cite{bencs2025zeros}, although---due to the variety of proof obstructions discussed in Sec.~\ref{sec:prior} that arise in the quantum case---our techniques are very different. For example, the key technical contribution of~\cite{bencs2025zeros} was to control the locations of the complex zeros of the partition function via a lengthy cluster expansion based on the famous result of~\cite{aizenman1987some}. These cluster expansions break down for noncommuting Hamiltonians, and thus we will need to develop new cluster expansions for mean-field quantum models.

\subsection{Summary of results}

Our first contribution shows a region at nonzero constant temperature where the SYK partition function does not contain complex zeros, i.e., complex $\beta$ such that $Z(\beta) = 0$. In physics language, this rigorously rules out the presence of a phase transition above a particular constant temperature~\cite{lee1952statistical,fisher1965nature}.

\begin{theorem}[Zero-free disk of SYK partition function, informal]\label{thm:main1}
With probability $1-o(1)$ over the disorder of $q$-local SYK Hamiltonian $H$, the partition function $Z(\beta) = \Tr e^{-\beta H}$ satisfies $Z(\beta) \neq 0$ for all
\begin{align}
    |\beta| \leq \frac{2^{q/2}}{q}\cdot \frac{19}{40e^{1/4}}\lr{\sqrt{1+\frac{4}{e}}-1}.
\end{align}
\end{theorem}

To the best of our knowledge, this is the first rigorous control of the complex zeros of a mean-field quantum spin glass. Our proof uses a cluster expansion that directly exploits the noncommutative properties of the SYK model. This departs from prior works in the quantum setting, where the difficulty of analyzing noncommuting Hamiltonian terms limited the use of cluster expansions to define polymers that grow with the support of combinations of Hamiltonian terms, yielding results that only encompass bounded-degree models~\cite{harrow2020classical,mann2021efficient,yao2022polynomial,mann2024algorithmic,zlokapa2026syk}. In the mean-field setting, these polymers do not suffice. Moreover, one cannot immediately borrow cluster expansions used in the classical mean-field setting: the cycle expansion of~\cite{aizenman1987some,bencs2025zeros} breaks down due to noncommutation. Indeed, far simpler properties than zero-freeness also become subtle to analyze in the quantum setting. For the SYK model, it is currently not known how to rigorously compute the annealed partition function, which in comparison is trivial (and zero-free) for classical spin glasses (see \eqref{eq:gmgf}). Our proof technique circumvents this issue by establishing zero-freeness of the annealed free energy via a cluster expansion whose polymers are constructed from Wick pairs.

To exploit the noncommutativity of the SYK Hamiltonian, one can use the tools developed in \cite{anschuetz2025strongly} to show the self-averaging property of the SYK model. In particular, \cite{anschuetz2025strongly} showed that the SYK partition function satisfies
\begin{align}\label{eq:selfavg}
	\left|\frac{1}{n}\E \log Z(\beta) - \frac{1}{n} \log \E \, Z(\beta)\right| = O(n^{-q/2}),
\end{align}
suggesting that the model is replica symmetric for all $\beta$.
Importantly, this self-averaging result does not rule out the presence of a phase transition, unlike Theorem~\ref{thm:main1}. To understand this, we briefly recall the Lee-Yang theory of phase transitions for external fields~\cite{lee1952statistical}, as extended by Fisher for complex-valued temperatures~\cite{fisher1965nature}. At real temperature $1/\beta$, the partition function $Z(\beta)$ is a sum of positive numbers and thus satisfies $Z(\beta) \neq 0$. A phase transition occurs at real $\beta_c$ when the free energy density $-\frac{1}{\beta n} \log Z(\beta)$ becomes non-analytic at $\beta_c$ in the $n\to\infty$ limit. This corresponds to an extensive density of complex zeros ($\beta$ such that $Z(\beta) = 0$) asymptotically approaching the $\Re \beta$ axis, pinching it in the limit at $\beta_c$. Theorem~\ref{thm:main1} is thus a sufficient condition to show the absence of a phase transition. One might hope that the proof of \eqref{eq:selfavg} can be extended to similarly show the absence of a phase transition, but as described in Sec.~\ref{sec:prior}, a naive extension of the proof technique does not suffice to show Theorem~\ref{thm:main1}. Indeed, the non-rigorous computation of~\cite{zlokapa2026syk} clarifies that \eqref{eq:selfavg} should not hold for all complex $\beta$.

Our second contribution is a quasi-polynomial classical algorithm that estimates local, bounded-norm observables to arbitrary inverse polynomial error on the SYK thermal state above some constant temperature. This follows by extending Theorem~\ref{thm:main1} to apply to the perturbed partition function $Z(\beta,\lambda) = \Tr e^{-\beta(H + \lambda O)}$ given a local observable $O$. One can use the absence of a phase transition in the perturbed model to identify a zero-free region in which the partition function can be estimated due to analyticity. We use Barvinok's interpolation method~\cite{barvinok2014computing,barvinok2016computing,barvinok2016approximating,barvinok2016combinatorics,barvinok2018approximating} to estimate the thermal expectation of $O$. Explicitly, we compute a finite difference approximation of $\Tr(O\rho_\beta) = -\frac{1}{\beta} \partial_\lambda\big|_{\lambda=0} \log Z(\beta, \lambda)$ obtained by truncating a Taylor series of $\log Z(\beta, \lambda)$ to logarithmically many terms. This gives the following theorem.

\begin{theorem}[Quasi-polynomial classical algorithms for local thermal SYK expectations, informal]\label{thm:main2}
Let $O = \sum_{a=1}^M c_a \psi_{A_a}$ be an $L$-local observable, i.e., $A_a \in \binom{[n]}{\ell}$ for $\ell \leq L$. If the observable is local and bounded-norm on each site, i.e.,
\begin{align}
	L=O(1), \quad \max_{x \in [n]} \sum_{a:x\in A_a} |c_a| = O(1),
\end{align}
then $\Tr(O \rho_\beta)$ can be classically estimated to additive error $\epsilon$ with cost $n^{O(\log(n/\epsilon))}$ for any
\begin{align}
    0 \leq \beta \leq \sqrt{\frac{2^q}{q \max\{q,L\}}}\frac{9}{20e^{1/4}}\lr{\sqrt{1+\frac{4}{e}}-1}.
\end{align}
\end{theorem}

We note that this classical algorithm bypasses the obstructions mentioned in Remark~\ref{rem:obstruct}: it maintains quasi-polynomial cost for choices of $\beta = \Theta(1), \epsilon = \Theta(1/\poly n)$ where the Gibbs state has a polynomial circuit lower bound, Gaussian state ansatzes fail to describe typical eigenstates, and instance-to-instance fluctuations can be resolved.

We briefly mention related results in the literature. The SYK model was originally proposed as a model with no phase transition at constant temperature due to replica symmetry~\cite{sachdev1993gapless,kitaev2015SYK}. Proving such properties rigorously is more challenging. In comparison, for the classical Sherrington-Kirkpatrick model, the replica symmetric phase has been rigorously known since \cite{aizenman1987some}. To obtain an algorithm that achieves arbitrarily small error $\epsilon$ in estimating the partition function, \cite{bencs2025zeros} extended the cluster expansion of~\cite{aizenman1987some}. In some sense, our work represents a quantum version of~\cite{aizenman1987some,bencs2025zeros} achieved by developing a cluster expansion amenable to disordered mean-field Hamiltonians with highly noncommuting terms. Our work is also closely related to the non-rigorous result for the SYK model in the recent work of~\cite{zlokapa2026syk}, which gives a replica computation that suggests that Barvinok's algorithm succeeds for all constant $\beta$. Our results can thus be viewed as a rigorous proof of this physics conjecture for a range of constant temperatures.

\section{Techniques}
The distinctly quantum properties of the SYK model pose additional challenges compared to the classical setting. We comment on obstructions to obtaining these results via prior methods, and then we sketch our technical contributions.

\subsection{Prior work}\label{sec:prior}
As discussed below, very few rigorous results exist for the SYK model due to the difficulty of applying standard techniques in random matrix theory to obtain sufficiently tight results.\footnote{Indeed, to the best of our knowledge, the most efficient rigorously known classical algorithm for estimating SYK thermal expectations is exact diagonalization.} Consequently, we also review relevant non-rigorous approaches used in the physics literature. In particular, we highlight that the very recent work~\cite{zlokapa2026syk} provides non-rigorous evidence based on the replica trick, path integration, and saddle point analyses in small-$\beta$ and large-$\beta$ limits to suggest that Theorem~\ref{thm:main2} of this work extends to all constant $\beta$ (in the large-$q$ limit). It is thus especially important to understand the obstacles encountered in making the results of~\cite{zlokapa2026syk} fully rigorous.

\textbf{Replica trick.} The original solution of classical spin mean-field glasses was obtained via the \emph{replica trick}, a powerful but non-rigorous method that correctly computes the quenched free energy of many disordered systems. Making such computations rigorous is difficult even for classical spin glasses; a series of breakthroughs achieved this for many classical models, including the Sherrington-Kirkpatrick model~\cite{guerra2002thermodynamic,guerra2003broken,talagrand2006free,talagrand2006parisi,chen2013aizenman,panchenko2013sherrington,panchenko2014parisi}. For classical spin glasses, the replica method also gives a fairly straightforward indication of a zero-free region throughout the replica symmetric phase~\cite{takahashi2011replica,takahashi2013zeros}. Obtaining this result rigorously, however, required the technically impressive work of~\cite{bencs2025zeros}. As the standard physics solutions of the SYK model and other quantum spin glasses also employ the replica trick, they also suffer from the classical obstructions to rigor in this regard. We note, however, that due to the points discussed below, even the non-rigorous replica computation establishing a zero-free region for SYK is significantly more complicated than for classical models~\cite{zlokapa2026syk}. Additionally, computing the quenched free energy rigorously is insufficient to obtain an algorithm that estimates observables to arbitrary additive error and succeeds with high probability.

\textbf{Path integration / quantum Monte Carlo.} Even before applying the replica trick, more basic complications arise in the non-rigorous physics computation of the SYK model. To describe these issues, we first recall the computation of a classical partition function by summing over all eigenstates / spin configurations. As a minimal example, consider computing the annealed partition function of the $q$-spin model $H_\mathrm{cl} = n^{-(q-1)/2}\sum_{i_1, \dots, i_q=1}^n A_{i_1\cdots i_q} Z_{i_1} \cdots Z_{i_q}$ with $A_{i_1\cdots i_q} \sim \cN(0, 1)$. The Gaussian moment generating function (MGF) immediately gives
\begin{align}\label{eq:gmgf}
	\E \, Z_\mathrm{cl}(\beta) = \E \sum_{\sigma \in \{\pm1\}^n} e^{-\beta H_\mathrm{cl}(\sigma)} = 2^n \exp[\frac{n\beta^2}{2}] \neq 0.
\end{align}
This is clearly nonzero for any (possibly complex) $\beta$, showing zero-freeness of the annealed partition function. In comparison, a physics computation shows that this quantity can vanish for the SYK model~\cite{khramtsov2021spectral,bunin2024fisher,zlokapa2026syk}. The analysis similarly applies the Gaussian MGF but to paths in a path integral. Since the eigenbasis is unknown, one integrates over an overcomplete basis of spin coherent states~\cite{sachdev2023quantum}. Each path is parameterized by $\tau \in [0,\beta]$, and the Gaussian MGF is applied to exponentiated scalar path values. However, due to the sign problem, this path integral is not over a probabilistic measure; in contrast, the partition function of a stoquastic Hamiltonian can be rigorously treated as a continuous-time Markov process via quantum Monte Carlo / Poisson Feynman-Kac~\cite{bravyi2006complexity,crawford2007thermodynamics,leschke2021free}. Although one may hope that the saddle point method can somehow be applied despite these issues, an additional difficulty arises: one must first take a vanishingly small imaginary timestep size $\Delta \tau \to 0$ before taking $n \to \infty$. Unless these limits are commuted, this introduces an extensive number of integrals, preventing the rigorous application of the saddle point method.

\textbf{Plefka expansion and quantum TAP equations.} A non-rigorous physics argument truncates a high-temperature expansion of the free energy at second order to derive the Thouless-Anderson-Palmer (TAP) equations, which recover local thermal expectations~\cite{thouless1977solution,plefka1982convergence}. Although this was proven to obtain the correct results for classical models via alternative rigorous methods~\cite{talagrand2003spin,chatterjee2010spin,bolthausen2014iterative,auffinger2019thouless}, this second-order truncation remains non-rigorous. For the Sherrington-Kirkpatrick model, truncating the small-$\beta$ expansion at order $O(\log n)$ was formally justified by the zero-free region established in \cite{bencs2025zeros}. In the quantum setting, an analogous Plefka expansion has been proposed to derive quantum TAP equations~\cite{biroli2000quantum}, but besides the lack of formal justification via a zero-free region, the quantum setting also uses the non-rigorous path integral discussed above.

\textbf{Random matrix theory.} One may hope that standard tools from random matrix theory may provide an alternative to the quantum path integral. However, the resulting bounds are prohibitively loose due to a fundamental failure of the moment method. Since the SYK is normalized to have extensive energy $\norm{H} = \Theta(n)$, the Taylor expansion of $e^{-\beta H}$ is dominated by terms at order $n$. As extensively discussed in~\cite{baldwin2020quenched}, this prevents naive $1/n$ expansions from converging. The best known bounds on quantities such as $\E \, \Tr H^k$ for $k = \Theta(n)$ lead to loose bounds of basic quantities such as the spectral norm~\cite{feng2019spectrum}. This issue also manifests itself when attempting to compute the annealed free energy as in \eqref{eq:gmgf}, since noncommutative corrections to the Gaussian MGF cannot be perturbatively treated in a $1/n$ expansion. Indeed, it is non-rigorously known for the SYK model that these corrections cause the annealed partition function to vanish at constant $\beta$: a physics computation shows that the spectral form factor $|Z(it)|^2$ for $t > 0$ has zeros at constant times $t$ separated by $O(1/n)$~\cite{khramtsov2021spectral}.

\textbf{Commutation index.} Outside of the more standard techniques available in physics and random matrix theory, recent work~\cite{anschuetz2025strongly} directly exploited the noncommutativity of the SYK model to show \eqref{eq:selfavg} and Remark~\ref{rem:obstruct}. The main technique is to analyze the \emph{commutation index} defined as~\cite{king2025triply}
\begin{align}\label{eq:commind}
	\Delta = \sup_{\ket\phi} \binom{n}{q}^{-1} \sum_{I \in \binom{[n]}{q}} \bra{\phi}\psi_I\ket{\phi}^2,
\end{align}
where the supremum is evaluated over all states $\ket{\phi}$. The insight of~\cite{anschuetz2025strongly} is to bound $\Delta$ in terms of the Lov\'asz theta function of a certain Johnson association scheme~\cite{delsarte1973algebraic}, giving $\Delta = \Theta(n^{-q/2})$. The concentration of the quenched and annealed free energies is then shown via a Lipschitz argument. Specifically, for real $\beta$, the commutation index appears by evaluating the gradient of the free energy with respect to the disorder, i.e.,
\begin{align}
    \norm{\nabla_J \log Z(\beta)}_2^2 = \beta^2 \sum_{I \in \binom{[n]}{q}} \Tr(\rho_\beta \psi_I)^2 \leq \binom{n}{q} \beta^2 \Delta.
\end{align}
This gradient bound suffices to show \eqref{eq:selfavg} via the Gaussian concentration of Lipschitz functions.
For complex $\beta$, one instead has
\begin{align}
    \norm{\nabla_J \log Z(\beta)}_2^2 = |\beta|^2 \sum_{I \in \binom{[n]}{q}} \frac{\abs{\Tr(e^{-\beta H}\psi_I)}^2}{\abs{\Tr(e^{-\beta H})}^2},
\end{align}
which can only be applied in a region already known to be zero-free. We thus require additional ideas to show zero-freeness and the absence of a phase transition.

\textbf{Cluster expansions.} The use of high-temperature expansions of $e^{-\beta H}$ into sums of products of local terms in $H$ has a long history~\cite{kotecky1986cluster,dobrushin1996estimates,park1982cluster,greenberg1969correlation}. To the extent of our knowledge, all rigorous proofs of zero-freeness of quantum Hamiltonians rely on cluster expansions whose polymers are supported on the union of the supports of all the Hamiltonian terms in such a product~\cite{harrow2020classical,mann2021efficient,yao2022polynomial,mann2024algorithmic,zlokapa2026syk}. This gives a universal zero-free disk of constant radius $|\beta| = O(1/D)$ for all degree-$D$ constant-local quantum Hamiltonians; for mean-field models such as SYK, however, this excludes constant temperature. In the classical literature, the influential work of~\cite{aizenman1987some} introduced a proof technique for mean-field cluster expansions, which was used by \cite{bencs2025zeros} to show zero-freeness of the Sherrington-Kirkpatrick model. However, as described in Sec.~\ref{sec:overview}, the cycle-counting approach of~\cite{aizenman1987some} generalizes poorly to noncommuting Hamiltonians. We thus give a different cluster expansion whose polymers are based on Wick pairs obtained by disorder-averaged moments.

\textbf{Barvinok's algorithm.} Barvinok's interpolation method~\cite{barvinok2014computing,barvinok2016computing,barvinok2016approximating,barvinok2016combinatorics,barvinok2018approximating} has been applied to classical systems such as the hardcore gas model and Ising model~\cite{peters2019conjecture,bezakova2018inapproximability,liu2019ising,liu2025correlation,guo2020zeros,bencs2021zero}, quantum many-body systems~\cite{harrow2020classical,mann2021efficient,yao2022polynomial,mann2024algorithmic}, and other contexts with a sign problem~\cite{eldar2018approximating,jiang2025positive}. In this work, we use it in a fairly simple manner since we show a zero-free disk, which is straightforward to analyze; the main technical difficulty lies in showing the zero-free property.

\subsection{Proof overview}\label{sec:overview}

Our main technical contribution, as well as the bulk of the proofs of Theorems~\ref{thm:main1} and~\ref{thm:main2}, is a cluster expansion for mean-field noncommuting models. More precisely, we actually require \emph{two} cluster expansions.
\begin{itemize}
    \item \emph{Annealed free energy.} To avoid computing the SYK annealed free energy, we use a cluster expansion to show that $\E \, Z$ admits a zero-free disk. The polymers correspond to Wick pairs obtained from Gaussian expectations. When the Hamiltonian is perturbed by an observable (as needed for Theorem~\ref{thm:main2}), the polymers are extended to pass through the deterministic support of the observable. To show convergence of the expansion of $\log \E\, Z$, we apply the Koteck\'y-Preiss condition with an activity weight proportional to the polymer support. This shows that for some constant $R > 0$,
    \begin{align}
        \inf_{|\beta| \leq R} |\E \, Z(\beta)| > 0
    \end{align}
    and $\log \E\, Z(\beta)$ admits a convergent cluster expansion (Theorem~\ref{thm:anneal}).
    \item \emph{Second moment.} To show that an SYK instance has a zero-free disk with high probability, it suffices to show that $\E \,|Z(\beta)|^2$ concentrates to $|\E \, Z(\beta)|^2$. We construct a cluster expansion similar to that of the annealed computation, where the only new contribution comes from ``mixed'' Wick pairs that have endpoints in both copies of the system. The contribution of these pairs can be heuristically bounded by the commutation index \eqref{eq:commind}, which exploits the noncommutativity of Hamiltonian terms. Most Wick pairs involve a repeated Hamiltonian term $\psi_I$; summing over all such possible terms contributes
    \begin{align}
        \sum_{I \in \binom{[n]}{q}} \abs{\tr(\psi_I O_1) \tr(\psi_I O_2)} \leq 2 \binom{n}{q} \Delta \tr|O_1| \tr |O_2|.
    \end{align}
    Alternatively, one cause the orthogonality of Majorana strings to give a more direct proof that more readily generalizes to less noncommutnig models; we use this strategy in Theorem~\ref{thm:sa} to show for some constant $0 < R' < R$ that
    \begin{align}\label{eq:mainbounds}
        \sup_{|\beta| \leq R'} \frac{\E \left|Z(\beta) - \E Z(\beta)\right|^2}{\left|\E Z(\beta)\right|^2} = O(n^{1-q/2}).
    \end{align}
\end{itemize}
We describe both of these cluster expansions in more detail in Secs.~\ref{sec:cluster1} and~\ref{sec:cluster2}, where we also provide a motivating classical example. For now, we complete the sketch of the proofs of the main theorems.

To complete the proof of Theorem~\ref{thm:main1}, we introduce slightly smaller constants $0 < r' < r < R'$. For $|\beta| \leq r'$, Cauchy's integral formula gives (see Corollary~\ref{cor:zf})
\begin{align}
    \abs{\frac{Z(\beta) - \E Z(\beta)}{\E Z(\beta)}} = \abs{\frac{1}{2\pi i} \oint_{|w|=r} \frac{1}{w-\beta}\cdot\frac{Z(w) - \E Z(w)}{\E Z(w)}} dw \leq \frac{r}{r-r'} \cdot \frac{1}{2\pi} \int_0^{2\pi} \abs{\frac{Z(re^{i\theta}) - \E Z(re^{i\theta})}{\E Z(re^{i\theta})}} d\theta,
\end{align}
since the annealed partition function is known to be zero-free. Squaring the LHS, we apply Jensen's inequality to the RHS to obtain (for all $|\beta| \leq r'$)
\begin{align}
    \frac{\left|Z(\beta) - \E Z(\beta)\right|^2}{\left|\E Z(\beta)\right|^2} \leq \lr{\frac{r}{r-r'}}^2 \sup_{|w|=r} \frac{\left|Z(w) - \E Z(w)\right|^2}{\left|\E Z(w)\right|^2}.
\end{align}
Finally, we take expectations and apply Markov's inequality to conclude by \eqref{eq:mainbounds} that for all $|\beta| \leq r'$, the partition function is zero-free with probability $1 - O(n^{1-q/2})$.

To show Theorem~\ref{thm:main2}, we extend the above proofs to the perturbed Hamiltonian $H \to H + \lambda O$ for sufficiently small $\lambda$. The thermal expectation is then evaluated by estimating the corresponding partition function $Z(\beta, \lambda)$ at and near $\lambda = 0$ to approximate
\begin{align}
    \Tr(O\rho_\beta) = -\frac{1}{\beta} \partial_\lambda\Big|_{\lambda=0} \log Z(\beta, \lambda)
\end{align}
via a finite difference. The error associated with this procedure is controlled via a standard argument (see, e.g., \cite{bravyi2006complexity}), allowing us to truncate the cluster expansions at order $k = O(\log n/\epsilon)$. Finally, to compute the coefficients (cumulants) in the expansion of $\log Z(\beta, \lambda)$, it suffices to compute the moments $\Tr(H^k)$ similarly to prior applications of Barvinok's method~\cite{harrow2020classical}.

\subsubsection{Annealed cluster expansion}\label{sec:cluster1}
To concretely describe our cluster expansion, we recall a rough statement of Koteck\'y-Preiss here, and we refer the reader to Theorem~\ref{thm:kp} for a formal presentation. Given polymer set $\cP$ and compatibility relation $\sim$ (such that all polymers $s_i \in \cP$ satisfy $s_i \nsim s_i$), the polymer partition function
\begin{align}\label{eq:ximain}
    \Xi(z) = 1 + \sum_{r \geq 1} \frac{1}{r!} \sum_{\substack{(s_1, \dots, s_r) \\ s_i \sim s_j}} \prod_{j=1}^r z_{s_j}.
\end{align}
satisfies the Koteck\'y-Preiss condition if there are functions $a, \rho:\cP\to[0,\infty)$ such that every $s_i \in \cP$ satisfies
\begin{align}\label{eq:kpmain}
    \sum_{s_j \,:\, s_j \nsim s_i} \rho(s_j) e^{a(s_j)} \leq a(s_i).
\end{align}
This implies that the cluster expansion of $\log \Xi$ converges absolutely.

Of the prior cluster expansions to show zero-freeness for quantum Hamiltonians~\cite{harrow2020classical,mann2021efficient,yao2022polynomial,mann2024algorithmic,zlokapa2026syk}, the most recent~\cite{zlokapa2026syk} gives the tightest bounds on a zero-free disk for bounded-degree Hamiltonians via the Koteck\'y-Preiss condition. These arguments are difficult to generalize to the mean-field setting, because they define polymers whose support consists of the union of the supports of the Hamiltonian terms involved in the polymer. While this is useful for a universal high-temperature phase independent of the Hamiltonian coefficients, it only gives a vanishingly small zero-free region $|\beta| = o(1)$. In contrast, our cluster expansion explicitly uses the Gaussian disorder to construct polymers with small support despite all-to-all connectivity.

\cite{aizenman1987some} developed a mean-field cluster expansion for the Sherrington-Kirkpatrick (SK) model
\begin{align}
    H_\mathrm{SK} = \sum_{I\in\binom{[n]}{2}} A_I \sigma_I, \quad A_I \sim \cN(0, 1/n),
\end{align}
where $\sigma_I$ denotes the product $\sigma_i \sigma_j$ for $I = \{i,j\}$, and $\sigma \in \{\pm1\}^n$ is a spin configuration.
The expansion, also used in~\cite{bencs2025zeros}, constructs polymers from cycles. As a rough sketch of this idea, we rewrite the normalized SK partition function as
\begin{align}\label{eq:skcycle}
    \tr e^{-\beta H_\mathrm{SK}} = \lr{\prod_{I\in\binom{[n]}{2}} \cosh(\beta A_I)} \sum_{E \subseteq \binom{[n]}{2}} \lr{\prod_{I \in E} \tanh (-\beta A_I)} \cdot 2^{-n} \sum_\sigma \prod_{\{i,j\}\in E} \sigma_i \sigma_j
\end{align}
Only terms corresponding to even graphs contribute to the partition function; i.e., $\deg_E(v)$ is even for all $v \in [n]$. Following this observation, we define a polymer $s$ to be a connected even edge set in the complete graph on $[n]$. Concretely, define a polymer $s_i$ to be a finite connected graph $(V(s_i), E(s_i))$ with $E(s_i) \subseteq \binom{[n]}{2}$, and every vertex in $V(s_i)$ has even degree in $E(s_i)$. The support of $s_i$ is given by $V(s_i)$, and the compatibility relation $s_i \sim s_j$ is $V(s_i) \cap V(s_j) = \emptyset$. In the form of \eqref{eq:ximain}, the polymer representation is thus
\begin{align}
    \tr e^{-\beta H_\mathrm{SK}} = \lr{\prod_{I\in\binom{[n]}{2}} \cosh(\beta A_I)} \sum_{r \geq 0} \frac{1}{r!} \sum_{\substack{(s_1, \dots, s_r) \\ s_i \sim s_j}} \prod_{u=1}^r w_{s_u}, \quad w(s_u) = \prod_{I \in E(s_u)} \tanh (-\beta A_I).
\end{align}
Because a connected graph with $v$ vertices and $e$ edges has family size $O(n^v)$ but weight $\EE{w(s)}^2 = O(n^{-e})$, the dominant contributions are unions of simple cycles each with $e=v$.

This argument becomes noticeably more difficult to apply for noncommuting models. While \eqref{eq:skcycle} made it evident that only cycles contribute to the partition function, one cannot directly obtain such an expression for the SYK model. We instead expand the partition function into moments and take an expectation over disorder. For simplicity, we describe the method here for the annealed partition function of the unperturbed SK model. In the Taylor expansion of $\E \, \tr e^{-\beta H_\mathrm{SK}}$, we evaluate the even moments (since odd moments vanish) with Wick's rule over the set $\Pi(2k)$ of pairings of $[2k]$:
\begin{align}
    \E\, \tr(H_\mathrm{SK}^{2k}) = n^{-k} \sum_{\pi \in \Pi(2k)} \sum_{\substack{I_1, \dots, I_{2k} \in \binom{[n]}{2}\\ I_a = I_b \text{ for } \{a,b\} \in \pi}} \tr(Z_{I_1} \cdots Z_{I_{2k}}).
\end{align}
A polymer is now defined on the tuple $\tau = (\pi, I_1, \dots, I_{2k})$, where the sets are consistent with the pairing ($I_a = I_b$ for $\{a,b\}\in\pi$). Let $S_{2k}$ be the set of all such tuples; we can then rewrite the annealed partition function as
\begin{align}
    \E \, \tr e^{-\beta H_\mathrm{SK}} = \sum_{k=0}^\infty \frac{(-\beta)^{2k}}{(2k)!} \sum_{\tau \in S_{2k}} n^{-k} \tr(Z_{I_1}\cdots Z_{I_{2k}}).
\end{align}
To obtain a polymer representation, we isolate connected components. Define the graph $G(\tau)$ on vertex set $[2k]$ that places an edge between $u,v \in [2k]$ if and only if $I_u \cap I_v \neq \emptyset$. Let $S_{2k}^\conn$ denote the subset of $\tau \in S_{2k}$ such that $G(\tau)$ is connected. Then the polymer set $\cP$ and the support of $s_i = (\pi,I_1,\dots,I_{2k}) \in \cP$ are 
\begin{align}
    \cP = \bigcup_{k \geq 1} S_{2k}^\conn, \quad \supp(s_i) = \bigcup_{j=1}^{2k} I_j.
\end{align}
Polymers are compatible $s_i \sim s_j$ if $\supp(s_i) \cap \supp(s_j) = \emptyset$. By factorizing $\tr(Z_{I_1} \cdots Z_{I_{2k}})$ into connected components, we can rewrite the annealed partition function in the form of \eqref{eq:ximain} as
\begin{align}\label{eq:skpoly}
    \E \, \tr e^{-\beta H_\mathrm{SK}} = \sum_{r=0}^\infty \frac{1}{r!} \sum_{\substack{(s_1, \dots, s_r) \\ s_i \sim s_j}} \prod_{u=1}^r z_{s_u}, \quad z_{s_u} = \frac{(-\beta/\sqrt n)^{|s_u|}}{|s_u|!} \tr \prod_{j=1}^{|s_u|} Z_{I_j^{(u)}}
\end{align}
where we use notation $s_u = (\pi_u, I_1^{(u)}, \dots, I_{|s_u|}^{(u)})$. To show convergence, we then check the Koteck\'y-Preiss condition \eqref{eq:kpmain} with $a(s_u) = |\supp(s_u)|/4$ and $\rho(s_u) = |z_{s_u}|$.

To adapt the above computation to the perturbed SYK model (Sec.~\ref{sec:first}), a few modifications must be made. For Hamiltonian $H = H_\mathrm{SYK} + \lambda O$, the polymers must now track
\begin{itemize}
    \item an indicator function $\eta : [m] \to \{0,1\}$ specifying in $\E \, \tr H^m$ if the $j$th slot belongs to an SYK Hamiltonian term or an observable term,
    \item a pairing $\pi$ of the $H_\mathrm{SYK}$ terms arising from Wick's theorem, and
    \item the supports $I_1, \dots, I_{2k}$ of the SYK terms and the supports of the deterministic observable insertions.
\end{itemize}
We form an intersection graph similarly to before and group terms into connected components. Since even Majorana strings supported on disjoint sets commute, these connected components factorize similarly to \eqref{eq:skpoly}, producing the desired polymer representation and allowing us to check the Koteck\'y-Preiss condition. This proves that for a disk of radius $|\beta| \cJ \lesssim 0.24\sqrt{q/\max\{q,L\}}$ for $\cJ=\sqrt{q/2^{q-1}}$, the annealed partition function is nonzero (Theorem~\ref{thm:anneal}). This is expected to be tight up to a constant independent of $q$ due to the non-rigorous results of~\cite{zlokapa2026syk}.

\subsubsection{Second moment cluster expansion}\label{sec:cluster2}

We use a cluster expansion similar to the annealed computation above to show \eqref{eq:mainbounds}; i.e., for all sufficiently small $|\beta|$,
\begin{align}\label{eq:conc}
    \frac{\E \left|Z(\beta) - \E Z(\beta)\right|^2}{\left|\E Z(\beta)\right|^2} = O(n^{1-q/2}).
\end{align}
Note that the disconnected product $\E Z(\beta)\E Z(\overline\beta)$ only involves Wick pairs within replicas. Hence, the ratio
\begin{align}
    \frac{\E\big[Z(\beta) Z(\overline\beta)\big]}{\E Z(\beta)\E Z(\overline\beta)}
\end{align}
is controlled by polymers containing at least one mixed Wick pair across two replicas. We will use $K_e$ to denote the support of a mixed Wick pair; more carefully, $K_e = I_u = I_v$ is a set in $\binom{[n]}{q}$, where $I_u, I_v$ correspond to a Wick pair $\{u,v\}$ that is split across the two replicas of the system.

We sketch our upper bound on the contribution from such mixed Wick pairs. To create a connected graph from $k$ Gaussian pairs and supported on $K_e$, we must add a set that intersects existing sites in the polymer. Once $K_e$ is fixed, there are roughly $O(n^{q-1})$ ways to choose a new set such that the component remains connected. (In the formal proof, we have to be more careful and track the $k$ dependence of this quantity, but this rough counting morally holds due to a $1/k!$ prefactor in the cluster expansion.) Each Wick pair carries a factor of $\sigma^2 = O(n^{-(q-1)})$, so each additional Wick pair roughly contributes $O(n^{q-1}\sigma^2) = O(1)$. The only nontrivial factor thus comes from choosing $K_e$: there are $O(n^q)$ ways to choose $K_e$, so this choice contributes $O(n^q \sigma^2) = O(n)$. In total, this counting argument says that the mixed Wick pairs contribute $O(n)$.

To show that the contribution from mixed Wick pairs is actually vanishingly small, we need to refine the above argument. We will split these mixed polymers into two classes. A polymer is a \emph{distinct-set polymer} if it contains a mixed Wick pair whose support $K_e$ appears nowhere else in the polymer, neither as another Gaussian support nor as one of the deterministic supports $A_{a_j}$; otherwise it is a \emph{repeated-set polymer}.

The repeated-set polymers are rarer, because every mixed-pair support must be reused elsewhere in the polymer. In particular, we have to constrain one of the additional Wick pairs to also be supported on $K_e$; this reduces a $O(n^{q-1}\sigma^2)$ contribution in the above argument to $O(\sigma^2)$. Hence, repeated-set polymers only contribute $O(n \cdot n^{-(q-1)}) = O(n^{2-q})$ (Lemma~\ref{lem:rep}).

We now count the contributions from distinct-set polymers. Cyclicity of trace lets us isolate the distinguished mixed pair as $\psi_{K_e}$, so
\begin{align}\label{eq:o1o2}
    \tr(W_1(P)) \tr(W_2(P)) = \tr(\psi_{K_e} O_1)\tr(\psi_{K_e} O_2),
\end{align}
where $W_1, W_2$ are formed by a product of Hamiltonian terms and observable terms, and $O_1$ and $O_2$ depend only on the remaining supports in the two replicas. There are two possible ways in which one could use this observation. One possibility, useful for the SYK model, is to relate this to the commutation index via Lemma~\ref{lem:comm2}: if one rearranges sums to ensure that $O_1, O_2$ are independent of the isolated mixed pair, then
\begin{align}
    \sum_{K_e \in \binom{[n]}{q}} \abs{\tr(\psi_{K_e} O_1) \tr(\psi_{K_e} O_2)} \leq \left[\sum_{K_e} \abs{\tr(\psi_{K_e} O_1)}^2\right]^{1/2} \left[\sum_{K_e} \abs{\tr(\psi_{K_e} O_2)}^2\right]^{1/2} \hspace{-1em} \leq 2 \binom{n}{q} \Delta \tr|O_1| \tr |O_2|.
\end{align}
Since $\Delta = O(n^{-q/2})$ for Majorana fermions (Theorem~\ref{thm:comm}), the contribution from $K_e$ is reduced from $O(n^q \sigma^2)$ to $O(n^q \Delta \sigma^2) = O(n^{1-q/2})$. Since the additional pairs in the polymer contribute $O(n^{q-1} \sigma^2) = O(1)$ as before, the total distinct-set contribution is $O(n^{1-q/2})$.

To give a proof strategy that we expect to apply beyond the SYK model (e.g., for classical models with constant $\Delta$), we use a different approach that uses orthogonality of terms instead of the commutation index. Observe that each factor $\tr(\psi_{K_e} O_b)$ vanishes unless $O_b$ is only supported on $K_e$. Suppose we delete $K_e$ from the support intersection graph of $W_1$; since $W_1$ was a single connected component, each of the remaining connected components must disjointly intersect $K_e$ and in fact contribute an even, nonempty subset of $K_e$. Since $|K_e| = q$, this means that the support intersection graph of $O_1$ is decomposed into at most $q/2$ connected components. If the $j$th component has $g_j$ Gaussian supports, then it contributes $O(n \cdot n^{(q-1)g_j} \cdot \sigma^{2g_j}) = O(n)$ to \eqref{eq:o1o2}, since we can write the connected component as a spanning tree with $n$ choices for the root. Hence, these connected components contribute at most $O(n^{q/2})$ and specify a unique $K_e$; multiplying by the variance $\sigma^2$ associated with $K_e$ yields a total contribution of $O(n^{q/2}\sigma^2) = O(n^{1-q/2})$ from the distinct-set pairs (Lemma~\ref{lem:dist}).

Once the contributions from mixed Wick pairs are bounded, we complete the proof by applying Koteck\'y-Preiss. The resulting bounds on the mixed and separated polymer partition functions yield \eqref{eq:conc} (Theorem~\ref{thm:sa}).

\subsection{Discussion and open questions}
Unlike most other systems studied in condensed matter and quantum chemistry, the SYK model has a well-defined asymptotic limit. It has thus been extensively studied in rigorous contexts as a model with distinctly quantum properties and a potential for provable quantum advantage~\cite{hastings2022optimizing,hastings2023field,anschuetz2025strongly,ramkumar2025high}. By showing a classical algorithm for estimating SYK thermal expectations, our work highlights the need for stronger evidence of classical hardness than heuristics such as polynomial quantum circuit complexity or the sign problem. A core theoretical obstruction to providing such evidence is that many ideas developed in the theory of classical average-case complexity~\cite{achlioptas2006solution,gamarnik2021overlap,hopkins2017bayesian,barak2019nearly,hopkins2018statistical,el2022sampling,el2025sampling} also extend to average-case hardness for quantum algorithms~\cite{basso2022performance,anschuetz2025efficient,zlokapa2025average}, in part due to their tight relationship to glassy physics~\cite{bandeira2022franz}. Hence, in the absence of significant structure, disordered instances often look average-case hard both classically and quantumly.

Our main technical contribution---a new cluster expansion for disordered quantum systems, including mean-field models---may aid in understanding the relationship between zero-freeness and other phenomena. We emphasize that the underlying methods apply to both commuting and non-commuting models, and thus may be of independent interest for classical and quantum applications beyond SYK. In the quantum setting, we suggest three future directions of particular interest.
\begin{itemize}
    \item \emph{Entanglement transitions}. \cite{bakshi2024high} showed the death of entanglement in bounded-degree models via a cluster expansion. Our cluster expansion may enable similar arguments for mean-field quantum systems to better understand the relationship between phase transitions in entanglement and other properties of the Gibbs state.
    \item \emph{Correlation decay}. It is known that zero-freeness implies point-to-point correlation decay in bounded-degree quantum systems~\cite{harrow2020classical}. It may be possible to extend such notions to quantum analogues of strong spatial mixing~\cite{brandao2019finite}, as in the classical setting~\cite{dobrushin1985completely,shao2021contraction,regts2023absence,liu2025correlation}. Even in the bounded-degree case, our method may be useful for analyzing disordered quantum systems.
    \item \emph{Mixing times}. The algorithmic question of mixing time bounds is closely related to stronger notions of correlation decay. Quantum algorithms for Gibbs sampling are largely expected to be efficient in a high-temperature, rapidly mixing phase due to their resemblance to physical system-bath interactions~\cite{chen2025efficient}. However, Barvinok's algorithm generically gives quasipolynomial-time algorithms in this high-temperature phase of any model due to the absence of Fisher zeros~\cite{fisher1965nature,zlokapa2026syk}. An interesting open question is to formally relate zero-freeness with mixing times of quantum Gibbs samplers.
\end{itemize}

We also comment on possible extensions of our work within the scope of SYK. While we only show results at sufficiently large constant temperature for the SYK model, non-rigorous arguments from physics suggest that Barvinok's method can be applied at \emph{any} constant temperature~\cite{zlokapa2026syk} because the SYK model is believed to have no phase transition at constant temperature~\cite{sachdev1993gapless,kitaev2015talks}. Showing this via our methods is likely to require additional technical insight: our use of Koteck\'y-Preiss is limited to showing a zero-free disk, while the SYK partition function appears to instead have a zero-free wedge~\cite{zlokapa2026syk} due to the complex zeros appearing on the imaginary $\beta$ axis from the spectral form factor~\cite{khramtsov2021spectral}. Another open question is to extend our work to obtain a fully polynomial time approximation scheme (FPTAS) for the SYK model. In other settings, a more refined analysis of Barvinok's method has yielded polynomial-time algorithms for both classical and quantum models~\cite{patel2017deterministic,yao2022polynomial,mann2024algorithmic}. Finally, we emphasize that our results only concern local thermal expectations; they do not address other tasks such as sampling measurement outcomes~\cite{bergamaschi2024quantum} or estimating time-evolved observables~\cite{larkin1969quasiclassical,swingle2016measuring,roberts2017chaos,google2025observation}, which remain interesting possible avenues for quantum advantage for the SYK model and other quantum many-body systems.

\paragraph{Acknowledgments.} AZ thanks Thiago Bergamaschi, David Gamarnik, Aram Harrow, Bobak Kiani, Kuikui Liu, and Francisco Pernice for discussions and comments, and the Hertz Foundation for support. ChatGPT Pro contributed (i) routine combinatorial calculations and bounds, which were checked by hand; (ii) Lemma~\ref{lem:comm2}, which motivates the second moment argument in the introductory material but is not formally used (in favor of a more general argument that applies to commuting Hamiltonians as well); and (iii) the connected component decomposition used to rearrange sums in Lemma~\ref{lem:dist}. The author verified the correctness and originality of all content including references.

\section{Preliminaries}\label{sec:prelim}
We restate our setup and main results to more clearly establish notation. We take the SYK Hamiltonian on $n$ fermions normalized as $\{\psi_i, \psi_j\}=\delta_{ij}$. A Majorana string on a $I \in \binom{[n]}{q}$ is indicated by
\begin{align}
    \psi_I = i^{q/2}\psi_{i_1}\cdots \psi_{i_q}.
\end{align}
We write the normalized trace as $\tr(\cdot) = \frac{1}{d}\Tr(\cdot)$ for Hilbert space dimension $d=2^{n/2}$. The SYK Hamiltonian is normalized as
\begin{align}
    H_0 = \sum_{I \in \binom{[n]}{q}} J_I \psi_I, \quad J_I \sim \cN(0, \sigma^2), \quad \sigma^2 = \frac{(q-1)!}{n^{q-1}}.
\end{align}
We will always assume $q \geq 4$ is an even constant and $n$ is even. In standard SYK notation, $\cJ = \sqrt{q/2^{q-1}}$. To compute local observable expectations, we will analyze the perturbed model
\begin{align}
    H = H_0 + \lambda O, \quad O = \sum_{a=1}^M c_a \psi_{A_a}
\end{align}
where each $A_a \subset [n]$ has even cardinality. We define
\begin{align}
    \Gamma = \max_{x \in [n]} \sum_{a:x\in A_a} |c_a|, \quad L = \max_{a \in [M]}|A_a|.
\end{align}
We will give quasipolynomial algorithms that estimate thermal expectations at constant temperature and any inverse polynomial error for $\Gamma, L = O(1)$.
We denote the partition function and normalized partition function of this perturbed Hamiltonian by $Z$ and $\wh Z$ respectively.

Our main claim is that with high probability the SYK model has a zero-free disk of constant radius. This proceeds from two results about the SYK model: the zero-freeness of its annealed free energy, and a second-moment bound that holds for complex $\beta$. These results combine with Cauchy's integral formula and Markov's inequality to imply zero-freeness of a typical SYK instance.

\begin{theorem}[Annealed zero-freeness]\label{thm:anneal}
    For any $|\lambda|\Gamma \leq 2^{-q/2} \max\{q,L\}^{-1/2}$, on the disk
    \begin{align}
        |\beta| \cJ \leq \sqrt{2}\lr{\sqrt{1+1/e}-1} \sqrt{\frac{q}{\max\{q,L\}}} \approx 0.24 \sqrt{\frac{q}{\max\{q,L\}}},
    \end{align}
    $\EE{\wh Z(\beta)} \neq 0$ and the annealed free energy $\log \EE{\wh Z(\beta)}$ admits an absolutely convergent cluster expansion.
\end{theorem}

\begin{theorem}[Concentration of partition function]\label{thm:sa}
    Let
    \begin{align}\label{eq:C}
    	C = 2^{q/2-1}e^{-1/4}\lr{\sqrt{1+\frac{4}{e}}-1}.
    \end{align}
    For any
    \begin{align}
        |\lambda|\Gamma \leq \frac{2^{-q/2}}{\sqrt{q\max\{q,L\}}},
    \end{align}
    we have
    \begin{align}
        \sup_{|\beta| \leq C/\sqrt{q\max\{q,L\}}} \frac{\E \left|\wh Z(\beta) - \E \wh Z(\beta)\right|^2}{\left|\E \wh Z(\beta)\right|^2} = O\lr{n^{1-q/2}}.
    \end{align}
\end{theorem}

\begin{corollary}[Absence of a phase transition]\label{cor:zf}
For any
\begin{align}
    |\lambda|\Gamma \leq \frac{2^{-q/2}}{\sqrt{q\max\{q,L\}}},
\end{align}
we have
\begin{align}
    \pr{\wh Z(\beta) \neq 0 \text{ for all } |\beta| < \frac{19C/20}{\sqrt{q\max\{q,L\}}}} \geq 1 - O\lr{n^{1-q/2}},
\end{align}
where $C$ is given by \eqref{eq:C}.
\end{corollary}
\begin{proof}
    Fix constants $R, r$ satisfying
    \begin{align}\label{eq:rR}
        0 < r < \frac{19C/20}{\sqrt{q\max\{q,L\}}}, \quad R = \frac{1}{2}\lr{r + \frac{C}{\sqrt{q\max\{q,L\}}}}
    \end{align}
    and define on $|\beta| \leq R$ the function
    \begin{align}
        F(\beta) = \frac{\wh Z(\beta) - \E \wh Z(\beta)}{\E \wh Z(\beta)}.
    \end{align}
    We rewrite this using Cauchy's integral formula on the circle $|w| = R$ as
    \begin{align}
        F(\beta) = \frac{1}{2\pi i} \oint_{|w|=R} \frac{F(w)}{w-\beta} dw
    \end{align}
    to obtain bound
    \begin{align}
        |F(\beta)| \leq \frac{1}{2\pi} \int_0^{2\pi} \frac{|F(R e^{i\theta})|and }{|Re^{i\theta}-\beta|} R d\theta \leq \frac{R}{R-r} \cdot \frac{1}{2\pi} \int_0^{2\pi} |F(Re^{i\theta})| d\theta.
    \end{align}
    Applying Jensen's inequality and taking the supremum over $|\beta| \leq r$, this gives
    \begin{align}
        \E\, \sup_{|\beta| \leq r} |F(\beta)|^2 \leq \lr{\frac{R}{R-r}}^2 \sup_{|w| = R} \E\,|F(w)|^2.
    \end{align}
    Applying Markov's inequality to the random variable $\sup_{|\beta| \leq r} |F(\beta)|^2$ gives
    \begin{align}
        \pr{\sup_{|\beta| \leq r} |F(\beta)|^2 \geq \frac{1}{4}} \leq 4 \E\, \sup_{|\beta| \leq r} |F(\beta)|^2 \leq 4\lr{\frac{R}{R-r}}^2 \sup_{|w| = R} \E\,|F(w)|^2.
    \end{align}
    Given our choice of $R$ in \eqref{eq:rR}, this combines with Theorem~\ref{thm:sa} to give
    \begin{align}
        \pr{\sup_{|\beta| \leq r} |F(\beta)| < \frac{1}{2}} \geq 1 - O(n^{1-q/2}).
    \end{align}
    Conditioned on this event, we conclude that
    \begin{align}
        |\wh Z(\beta)| = |\E \wh Z(\beta)| \, |1 + F(\beta)| \geq \frac{1}{2} |\E \wh Z(\beta)|.
    \end{align}
    By Theorem~\ref{thm:anneal}, we conclude that
    \begin{align}
        |\wh Z(\beta)| > 0 \quad \text{for all} \quad |\beta| \leq r
    \end{align}
    with probability $1 - O(n^{1-q/2})$.
\end{proof}

In Theorem~\ref{thm:obs}, we show the algorithmic consequence of this zero-freeness, which we informally stated in Theorem~\ref{thm:main2}.

We now describe the main existing technical tools we require, as well as some useful bounds.

\subsection{Commutation index}

Although our proofs don't require the commutation index, we formally recall it due to its relevance to prior results and for building intuition.
We denote the set of all $q$-body Majorana operators on $n$ modes by
\begin{align}
    \cS_q^n = \left\{\psi_I \,:\, I \in \binom{[n]}{q}\right\}
\end{align}
and recall a fact about its commutation index.

\begin{theorem}[Commutation index]\label{thm:comm}
    \begin{align}
        \Delta(\cS_q^n) = \sup_\rho \binom{n}{q}^{-1} \sum_{I \in \binom{[n]}{q}} \lr{\Tr(\psi_I \rho)}^2 = \Theta\lr{n^{-q/2}}.
    \end{align}
\end{theorem}

In the heuristic argument described in the introduction, we used the following second moment bound.

\begin{lemma}[Commutation index for second moment]\label{lem:comm2}
    For every trace-class operator $X$,
    \begin{align}
        \binom{n}{q}^{-1} \sum_{I \in \binom{[n]}{q}} \abs{\tr(\psi_I X)}^2 \leq 2 \Delta(\cS_q^n)(\tr\abs{X})^2.
    \end{align}
\end{lemma}
\begin{proof}
Let $H$ be Hermitian with decomposition $H = H_+ - H_-$ onto positive and negative eigenvalues. Set $\rho_+ = H_+/\Tr(H_+)$ and $\rho_- = H_-/\Tr(H_-)$, omitting either quantity if vanishing. Since $\rho_\pm$ are positive and $\psi_I$ is Hermitian,
\begin{align}
    \lr{\binom{n}{q}^{-1}\sum_I \abs{\Tr(\psi_I H)}^2}^{1/2} &\leq \Tr(H_+)\lr{\binom{n}{q}^{-1}\sum_I \abs{\Tr(\psi_I \rho_+)}^2}^{1/2}\nonumber\\
    &\quad + \Tr(H_-)\lr{\binom{n}{q}^{-1}\sum_I \abs{\Tr(\psi_I \rho_-)}^2}^{1/2}.
\end{align}
Applying the definition
\begin{align}
    \Delta(\cS_q^n) = \sup_\rho \binom{n}{q}^{-1} \sum_{I \in \binom{[n]}{q}} \lr{\Tr(\psi_I \rho)}^2
\end{align}
and squaring both sides gives
\begin{align}
    \binom{n}{q}^{-1}\sum_I \abs{\Tr(\psi_I H)}^2 \leq \lr{\left[\Tr(H_+) + \Tr(H_-)\right]\sqrt{\Delta(\cS_q^n)}}^2 = \Delta(\cS_q^n) \norm{H}_1^2.
\end{align}
Now decompose $X = A + iB$ for Hermitian $A, B$ and apply the above result to obtain
\begin{align}
    \binom{n}{q}^{-1}\sum_I \abs{\Tr(\psi_I X)}^2 &= \binom{n}{q}^{-1}\sum_I \abs{\Tr(\psi_I A)}^2 + \binom{n}{q}^{-1}\sum_I \abs{\Tr(\psi_I B)}^2 \\
    &\leq \Delta(\cS_q^n)\lr{\norm{A}_1^2 + \norm{B}_1^2} \\
    &\leq 2\Delta(\cS_q^n)\norm{X}_1^2,
\end{align}
where the first equality uses the fact that $\Tr(\psi_I A), \Tr(\psi_I B) \in \R$, and the final inequality uses $\norm{A}_1, \norm{B}_1 \leq \norm{X}_1$ since $A = (X+X^\dagger)/2$ implies $\norm{A}_1 \leq (\norm{X}_1 + \norm{X^\dagger}_1)/2$ and similarly for $B$. We state the final result in terms of the normalized trace.
\end{proof}

\subsection{Koteck\'y-Preiss}

Our main results rely on Koteck\'y-Preiss. We state the theorem as it is presented in~\cite{fernandez2007cluster} (see (2.6) and (2.14)) and give a slightly more general version of the corollary they report there.

\begin{theorem}[Koteck\'y-Preiss]\label{thm:kp}
    Let $\cP$ be a polymer set with incompatibility relation $\nsim$ such that every polymer $s \in \cP$ satisfies $s \nsim s$. For each finite $\Lambda \subset \cP$ and vector $(z_s)_s \in \C^\cP$, define the polymer partition function
    \begin{align}
        \Xi_\Lambda(z) = 1 + \sum_{r \geq 1} \frac{1}{r!} \sum_{(s_1,\dots,s_r) \in \Lambda^r} \lr{\prod_{j=1}^r z_{s_j}}\lr{\prod_{1 \leq j<k \leq r}1\{s_j\sim s_k\}}.
    \end{align}
    Let $a, \rho:\cP\to[0,\infty)$ be functions such that
    \begin{align}\label{eq:kp}
        \sum_{s' \in \cP \,:\, s' \nsim s} \rho_{s'} e^{a(s')} \leq a(s) \text{ for all } s \in \cP.
    \end{align}
    Let $\cG(s_1, \dots, s_r)$ be the graph with vertex set $[r]$ and edge set
    \begin{align}
        E(\cG(s_1, \dots, s_r)) = \left\{\{i,j\}:1\leq i < j \leq r, \, s_i \nsim s_j \right\}.
    \end{align}
    Define the Ursell coefficient
    \begin{align}
        \phi^T(s_1, \dots, s_r) = \sum_{\substack{G \text{ connected graph on }[r]\\ E(G) \subseteq E(\cG(s_1,\dots,s_r))}} (-1)^{|E(G)|},
    \end{align}
    with the convention $\phi^T(s_1) = 1$ when $r=1$. Then for every finite $\Lambda \subset \cP$ and $z$ satisfying $|z_s| \leq \rho_s$ for all $s \in \cP$, one has
    \begin{align}
        \Xi_\Lambda(z) \neq 0
    \end{align}
    and the cluster expansion
    \begin{align}
        \log \Xi_\Lambda(z) &= \sum_{r \geq 1} \frac{1}{r!} \sum_{(s_1,\dots,s_r)\in \Lambda^r} \phi^T(s_1,\dots,s_r) \prod_{j=1}^r z_{s_j}
    \end{align}
    converges absolutely and has termwise derivative
    \begin{align}
        \partial_{z_s} \log \Xi_\Lambda(z) &= 1 + \sum_{r \geq 1} \frac{1}{r!} \sum_{(s_1,\dots,s_r)\in \Lambda^r} \phi^T(s,s_1,\dots,s_r) \prod_{j=1}^r z_{s_j},
    \end{align}
    which also converges absolutely. Finally, for every $s \in \cP$,
    \begin{align}\label{eq:pibound}
        1 + \sum_{r \geq 1} \frac{1}{r!} \sum_{(s_1,\dots,s_r)\in \cP^r} \abs{\phi^T(s, s_1,\dots,s_r)} \prod_{j=1}^r \rho_{s_j} \leq e^{a(s)},
    \end{align}
    i.e., for every $s \in \Lambda \subset \cP$ and $z$ satisfying $|z_s| \leq \rho_s$ for all $s \in \cP$,
    \begin{align}
        \abs{\frac{\partial_{z_s} \Xi_\Lambda(z)}{\Xi_\Lambda(z)}} \leq e^{a(s)}.
    \end{align}
\end{theorem}
\begin{corollary}[Pinned Koteck\'y-Preiss]\label{cor:kp}
    Let $\Lambda \subset \cP$ be finite and let $z$ satisfy $|z_s| \leq \rho_s$ for all $s \in \cP$. For any $\Lambda_0 \subseteq \Lambda$,
    \begin{align}
        \abs{\log\frac{\Xi_\Lambda(z)}{\Xi_{\Lambda_0}(z)}} \leq \sum_{s \in \Lambda\setminus \Lambda_0} |z_s|e^{a(s)}.
    \end{align}
\end{corollary}
\begin{proof}
    For $|\lambda| \leq 1$, define
    \begin{align}
        z_s^{(\lambda)} = \begin{cases}
            \lambda z_s & s \in \Lambda \setminus \Lambda_0 \\
            z_s & s \notin \Lambda \setminus \Lambda_0.
        \end{cases}
    \end{align}
    Then since $|z_s^{(\lambda)}| \leq |z_s| \leq \rho_s$, we have
    \begin{align}
        \Xi_\Lambda(z^{(\lambda)}) \neq 0.
    \end{align}
    The chain rule gives
    \begin{align}
        \frac{d}{d\lambda} \log \Xi_\Lambda(z^{(\lambda)}) = \sum_{s \in \Lambda \setminus \Lambda_0} z_s \partial_{z_s} \log \Xi_\Lambda(z^{(\lambda)}).
    \end{align}
    By Theorem~\ref{thm:kp}, we have derivative
    \begin{align}
        \partial_{z_s} \log \Xi_\Lambda(z^{(\lambda)}) = 1 + \sum_{r\geq 1} \frac{1}{r!} \sum_{(s_1,\dots,s_r)\in\Lambda^r} \phi^T(s,s_1,\dots,s_r) \prod_{j=1}^r z_{s_j}^{(\lambda)}
    \end{align}
    and thus
    \begin{align}
        \abs{\partial_{z_s} \log \Xi_\Lambda(z^{(\lambda)})} &\leq 1 + \sum_{r\geq 1} \frac{1}{r!} \sum_{(s_1,\dots,s_r)\in\Lambda^r} \abs{\phi^T(s,s_1,\dots,s_r)} \abs{z_{s_1}^{(\lambda)}}\cdots\abs{z_{s_r}^{(\lambda)}}\\
        &\leq 1 + \sum_{r\geq 1} \frac{1}{r!} \sum_{(s_1,\dots,s_r)\in\cP^r} \abs{\phi^T(s,s_1,\dots,s_r)} \rho_{s_1} \cdots \rho_{s_r}\\
        &\leq e^{a(s)}
    \end{align}
    where the final inequality follows from \eqref{eq:pibound}.
    Hence, we have
    \begin{align}
        \left|\frac{d}{d\lambda}\log \Xi_\Lambda(z^{(\lambda)})\right| \leq \sum_{s \in \Lambda \setminus \Lambda_0} |z_s| \left|\partial_{z_s} \log \Xi_\Lambda(z^{(\lambda)})\right| \leq \sum_{s \in \Lambda \setminus \Lambda_0}|z_s|e^{a(s)}.
    \end{align}
    Since $\Xi_\Lambda(z) = \Xi_\Lambda(z^{(1)})$ and $\Xi_{\Lambda_0}(z) = \Xi_\Lambda(z^{(0)})$, we can integrate to obtain
    \begin{align}
        \log \frac{\Xi_\Lambda(z)}{\Xi_{\Lambda_0}(z)} = \int_0^1 \frac{d}{d\lambda} \log \Xi_\Lambda(z^{(\lambda)}) d\lambda,
    \end{align}
    which gives
    \begin{align}
        \left|\log \frac{\Xi_\Lambda(z)}{\Xi_{\Lambda_0}(z)}\right| \leq \int_0^1 \sum_{s \in \Lambda\setminus\Lambda_0} |z_s| e^{a(s)} d\lambda = \sum_{s \in \Lambda\setminus\Lambda_0} |z_s|e^{a(s)}.
    \end{align}
\end{proof}

\subsection{Computing thermal expectations}

We give a standard argument for estimating thermal expectations $\Tr(O\rho_\beta) = -\frac{1}{\beta} \partial_\lambda\big|_{\lambda=0} \log Z_\lambda(\beta)$ for partition function $Z_\lambda$ of $H_0 + \lambda O$ via a finite difference at small $\lambda$. Due to the approximation of the derivative by a finite difference, smaller $\lambda$ reduces the error; it suffices to take inverse polynomial $\lambda$ to obtain inverse polynomial additive error $\epsilon$ of the thermal expectation. (See Lemma 11 of \cite{bravyi2021complexity} for a similar discussion of this argument.) Since the perturbed Hamiltonian's partition function remains zero-free at constant $\lambda$ given a local, bounded-norm observable, this argument holds for any relevant choice of $\epsilon$. Finally, to compute the coefficients in the expansion of $\log Z_\lambda$, we compute the cumulant expansion from moments of the Hamiltonian, similarly to \cite{harrow2020classical} and prior works. Accounting for all these steps yields a final runtime of $n^{O(\log n/\epsilon)}$.

\begin{theorem}[Estimating local observables]\label{thm:obs}
	There is a deterministic classical algorithm which, given the couplings $\{J_I\}$ of an SYK instance $H_0$,
	the coefficients $c_a$ and sets $A_a \subset [n]$ of an observable
	\begin{align}
		O = \sum_{a=1}^M c_a \psi_{A_a}, \qquad L = \max_a |A_a|, \qquad \Gamma = \max_{x\in[n]} \sum_{a:x\in A_a}|c_a|,
	\end{align}
	an inverse temperature satisfying for $C$ given by \eqref{eq:C}
	\begin{align}
		0 \leq \beta < \frac{9C/10}{\sqrt{q\max\{q,L\}}},
	\end{align}
	and any
	\begin{align}
		0 < \epsilon < \frac{2^{3-q/2}\beta n^2 \Gamma}{\sqrt{q\max\{q,L\}}},
	\end{align}
	outputs in time $n^{O(L \log n \Gamma/\epsilon)}$ an estimate of $\Tr(O \rho_\beta)$ within additive error $\epsilon$ with probability $1 - O(n^{1-q/2})$ over the SYK disorder.
\end{theorem}
\begin{proof}
	Let $\wh Z_\lambda$ denote the normalized partition function of $H_0 + \lambda O$. We will give an algorithm to estimate $\wh Z_0$ and $\wh Z_h$ for
	\begin{align}
		h = \frac{\epsilon}{8\beta n^2 \Gamma^2}
	\end{align}
	which we use to give a finite difference approximation of $\Tr(O\rho_\beta) = -\frac{1}{\beta} \partial_\lambda\big|_{\lambda=0} \log \wh Z_\lambda$. By the stated assumption on $\epsilon$, note that
	\begin{align}
		h\Gamma \leq \frac{2^{-q/2}}{\sqrt{q\max\{q,L\}}},
	\end{align}
	and thus the zero-free disk of Corollary~\ref{cor:zf} includes the input value of $\beta$ for both $\lambda=0$ and $\lambda=h$. To control the error of the Taylor series of $\log \wh Z_\lambda$, we will ultimately use Cauchy's estimate on a disk of radius $R$ defined as
		\begin{align}
		R = (1-\delta/2)\rho, \quad \rho = \frac{19C/20}{\sqrt{q\max\{q,L\}}},
	\end{align}
	so $\beta \leq (1-\delta)\rho < R < \rho$, where $\rho$ is within the radius of convergence given by Corollary~\ref{cor:zf} and $\delta = 1/19$.
		This will require an upper bound on $\sup_{|z| \leq R} |\log \wh Z_\lambda(z)|$. We obtain this by Borel-Caratheodory since $\log \wh Z_\lambda(0) = 0$:
	\begin{align}\label{eq:bc}
		\sup_{|z| \leq R} |\log \wh Z_\lambda(z)| \leq \frac{2R}{(1-\delta/4)\rho - R} \sup_{|w| \leq (1-\delta/4)\rho} \Re \log \wh Z_\lambda(w).
	\end{align}
		For $H = H_0 + \lambda O$, we have
	\begin{align}
		\abs{\tr(e^{-w H})} \leq \tr \abs{e^{-wH}} \leq \tr e^{|w| \, |H|} \leq e^{|w| \, \norm{H}}
	\end{align}
	and thus
	\begin{align}
		\Re \log \wh Z_\lambda(w) = \log \abs{\tr e^{-w H}} \leq |w|\, \norm{H}.
	\end{align}
	We first upper bound the norm of $H_0$: by Theorem 4.1.1 of~\cite{tropp2015introduction},
	\begin{align}\label{eq:syknorm}
		\pr{\norm{H_0} \geq t} \leq 2^{n/2+1} \exp[-\frac{t^2}{2 \cdot \binom{n}{q} \frac{(q-1)!}{n^{q-1}} 2^{-q}}] \implies \pr{\norm{H_0} > \sqrt{\frac{1.01\log 2}{q 2^q}} n} = \exp[-\Omega(n)].
	\end{align}
	Using $\norm{O} \leq n\Gamma$, we have that
	\begin{align}
		\norm{H_0 + h O} \leq \lr{h\Gamma + \sqrt{\frac{1.01\log 2}{q 2^q}}} n
	\end{align}
	with probability $1 - e^{-\Omega(n)}$. Combining this with \eqref{eq:bc} gives
	\begin{align}
		\sup_{|z| \leq R} |\log \wh Z_\lambda(z)| \leq B n, \quad B = \frac{2R(1-\delta/4)\rho}{(1-\delta/4)\rho - R}\lr{h\Gamma + \sqrt{\frac{1.01\log 2}{q 2^q}}}.
	\end{align}
	Consequently, by Cauchy's estimate on the circle $|z|=R$, the Taylor expansion
	\begin{align}
		\log \wh Z_\lambda(z) = \sum_{m \geq 1} a_m(\lambda) z^m
	\end{align}
	satisfies
	\begin{align}
		\abs{a_m(\lambda)} \leq \sup_{|z| \leq R} |\log \wh Z_\lambda(z)| R^{-m} \leq B n R^{-m}
	\end{align}
	and thus truncating the Taylor series at $m \leq K$ has remainder at most $B n \sum_{m > K} (\beta/R)^m$. We choose $K$ such that
	\begin{align}\label{eq:Kbound}
		B n \sum_{m > K} \lr{\frac{\beta}{R}}^m \leq \frac{\beta h \epsilon}{4},
	\end{align}
	e.g.,
	\begin{align}
		K \geq \frac{\log\lr{4Bn/\beta h\epsilon(1-\beta/R)}}{\log(R/\beta)}.
	\end{align}
	For the finite difference, this truncation gives error
	\begin{align}
		\abs{-\frac{1}{\beta h}\lr{\sum_{m=1}^K a_m(h) \beta^m - \sum_{m=1}^K a_m(0) \beta^m} + \frac{1}{\beta h}\lr{\log \wh Z_h(\beta) - \log \wh Z_0(\beta)}} \leq \frac{2}{\beta h} \cdot \frac{\beta h \epsilon}{4} = \frac{\epsilon}{2}.
	\end{align}
	We will add the error associated with the truncation to the finite difference error, which is bounded by Taylor's theorem as
	\begin{align}\label{eq:taylor}
		\abs{-\frac{\log \wh Z_h(\beta) - \log \wh Z_0(\beta)}{\beta h} - \Tr(O \rho_\beta)} &= \frac{1}{\beta}\abs{-\frac{\log \wh Z_h(\beta) - \log \wh Z_0(\beta)}{h} - \partial_\lambda\Big|_{\lambda=0} \log \wh Z_\lambda(\beta)} \\
        &\leq \frac{h}{2\beta} \max_{0 \leq \lambda \leq h} \abs{\partial_\lambda^2 \log \wh Z_\lambda(\beta)}.
	\end{align}
	The argument follows Lemma 11 of~\cite{bravyi2021complexity}. By Duhamel's formula,
	\begin{align}
		\partial_\lambda^2 \log \wh Z_\lambda(\beta) = \beta^2 \int_0^1 \frac{\tr(e^{-s\beta(H_0+\lambda O)}Oe^{-(1-s)\beta(H_0+\lambda O)}O)}{\tr(e^{-\beta(H_0+\lambda O)})}ds - \beta^2 \lr{\frac{\tr(Oe^{-\beta(H_0+\lambda O)})}{\tr(e^{-\beta(H_0+\lambda O)})}}^2.
	\end{align}
	In the eigenbasis of the perturbed Hamiltonian $H_0 + \lambda O = \sum_i E_i\ketbra{i}$, let $p_i = e^{-\beta E_i} / Z_\lambda(\beta)$ denote the Gibbs weight of $E_i$. Since $O$ is Hermitian, we can rewrite the Duhamel form as
	\begin{align}
		\frac{\partial_\lambda^2 \log \wh Z_\lambda(\beta)}{\beta^2} = \sum_{ij} \lr{\int_0^1 p_i^s p_j^{1-s} ds} \abs{O_{ij}}^2 - \lr{\sum_i p_i O_{ii}}^2.
	\end{align}
	This quantity is nonnegative since $\sum_i p_i O_{ii}^2 - \lr{\sum_i p_i O_{ii}}^2 \geq 0$ for Hermitian $O$. For the upper bound, AM-GM gives $\int_0^1 p_i^s p_j^{1-s} ds \leq (p_i + p_j)/2$ and thus
	\begin{align}
		\frac{\partial_\lambda^2 \log \wh Z_\lambda(\beta)}{\beta^2} \leq \sum_{ij} \frac{p_i+p_j}{2} \abs{O_{ij}}^2 - \lr{\sum_i p_i O_{ii}}^2 = \Tr(O^2 \rho_\beta) - \Tr(O\rho_\beta)^2 \leq \norm{O}^2.
	\end{align}
	This gives $0 \leq \partial_\lambda^2 \log \wh Z_\lambda(\beta) \leq \beta^2 \norm{O}^2 \leq \beta^2 n^2 \Gamma^2$. Returning to the Taylor bound \eqref{eq:taylor}, we have
	\begin{align}
		\abs{-\frac{\log \wh Z_h(\beta) - \log \wh Z_0(\beta)}{\beta h} - \Tr(O \rho_\beta)} \leq \frac{h}{2\beta} \cdot \beta^2 n^2 \Gamma^2 \leq \frac{\epsilon}{16}.
	\end{align}
	Since $\epsilon/2 + \epsilon/16 < \epsilon$, this truncation and finite difference gives an estimate of the thermal expectation within additive error $\epsilon$.
	
	It remains to bound the runtime of computing coefficients $a_m(h)$ and $a_m(0)$ for all $1 \leq m \leq K$. We describe how to do this for $H = H_0 + \lambda O$ following \cite{harrow2020classical}. We first compute the moments $\mu_r = \tr(H^r)$, so that
	\begin{align}
		\wh Z(z)=\tr(e^{-zH})=\sum_{r\geq 0}\frac{(-1)^r\mu_r}{r!}z^r,
	\end{align}
	and then convert $\{\mu_r\}_{r\leq K}$ to the \emph{cumulants} $\{\kappa_r\}_{r\leq K}$, which are the coefficients in $\log \wh Z(z) = \sum_{m \geq 1} \frac{(-1)^m}{m!} \kappa_m z^m$; these give the desired coefficients $a_r=(-1)^r\kappa_r/r!$. To do this efficiently, we use the moment-to-cumulant recurrence costing $O(K^2)$ arithmetic operations,
    \begin{align}
    \kappa_1=\mu_1, \qquad \kappa_r=\mu_r-\sum_{i=1}^{r-1}\binom{r-1}{i-1}\kappa_i \mu_{r-i} \quad (r \geq 2).
    \end{align}
    Since $H$ is a sum of $O(n^q) + n^{O(L)}$ monomials, the number of terms in $H^j$ is at most $n^{O((q+L)j)}$. Hence, we obtain a total arithmetic runtime of $n^{O((q+L)K)}$, which produces the claimed runtime by substituting \eqref{eq:Kbound} and taking constant $q,L$.
\end{proof}

\subsection{Useful bounds}
We record some bounds that we will use repeatedly.

\begin{fact}\label{fac:tfunc}
    For $0 \leq x \leq 1/e$, the series
    \begin{align}
        T(x) = \sum_{k \geq 1} \frac{k^{k-1}}{k!} x^k
    \end{align}
    converges and satisfies $T(x) \leq 1$. In particular,
    \begin{align}
        T\lr{\frac{1}{2\sqrt e}} = \frac{1}{2} \quad \text{and} \quad T(x) \leq \frac{1}{2} \text{ for all } 0 \leq x \leq \frac{1}{2\sqrt e}
    \end{align}
    and
    \begin{align}
        T\lr{\frac{1}{4 e^{1/4}}} = \frac{1}{4} \quad \text{and} \quad T(x) \leq \frac{1}{4} \text{ for all } 0 \leq x \leq \frac{1}{4 e^{1/4}}.
    \end{align}
\end{fact}
\begin{proof}
By the Stirling lower bound $k! \geq \sqrt{2\pi k}\,(k/e)^k$, we have that $T(x)$ converges on $[0,1/e]$ since
\begin{align}
    0 \leq \frac{k^{k-1}x^k}{k!} \leq \frac{k^{k-1}e^{-k}}{k!} \leq \frac{1}{\sqrt{2\pi}k^{3/2}}.
\end{align}
The Lagrange inversion theorem shows that $T$ satisfies $x = T(x) e^{-T(x)}$; hence, $T(x)\leq 1$ for $x\leq 1/e$, and $T(te^{-t})=t$ for $t \in [0,1]$. Taking $t=1/2$ and
$t=1/4$ thus gives
\begin{align}
    T\lr{\frac{1}{2\sqrt e}} = \frac{1}{2}, \qquad T\lr{\frac{1}{4e^{1/4}}}=\frac{1}{4}.
\end{align}
Since all coefficients of $T$ are nonnegative, $T$ is nondecreasing on $[0,1/e]$, and the two stated inequalities follow.
\end{proof}

\begin{lemma}[Counting factors]\label{lem:count}
    Let $B_1,\dots,B_r \subseteq [n]$ be fixed sets with $|B_j|\leq L$. Form the intersection graph $G$ on the labeled objects
    \begin{align}
        K_1,\dots,K_k,B_1,\dots,B_r,
    \end{align}
    with an edge iff the corresponding supports intersect. Then the following bounds hold.
    \begin{enumerate}
        \item If $x\in[n]$ is fixed, then the number of ordered $k$-tuples $(K_1,\dots,K_k)\in \binom{[n]}{q}^k$ such that
        \begin{align}
            x\in \Big(\bigcup_{i=1}^k K_i\Big)\cup \Big(\bigcup_{j=1}^r B_j\Big)
        \end{align}
        and $G$ is connected is at most
        \begin{align}
            \binom{n-1}{q-1}^{k} (kq+rL)^{k+r-1}.
        \end{align}

        \item If $k\geq 1$ and $K_1\in\binom{[n]}{q}$ is fixed, then the number of ordered $(k-1)$-tuples $(K_2,\dots,K_k)\in \binom{[n]}{q}^{k-1}$ such that $G$ is connected is at most
        \begin{align}
            q \binom{n-1}{q-1}^{k-1} (kq+rL)^{k+r-2}.
        \end{align}

        \item If $k\geq 2$, $e\in\{2,\dots,k\}$ is fixed, and $K_1\in\binom{[n]}{q}$ is fixed, then the number of ordered $(k-2)$-tuples $(K_2,\dots,K_{e-1},K_{e+1},\dots,K_k)\in \binom{[n]}{q}^{k-2}$ such that $K_e=K_1$ and $G$ is connected is at most
        \begin{align}
            q \binom{n-1}{q-1}^{k-2} (kq+rL)^{k+r-3}.
        \end{align}

        \item If $k\geq 1$, $r\geq 1$, $f\in[r]$ is fixed, and $K_1=B_f\in\binom{[n]}{q}$, then the number of ordered $(k-1)$-tuples $(K_2,\dots,K_k)\in \binom{[n]}{q}^{k-1}$ such that $G$ is connected is at most
        \begin{align}
            q \binom{n-1}{q-1}^{k-1} (kq+rL)^{k+r-3}.
        \end{align}
    \end{enumerate}
    In cases (3) and (4), when $k+r=2$ the count is 1.
\end{lemma}
\begin{proof}
    For $N=k+r$ and $u\in[N]$, write
    \begin{align}
        S_u =
        \begin{cases}
            K_u,& 1\leq u\leq k,\\
            B_{u-k},& k+1\leq u\leq N.
        \end{cases}
    \end{align}
    For all cases, we will upper bound the indicator function that the intersection graph $G$ is connected by summing over the possible spanning trees, i.e.,
    \begin{align}
        1\{G\text{ connected}\} \leq \sum_{T\in\mathrm{Tree}(N)} 1\{E(T) \subseteq E(G)\},
    \end{align}
    since every connected graph has at least one spanning tree.
    We will also use the weighted Cayley identity
    \begin{align}\label{eq:weighted-cayley-lemcount}
        \sum_{T\in \mathrm{Tree}(N)} \prod_{u=1}^N |S_u|^{\deg_T(u)-1}
        =
        \lr{\sum_{u=1}^{k+r} |S_u|}^{N-2}.
    \end{align}
    For (1), if the intersection graph is connected and $x\in \bigcup_{u=1}^{k+r} S_u$, then there is a root vertex $v\in[N]$ with $x\in S_v$ and a spanning tree $T\subseteq G$. Root $T$ at $v$. If $v\le k$, then $K_v$ must contain $x$, so there are at most $\binom{n-1}{q-1}$ choices for $K_v$; if $v>k$, there is no choice since $B_{v-k}$ is fixed.
    For every non-root vertex $w$, let $p(w)$ be its parent. If $w\le k$, then the edge $\{w,p(w)\}$ requires $K_w\cap S_{p(w)}\neq \emptyset$, so there are at most $|S_{p(w)}|\binom{n-1}{q-1}$ choices for $K_w$. If $w>k$, then $B_{w-k}$ is fixed and $1\{B_{w-k}\cap S_{p(w)}\neq \emptyset\}\le |S_{p(w)}|$. Hence, for fixed rooted $(v,T)$, the number of admissible configurations is at most
    \begin{align}
        1\{x\in S_v\}\binom{n-1}{q-1}^{k}\prod_{u=1}^N |S_u|^{\#\mathrm{children}_T(u)} = 1\{x\in S_v\}\binom{n-1}{q-1}^{k} |S_v| \prod_{u=1}^N |S_u|^{\deg_T(u)-1}.
    \end{align}
    Summing over $v$ and $T$ gives
    \begin{align}
        \#\{\text{admissible }(K_1,\dots,K_k)\}
        &\le
        \binom{n-1}{q-1}^{k}
        \sum_{v=1}^N 1\{x\in S_v\} |S_v|
        \sum_{T\in \mathrm{Tree}(N)} \prod_{u=1}^N |S_u|^{\deg_T(u)-1}\\
        &\le
        \binom{n-1}{q-1}^{k}
        \Big(\sum_{v=1}^N |S_v|\Big)
        \Big(\sum_{u=1}^N |S_u|\Big)^{N-2}\\
        &\le
        \binom{n-1}{q-1}^{k} (qk+Lr)^{N-1},
    \end{align}
    proving (1).

    For (2), root a spanning tree at the fixed vertex corresponding to $K_1$. The root contributes a factor $|K_1| = q$, while every non-root variable Gaussian set contributes at most $|S_{p(w)}|\binom{n-1}{q-1}$, and every fixed non-root $B_j$ contributes at most $|S_{p(w)}|$ as in the previous case. Thus
    \begin{align}
        \#\{\text{admissible }(K_2,\dots,K_k)\}
        \le
        q\binom{n-1}{q-1}^{k-1}
        \sum_{T\in \mathrm{Tree}(N)} \prod_{u=1}^N |S_u|^{\deg_T(u)-1} \le
        q\binom{n-1}{q-1}^{k-1} (qk+Lr)^{N-2},
    \end{align}
    proving (2).

    For (3), fix $e\in\{2,\dots,k\}$ and impose $K_e=K_1$. Since $K_e$ and $K_1$ have identical support, they have identical neighbors in the intersection graph. Therefore the full graph on
    \begin{align}
        K_1,\dots,K_k,B_1,\dots,B_r
    \end{align}
    is connected iff the graph obtained by deleting the duplicate vertex $K_e$ is connected. Applying (2) to the reduced family
    \begin{align}
        K_1,K_2,\dots,K_{e-1},K_{e+1},\dots,K_k,B_1,\dots,B_r
    \end{align}
    gives
    \begin{align}
        q \binom{n-1}{q-1}^{k-2}\big(q(k-1)+Lr\big)^{k+r-3}
        \le
        q \binom{n-1}{q-1}^{k-2} (qk+Lr)^{k+r-3},
    \end{align}
    proving (3).

    For (4), if $B_f=K_1$, then $B_f$ and $K_1$ again have identical support and hence identical neighbors. So the full graph is connected iff the graph obtained by deleting $B_f$ is connected. Applying (2) to the reduced family
    \begin{align}
        K_1,\dots,K_k,B_1,\dots,B_{f-1},B_{f+1},\dots,B_r
    \end{align}
    gives
    \begin{align}
        q \binom{n-1}{q-1}^{k-1}\big(qk+L(r-1)\big)^{k+r-3}
        \le
        q \binom{n-1}{q-1}^{k-1} (qk+Lr)^{k+r-3},
    \end{align}
    proving (4).
\end{proof}

\section{First moment}\label{sec:first}

\subsection{Polymer representation}
We first derive a polymer representation for the annealed partition function. To do this, we decompose the annealed partition function as
\begin{align}
    \EE{\wh Z(\beta)} = \EE{\tr e^{-\beta H}} = \sum_{\ell=0}^\infty \frac{(-\beta)^\ell}{\ell!} \EE{\tr(H^\ell)}.
\end{align}
The even moments of the unperturbed Hamiltonian can be decomposed as
\begin{align}
    \EE{\tr(H_0^{2k})} = \sigma^{2k} \sum_{\pi \in \Pi(2k)} \sum_{I_1,\dots,I_{2k} \in \binom{[n]}{q}} \lr{\prod_{\{a,b\}\in\pi} 1\{I_a=I_b\}} \tr(\psi_{I_1}\cdots \psi_{I_{2k}}),
\end{align}
where $\Pi$ denotes the set of all pairings of $2k$ elements.

To compute $\EE{\tr(H^m)}$ for the perturbed Hamiltonian, we have to include the additional terms $\psi_{A_a}$. We track these terms with the function $\eta : [m] \to \{0,1\}$ that indicates whether slot $j$ belongs to $H_0$ ($\eta(j)=0$) or to the perturbation ($\eta(j) = 1$). The pairings $\pi$ then belong to the set $\eta^{-1}(0)$. This gives
\begin{align}
    \EE{\tr(H^m)} &=  \sum_{\substack{\eta:[m]\to\{0,1\}\\ |\eta^{-1}(0)| \text{ even}}} \sigma^{|\eta^{-1}(0)|} \lambda^{|\eta^{-1}(1)|} \lr{\prod_{j:\eta(j)=1} c_{a_j}} \sum_{\pi \in \Pi(|\eta^{-1}(0)|)} \nonumber\\
    &\quad \sum_{\substack{I_j \in \binom{[n]}{q} : j \in \eta^{-1}(0)\\ I_u = I_v \, \forall\, \{u,v\} \in \pi}} \sum_{a_j \in [M] : j \in \eta^{-1}(1)} \tr \prod_{j=1}^m \begin{cases}
        \psi_{I_j} & \eta(j) = 0\\
        \psi_{A_{a_j}} & \eta(j) = 1.
    \end{cases}
\end{align}
For brevity, we let $S_m$ denote the set of consistent tuples $(\eta, \pi, \vec I, \vec a)$ satisfying the various conditions underneath the sums.
Throughout this section, we will also use notation
\begin{align}
    X_j(\tau) = \begin{cases}
        \psi_{I_j} & \eta(j) = 0\\
        \psi_{A_{a_j}} & \eta(j) = 1
    \end{cases}
\end{align}
so
\begin{align}
    \EE{\tr(H^m)} = \sum_{(\eta, \pi, \vec I, \vec a) \in S_m} \sigma^{|\eta^{-1}(0)|} \lambda^{|\eta^{-1}(1)|} \lr{\prod_{j:\eta(j)=1} c_{a_j}} \tr \prod_{j=1}^m X_j(\tau).
\end{align}
We define an intersection graph $G(\tau)$ for $\tau \in S_m$. The graph has vertex set $[m]$ and edges between distinct $u, v \in [m]$ iff the supports intersect: i.e.,
\begin{align}
    \begin{cases}
        I_u \cap I_v \neq \emptyset & \eta(u) = \eta(v) = 0\\
        I_u \cap A_{a_v} \neq \emptyset & \eta(u) = 0, \eta(v) = 1\\
        A_{a_u} \cap A_{a_v} \neq \emptyset & \eta(u) = \eta(v) = 1.
    \end{cases}
\end{align}
Let $S_m^\conn$ denote the set of $\tau \in S_m$ whose intersection graph $G(\tau)$ is connected. The set of all polymers is defined as
\begin{align}
    \cP = \bigcup_{m \geq 1} S_m^\conn.
\end{align}
We set two polymers $P, Q \in \cP$ to be compatible ($P\sim Q$) if
\begin{align}
    \supp(P) \cap \supp(Q) = \emptyset
\end{align}
for support
\begin{align}
    \supp(\tau) = \lr{\bigcup_{j:\eta(j)=0} I_j} \cup \lr{\bigcup_{j:\eta(j)=1} A_{a_j}}.
\end{align}
Before deriving the polymer representation of the annealed partition function, we note two useful properties of the intersection graph.

\begin{lemma}[Decomposition of moments into connected components]\label{lem:decomp}
    Let $\tau = (\eta,\pi,\vec I, \vec a) \in S_m$ and let $C_1, \dots, C_r \subseteq [m]$ be the connected components of $G(\tau)$.
    Then
    \begin{align}
        \sigma^{|\eta^{-1}(0)|}\lambda^{|\eta^{-1}(1)|} \tr \prod_{j=1}^m X_j(\tau) = \prod_{u=1}^r \sigma^{|\{j \in C_u : \eta(j)=0\}|} \lambda^{|\{j \in C_u : \eta(j)=1\}|} \tr \prod_{j\in C_u} X_j(\tau),
    \end{align}
    where the product preserves the ordering of $\vec I$ and $\vec a$.
\end{lemma}
\begin{proof}
    If $u \neq v$, then
    \begin{align}
        \lr{\bigcup_{j\in C_u} \supp(X_j(\tau))} \cap \lr{\bigcup_{j\in C_v} \supp(X_j(\tau))} = \emptyset
    \end{align}
    otherwise there would exist an edge between $C_u$ and $C_v$. Every $X_j(\tau)$ is an even Majorana string, so operators supported on disjoint sets commute. Hence,
    \begin{align}
        \prod_{j=1}^m X_j(\tau) = \prod_{u=1}^r \prod_{j \in C_u} X_j(\tau),
    \end{align}
    where the product preserves the relative order of indices inside $C_u$. Factorization of trace over the resulting tensor product decomposition shows the first property. To show that $|\{j\in C_u:\eta(j)=0\}|$ has even cardinality, we observe that each connected component is a disjoint union of pairs, since if $\{u, v\} \in \pi$ then $I_u \cap I_v = I_u \neq \emptyset$.
\end{proof}

\begin{lemma}[Polymer representation]\label{lem:polymer}
    \begin{align}\label{eq:polymer}
        \EE{\tr e^{-\beta H}} = \sum_{r=0}^\infty \frac{1}{r!} \sum_{\substack{s_1, \dots, s_r \in \cP\\ s_i \sim s_j}} \prod_{u=1}^r z_{s_u}(\beta),
    \end{align}
    where for $s = (\eta, \pi, \vec I, \vec a)$ we let
    \begin{align}
        z_s(t) = \frac{(-t\sigma)^{|\eta^{-1}(0)|} (-t\lambda)^{|\eta^{-1}(1)|}}{\lr{|\eta^{-1}(0)| + |\eta^{-1}(1)|}!} \lr{\prod_{j:\eta(j)=1} c_{a_j}} \tr \prod_{j=1}^{|\eta^{-1}(0)| + |\eta^{-1}(1)|} X_j(s).
    \end{align}
\end{lemma}
\begin{proof}
    Decompose polymer $s_u$ as $(\eta_u, \pi_u, \vec I^{(u)}, \vec a^{(u)})$ and let $m_u = |\eta_u^{-1}(0)| + |\eta_u^{-1}(1)|$.
    We will compute the coefficient of $\beta^m$ in \eqref{eq:polymer} and compare it to the coefficient of $\beta^m$ in
    \begin{align}\label{eq:mgf}
        \EE{\tr e^{-\beta H}} = \sum_{m=0}^\infty \frac{(-\beta)^m}{m!} \EE{\tr(H^m)}.
    \end{align}
    I.e., it suffices to show that
    \begin{align}\label{eq:trcoeff}
        \frac{\EE{\tr(H^m)}}{m!} = \sum_{r=0}^\infty \frac{1}{r!} \sum_{\substack{s_1, \dots, s_r \in \cP\\ s_i \sim s_j\\ m_1 + \cdots + m_r = m}} \prod_{u=1}^r \frac{\sigma^{|\eta_u^{-1}(0)|} \lambda^{|\eta_u^{-1}(1)|}}{m_u!} \lr{\prod_{j:\eta_u(j)=1} c_{a_j}} \tr \prod_{j=1}^{m_u} X_j(s_u).
    \end{align}
    Recall from Lemma~\ref{lem:decomp} that if $\tau=(\eta,\pi,\vec I,\vec a)\in S_m$ has connected components $C_1,\dots,C_r$ in $G(\tau)$, then we can decompose $\tau$ over these components. We perform this by introducing $\rho_u:[m_u]\to C_u$ as the unique increasing bijection with $m_u=|C_u|$; then
    \begin{align}\label{eq:trhn}
        \EE{\tr(H^m)}
        =
        \sum_{\tau\in S_m}
        \prod_{u=1}^r
        \left[
            \sigma^{|\{j\in C_u:\eta(j)=0\}|}
            \lambda^{|\{j\in C_u:\eta(j)=1\}|}
            \lr{\prod_{j\in C_u:\eta(j)=1} c_{a_j}}
            \tr\prod_{a=1}^{m_u} X_{\rho_u(a)}(\tau)
        \right].
    \end{align}
    We can construct polymers $(s_1,\dots,s_r)$ corresponding to these connected components, with each polymer decomposed as
    \begin{align}
        s_u = (\eta_u, \pi_u, \vec I^{(u)}, \vec a^{(u)}) \in S_{m_u}^\conn.
    \end{align}
    Namely, we set
    \begin{align}
        \eta_u(a) = \eta(\rho_u(a))
    \end{align}
    and if $\eta_u(a)=0$, we set $I_a^{(u)} = I_{\rho_u(a)}$, while if $\eta_u(a)=1$, we set $a_a^{(u)} = a_{\rho_u(a)}$. Since $\{b,c\}\in\pi$ implies $I_b=I_c$, the property $I_b\cap I_c\neq\emptyset$ ensures that no pair is split across connected components; we thus put
    \begin{align}
        \{\rho_u^{-1}(b),\rho_u^{-1}(c)\}\in \pi_u
        \qquad\text{if}\qquad
        \{b,c\}\in\pi
        \text{ and }
        b,c\in C_u.
    \end{align}
    Note that $s_i\sim s_j$ for $i\neq j$ since $C_i$ and $C_j$ are disjoint connected components, and that $s_u\in S_{m_u}^{\conn}$ since $C_u$ is connected. Moreover, the ordering of the $r$ items in $(s_1,\dots,s_r)$ does not matter.

    To rewrite \eqref{eq:trhn} as a sum over $S_m^\conn$, we show a bijection (up to the ordering of the $r$-tuple) by constructing $\tau\in S_m$ from compatible polymers $s_1,\dots,s_r\in\cP$. Given sizes $m_1,\dots,m_r\geq 1$ such that $\sum_{u=1}^r m_u=m$, given an ordered tuple of disjoint subsets $C_1,\dots,C_r\subseteq [m]$ satisfying
    \begin{align}
        \bigsqcup_{u=1}^r C_u = [m],
        \qquad
        |C_u|=m_u,
    \end{align}
    and given $(s_1,\dots,s_r)$ with $s_u\in S_{m_u}^{\conn}$ and $s_i\sim s_j$ for $i\neq j$, we construct $\tau=(\eta,\pi,\vec I,\vec a)\in S_m$ as follows. Set $\rho_u:[m_u]\to C_u$ to be the unique increasing bijection. As in the first direction, set
    \begin{align}
        s_u = (\eta_u, \pi_u, \vec I^{(u)}, \vec a^{(u)})
    \end{align}
    with $\eta(\rho_u(a))=\eta_u(a)$; if $\eta_u(a)=0$, set $I_{\rho_u(a)} = I_a^{(u)}$, and if $\eta_u(a)=1$, set $a_{\rho_u(a)} = a_a^{(u)}$.
    Define the pairing
    \begin{align}
        \pi
        =
        \bigsqcup_{u=1}^r
        \left\{
            \{\rho_u(a),\rho_u(b)\}
            \,:\,
            \{a,b\}\in \pi_u
        \right\}.
    \end{align}
    Since the polymers are pairwise compatible, no edge of the global intersection graph joins two different $C_u$'s, so the connected components of $G(\tau)$ are exactly $C_1,\dots,C_r$. Thus $\tau\in S_m$. Note that here, the ordering of $(C_1,\dots,C_r)$ and $(s_1,\dots,s_r)$ does not matter (as long as they agree), so $r!$ choices produce the same $\tau$.

    This bijection implies the decomposition
    \begin{align}
        \EE{\tr(H^m)}
        =
        \sum_{r=0}^\infty
        \frac{1}{r!}
        \sum_{\substack{m_1+\cdots+m_r=m\\ m_u\geq 1}}
        \sum_{\substack{C_1,\dots,C_r\subseteq [m]\\ \bigsqcup_u C_u=[m]\\ |C_u|=m_u}}
        \sum_{\substack{s_1,\dots,s_r\in \cP\\ s_i\sim s_j\\ |s_u|=m_u}}
        \prod_{u=1}^r
        \left[
            \sigma^{|\eta_u^{-1}(0)|}
            \lambda^{|\eta_u^{-1}(1)|}
            \lr{\prod_{j:\eta_u(j)=1} c_{a_j^{(u)}}}
            \tr\prod_{j=1}^{m_u} X_j(s_u)
        \right].
    \end{align}
    We replace the sum over $C_1,\dots,C_r$ by the counting factor $\binom{m}{m_1,\dots,m_r}$, so
    \begin{align}
        \frac{\EE{\tr(H^m)}}{m!}
        =
        \sum_{r=0}^\infty
        \frac{1}{r!}
        \sum_{\substack{m_1+\cdots+m_r=m\\ m_u\geq 1}}
        \sum_{\substack{s_1,\dots,s_r\in \cP\\ s_i\sim s_j\\ |s_u|=m_u}}
        \prod_{u=1}^r
        \left[
            \frac{\sigma^{|\eta_u^{-1}(0)|}\lambda^{|\eta_u^{-1}(1)|}}{m_u!}
            \lr{\prod_{j:\eta_u(j)=1} c_{a_j^{(u)}}}
            \tr\prod_{j=1}^{m_u} X_j(s_u)
        \right].
    \end{align}
    We combine the sums over $m_1+\cdots+m_r=m$ and $s_1,\dots,s_r$, noting that every polymer $s_u\in\cP$ automatically satisfies
    \begin{align}
        |s_u|
        =
        |\eta_u^{-1}(0)|
        +
        |\eta_u^{-1}(1)|
        = m_u,
    \end{align}
    to obtain
    \begin{align}
        \frac{\EE{\tr(H^m)}}{m!}
        =
        \sum_{r=0}^\infty
        \frac{1}{r!}
        \sum_{\substack{s_1,\dots,s_r\in \cP\\ s_i\sim s_j\\ m_1+\cdots+m_r=m}}
        \prod_{u=1}^r
        \left[
            \frac{\sigma^{|\eta_u^{-1}(0)|}\lambda^{|\eta_u^{-1}(1)|}}{m_u!}
            \lr{\prod_{j:\eta_u(j)=1} c_{a_j^{(u)}}}
            \tr\prod_{j=1}^{m_u} X_j(s_u)
        \right].
    \end{align}
    This is exactly \eqref{eq:trcoeff}, completing the proof.
\end{proof}

\subsection{Koteck\'y-Preiss}
To apply Koteck\'y-Preiss, we set 
\begin{align}
    a(P) = \frac{|\supp(P)|}{2\alpha}
\end{align}
for some $\alpha$ satisfying $\alpha \geq \max\{q, L\}$.
\begin{lemma}[Annealed condition]\label{lem:annealcrit}
    For $t \in \C$, define
    \begin{align}
        \Xi(t) = \sup_{x \in [n]} \sum_{\substack{s \in \cP\\ x \in \supp(s)}} |z_s(t)| e^{|\supp(s)|/2\alpha}.
    \end{align}
    If $\Xi(t) \leq 1/2\alpha$, then \eqref{eq:kp} holds for $a(s) = |\supp(s)|/2\alpha$.
\end{lemma}
\begin{proof}
    If $s' \nsim s$, then $\supp(s') \cap \supp(s) \neq \emptyset$, so
    \begin{align}
        1\{s' \nsim s\} \leq \sum_{x \in \supp(s)} 1\{x \in \supp(s')\}.
    \end{align}
    This implies \eqref{eq:kp} since
    \begin{align}
        \sum_{s' \nsim s} |z_{s'}(t)| e^{a(s')} \leq \sum_{x \in \supp(s)} \sum_{s':x\in\supp(s')} |z_{s'}(t)| e^{|\supp(s')|/2\alpha} \leq |\supp(s)| \Xi(t) \leq \frac{|\supp(s)|}{2\alpha}.
    \end{align}
\end{proof}
We proceed to bound $\Xi(t)$ to prove Theorem~\ref{thm:anneal}, which is obtained directly from the following lemma.
\begin{lemma}[Annealed Koteck\'y-Preiss]\label{lem:annealcritsat}
    If
    \begin{align}\label{eq:anncond}
        \frac{\alpha e^{q/2\alpha}}{2^{q+1}}|t|^2 + \alpha e^{L/2\alpha} \Gamma |\lambda t| \leq \frac{1}{2\sqrt e}
    \end{align}
    then $\Xi(t) \leq 1/2\alpha$.
\end{lemma}
\begin{proof}
    Let polymer $s = (\eta, \pi, \vec I, \vec a)$ have $|\eta^{-1}(0)| = 2k$ vertices from the SYK Hamiltonian and $|\eta^{-1}(1)| = r$ vertices from the perturbation. Since $\norm{\psi_I} = 2^{-q/2}$ and $\norm{\psi_{A_a}} = 2^{-|A_a|/2} \leq 1$, we have
    \begin{align}
        \tr \prod_{j=1}^{2k+r} X_j(s) \leq 2^{-qk}.
    \end{align}
    For $L = \max_a |A_a|$, we also have
    \begin{align}
        \abs{\supp(s)} \leq qk + Lr
    \end{align}
    and thus
    \begin{align}
        |z_s(t)| e^{a(s)} \leq \frac{(|t|\sigma)^{2k}|\lambda t|^r}{(2k+r)!} \lr{\prod_{j:\eta(j)=1} |c_{a_j}|} 2^{-qk} e^{(qk+Lr)/2\alpha} \leq \frac{(|t|\sigma)^{2k}|\lambda t|^r}{(2k+r)!} \Gamma^r 2^{-qk} e^{(qk+Lr)/2\alpha}.
    \end{align}
    We sum over all such polymers for some fixed $x \in [n]$ to obtain
    \begin{align}
        \sum_{\substack{s \in \cP \\ x \in \supp(s)\\ |\eta^{-1}(0)| = 2k, |\eta^{-1}(1)|=r}} |z_s(t)| e^{a(s)} &\leq \frac{(|t|\sigma)^{2k}|\lambda t|^r}{(2k+r)!} \Gamma^r 2^{-qk} e^{(qk+Lr)/2\alpha} \nonumber\\
        &\quad \times \#\{s \in S_m^\conn : x \in \supp(s), |\eta^{-1}(0)| = 2k, |\eta^{-1}(1)|=r\}.
    \end{align}
    We now bound the counting factor. There are $\binom{2k+r}{r}$ choices for assignments to the perturbation. Recall that identifying $s$ requires identifying a pairing $\pi$ and consistent subsets $I_1, \dots, I_{2k}$ such that $I_a = I_b$ for all $\{a, b\} \in \pi$. Let $K_1, \dots, K_k$ denote the sets corresponding to pairs of $I_1, \dots, I_{2k}$. For each tuple $(K_1, \dots, K_k)$, there are $(2k-1)!!$ choices for $\pi$. The remaining counting factor is given by Lemma~\ref{lem:count} to yield
    \begin{align}
        \sum_{\substack{s \in \cP \\ x \in \supp(s)\\ |\eta^{-1}(0)| = 2k, |\eta^{-1}(1)|=r}} |z_s(t)| e^{a(s)} \leq \binom{2k+r}{r}(2k-1)!! \binom{n-1}{q-1}^k (kq+rL)^{k+r-1} \frac{(|t|\sigma)^{2k}|\lambda t|^r}{(2k+r)!} \Gamma^r 2^{-qk} e^{(qk+Lr)/2\alpha}.
    \end{align}
    Using $\binom{2k+r}{r}(2k-1)!!/(2k+r)! = 1/2^kk!r!$, $\binom{n-1}{q-1} \sigma^2 \leq 1$ and the assumed condition $qk+Lr \leq \alpha(k+r)$, we have
    \begin{align}\label{eq:full1sum}
        \sum_{\substack{s \in \cP \\ x \in \supp(s)\\ |\eta^{-1}(0)| = 2k, |\eta^{-1}(1)|=r}} |z_s(t)| e^{a(s)} \leq \frac{1}{\alpha} \frac{(k+r)^{k+r-1}}{k!r!} \lr{\frac{\alpha e^{q/2\alpha}}{2^{q+1}}|t|^2}^k \lr{\alpha e^{L/2\alpha}\Gamma |\lambda t|}^r.
    \end{align}
    To apply Fact~\ref{fac:tfunc}, we use
    \begin{align}
        \sum_{\substack{k,r \geq 0\\ k+r \geq 1}} \frac{(k+r)^{k+r-1}}{k!r!} y_1^k y_2^r \leq T(y_1 + y_2)
    \end{align}
    to rewrite the sum in \eqref{eq:full1sum} as
    \begin{align}
        \Xi(t) \leq \frac{1}{\alpha} T\lr{\frac{\alpha e^{q/2\alpha}}{2^{q+1}}|t|^2 + \alpha e^{L/2\alpha} \Gamma |\lambda t|}.
    \end{align}
    By Fact~\ref{fac:tfunc}, enforcing
    \begin{align}
        \frac{\alpha e^{q/2\alpha}}{2^{q+1}}|t|^2 + \alpha e^{L/2\alpha} \Gamma |\lambda t| \leq \frac{1}{2\sqrt e}
    \end{align}
    implies that
    \begin{align}
        \sum_{s \in \cP : x \in \supp(s)} |z_s(t)| e^{|\supp(s)|/2\alpha} \leq \frac{1}{2\alpha}.
    \end{align}
\end{proof}

\subsection{Proof of Theorem~\ref{thm:anneal}}
The proof of Theorem~\ref{thm:anneal} follows by applying the result of Lemma~\ref{lem:annealcritsat} to Lemma~\ref{lem:annealcrit} and Theorem~\ref{thm:kp}. We report a slightly looser sufficient condition on $|\beta|$ to write the zero-free region described by \eqref{eq:anncond} as a disk.

\begin{lemma}[Annealed partition function zero-free disk]\label{lem:annealdisk}
The condition \eqref{eq:anncond} is implied by choosing
\begin{align}\label{eq:tann}
	\alpha = \max\{q,L\}, \quad |\lambda|\Gamma \leq 2^{-q/2} \max\{q,L\}^{-1/2}, \quad |t|\cJ = |t| \cdot \sqrt{\frac{q}{2^{q-1}}} \leq \sqrt{2}\lr{\sqrt{1+1/e}-1} \sqrt{\frac{q}{\alpha}}.
\end{align}
\end{lemma}
\begin{proof}
	Note that the condition on $\Gamma|\lambda|$ in \eqref{eq:tann} implies that
	\begin{align}
		|t| \leq \lr{\sqrt{1+1/e}-1} \min\left\{\frac{2^{q/2}}{\sqrt{\alpha}}, \frac{1}{\Gamma|\lambda|\alpha}\right\}.
	\end{align}
	Since additionally $\alpha = \max\{q, L\}$, we have
	\begin{align}
		\frac{\alpha e^{q/2\alpha}}{2^{q+1}}|t|^2 + \alpha e^{L/2\alpha} \Gamma |\lambda t| &\leq \frac{\alpha e^{1/2}}{2^{q+1}}|t|^2 + \alpha e^{1/2} \Gamma |\lambda t|\\
		&\leq \frac{e^{1/2}}{2} \cdot \lr{\sqrt{1+1/e}-1}^2 + e^{1/2}\lr{\sqrt{1+1/e}-1}\\
		&\leq \frac{1}{2\sqrt e}
	\end{align}
	where in the first inequality, we used $e^{q/2\alpha}, e^{L/2\alpha} \leq e^{1/2}$; in the second inequality we used the fact that \eqref{eq:tann} implies
	\begin{align}
		|t|^2 \leq \lr{\sqrt{1+1/e}-1}^2 \frac{2^q}{\alpha}, \quad |t| \leq \lr{\sqrt{1+1/e}-1} \frac{1}{\Gamma|\lambda|\alpha}.
	\end{align}
	This satisfies \eqref{eq:anncond}.
\end{proof}

\section{Second moment}
\subsection{Polymer representation}
Similarly to the annealed computation, we expand the second moment
\begin{align}
    \EE{\wh Z(\beta) \wh Z(\gamma)} = \sum_{\ell_1,\ell_2 \geq 0} \frac{(-\beta)^{\ell_1}}{\ell_1!}\frac{(-\gamma)^{\ell_2}}{\ell_2!} \EE{\tr(H^{\ell_1})\tr(H^{\ell_2})}.
\end{align}
Since we have two copies of the system, we introduce
\begin{align}
    \mu \,:\, [m] \to \{1,2\}
\end{align}
to record the replica label. Each polymer will thus now have the form $\tau = (\eta, \pi, \mu, \vec I, \vec a)$. For $\ell_1+\ell_2=2k$, the moments of the unperturbed Hamiltonian are expanded as
\begin{align}
    \EE{\tr(H^{\ell_1})\tr(H^{\ell_2})} &= \binom{2k}{\ell_1}^{-1} \sigma^{2k}\sum_{\pi \in \Pi(2k)} \sum_{\substack{\mu:[2k]\to\{1,2\}\\|\mu^{-1}(1)|=\ell_1\\ |\mu^{-1}(2)|=\ell_2}}\sum_{I_1,\dots,I_{2k}} \nonumber\\
    &\quad \times \lr{\prod_{\{u,v\}\in\pi}1\{I_u=I_v\}} \tr(\prod_{j:\mu(j)=1}\psi_{I_j}) \tr(\prod_{j:\mu(j)=2}\psi_{I_j})
\end{align}
since there are $\binom{2k}{\ell_1}$ choices for $\mu$. For the perturbed Hamiltonian, let $T_m(\ell_1, \ell_2)$ for $m = \ell_1+\ell_2$ denote the set of all $\tau = (\eta, \pi, \mu, \vec I, \vec a)$ such that
\begin{align}
    |\mu^{-1}(1)| = \ell_1, \quad |\mu^{-1}(2)| = \ell_2
\end{align}
and such that for every $\{u,v\} \in \pi$, $I_u = I_v$. Then
\begin{align}\label{eq:2rep-moment}
    \EE{\tr(H^{\ell_1})\tr(H^{\ell_2})} = \binom{m}{\ell_1}^{-1} \sum_{\tau \in T_m(\ell_1, \ell_2)} \sigma^{|\eta^{-1}(0)|} \lambda^{|\eta^{-1}(1)|} \lr{\prod_{j:\eta(j)=1} c_{a_j}} \tr(W_1(\tau)) \tr(W_2(\tau))
\end{align}
for
\begin{align}
    W_b(\tau) = \prod_{j:\mu(j)=b} X_j(\tau),
\end{align}
where the product preserves the order of $\vec I, \vec a$.

We now define separated and mixed graphs to track the connected contribution. Define the graph $G^\sep(\tau)$ on vertex set $[m]$ by putting an edge between distinct $u,v \in [m]$ iff
\begin{align}
    \mu(u) = \mu(v) \quad \text{and}\quad \supp(X_u(\tau)) \cap \supp(X_v(\tau)) \neq \emptyset.
\end{align}
This corresponds to two non-interacting replicas; we will write $\EE{\wh Z} \cdot \EE{\wh Z}$ in terms of $G^\sep$.

To address the connected piece $\EE{\wh Z(\beta) \wh Z(\gamma)}$, we introduce a second type of edge between replicas. Define the graph $G^\mix(\tau)$ on vertex set $[m]$ by putting an edge between distinct $u, v \in [m]$ iff either
\begin{enumerate}
    \item $\mu(u) = \mu(v)$ and $\supp(X_u(\tau)) \cap \supp(X_v(\tau)) \neq \emptyset$, or
    \item $\{u,v\} \in \pi$, $\eta(u)=\eta(v) = 0$ and $\mu(u) \neq \mu(v)$.
\end{enumerate}
The second type of edge is between distinct replicas ($\mu(u) \neq \mu(v)$) and only connects Wick-paired contributions from the original SYK Hamiltonian.

We write $T_m^\sep, T_m^\mix$ to denote the set of $\tau \in T_m$ whose intersection graph $G^\sep, G^\mix$ is connected, and we define polymers
\begin{align}
    \cP^\sep = \bigcup_{m \geq 1} T_m^\sep, \quad \cP^\mix = \bigcup_{m \geq 1} T_m^\mix.
\end{align}
We define the replica supports
\begin{align}
    \supp_b(\tau) = \bigcup_{j\,:\,\mu(j)=b} \supp(X_j(\tau)),
    \qquad b\in\{1,2\},
\end{align}
and the total support on $\{1,2\}\times[n]$ by
\begin{align}
    \supp(\tau) =
    \lr{\{1\}\times \supp_1(\tau)}\cup
    \lr{\{2\}\times \supp_2(\tau)},
\end{align}
which implies
\begin{align}
    |\supp(\tau)| = |\supp_1(\tau)| + |\supp_2(\tau)|.
\end{align}
Two polymers are compatible ($P \sim Q$) iff their supports are disjoint, equivalently iff
\begin{align}
    \supp(P)\cap \supp(Q)=\emptyset
    \iff
    \supp_1(P)\cap \supp_1(Q)=\emptyset
    \text{ and }
    \supp_2(P)\cap \supp_2(Q)=\emptyset.
\end{align}
As a direct consequence of Lemma~\ref{lem:polymer}, we have the following.

\begin{corollary}[Disconnected polymer representation]\label{cor:polymer2}
    For polymer $P = (\eta, \pi, \mu, \vec I, \vec a)$, define
    \begin{align}\label{eq:z2}
        z_P(\beta, \gamma) = \frac{\sigma^{|\eta^{-1}(0)|} \lambda^{|\eta^{-1}(1)|}}{\lr{|\eta^{-1}(0)|+|\eta^{-1}(1)|}!} (-\beta)^{|\mu^{-1}(1)|} (-\gamma)^{|\mu^{-1}(2)|} \lr{\prod_{j:\eta(j)=1} c_{a_j}} \tr(W_1(P)) \tr(W_2(P)),
    \end{align}
    we have
    \begin{align}
        \EE{\wh Z(\beta)} \cdot \EE{\wh Z(\gamma)} &= \sum_{r=0}^\infty \frac{1}{r!} \sum_{\substack{P_1, \dots, P_r \in \cP^\sep\\ P_i \sim P_j}} \prod_{u=1}^r z_{P_u}(\beta, \gamma) \label{eq:zzsep}.
    \end{align}
\end{corollary}

We now address connected components of $G^\mix$ similarly to Lemma~\ref{lem:decomp}.

\begin{lemma}[Decomposition into mixed connected components]\label{lem:decomp2}
    Let $\tau = (\eta, \pi, \mu, \vec I, \vec a) \in T_m(\ell_1, \ell_2)$ and let $C_1, \dots, C_r \subseteq [m]$ be the connected components of $G^\mix(\tau)$. Then
    \begin{align}
        &\sigma^{|\eta^{-1}(0)|} \lambda^{|\eta^{-1}(1)|} \tr(W_1(\tau)) \tr(W_2(\tau)) \nonumber\\
        &= \prod_{u=1}^r \sigma^{|\{j\in C_u : \eta(j)=0\}|} \lambda^{|\{j\in C_u : \eta(j)=1\}|} \tr(\prod_{j\in C_u:\mu(j)=1} X_j(\tau)) \tr(\prod_{j\in C_u:\mu(j)=2} X_j(\tau)),
    \end{align}
    where the product preserves the ordering of $\vec I, \vec a$.
\end{lemma}
\begin{proof}
    Fix distinct connected components $C_u,C_v$ and a replica label $b\in\{1,2\}$. We have that
    \begin{align}
        \left(\bigcup_{j\in C_u:\mu(j)=b}\supp(X_j(\tau))\right) \cap \left(\bigcup_{j\in C_v:\mu(j)=b}\supp(X_j(\tau))\right) = \emptyset,
    \end{align}
    otherwise an edge would connect $C_u, C_v$. Since every $X_j(\tau)$ is an even Majorana string, operators with disjoint supports commute; hence,
    \begin{align}
        W_b(\tau)=\prod_{j:\mu(j)=b} X_j(\tau) = \prod_{u=1}^r \prod_{j\in C_u:\mu(j)=b} X_j(\tau),
    \end{align}
    where the relative order of indices is preserved.
    Because the supports of the component-blocks are pairwise disjoint within each replica, normalized trace factorizes over the decomposition; factorizing the prefactor $\sigma^{|\eta^{-1}(0)|}\lambda^{|\eta^{-1}(1)|}$ thus gives the claimed result
    \begin{align}
        &\sigma^{|\eta^{-1}(0)|} \lambda^{|\eta^{-1}(1)|} \tr(W_1(\tau)) \tr(W_2(\tau)) \nonumber\\
        &= \prod_{u=1}^r \sigma^{|\{j\in C_u : \eta(j)=0\}|} \lambda^{|\{j\in C_u : \eta(j)=1\}|} \tr(\prod_{j\in C_u:\mu(j)=1} X_j(\tau)) \tr(\prod_{j\in C_u:\mu(j)=2} X_j(\tau)).
    \end{align}
\end{proof}

\begin{lemma}[Connected polymer representation]\label{lem:polymer2}
    \begin{align}
        \EE{\wh Z(\beta) \wh Z(\gamma)} &= \sum_{r=0}^\infty \frac{1}{r!} \sum_{\substack{P_1, \dots, P_r \in \cP^\mix\\ P_i \sim P_j}} \prod_{u=1}^r z_{P_u}(\beta, \gamma) \label{eq:zz},
    \end{align}
    where $z_P$ is given by \eqref{eq:z2}.
    Note that this is the same as \eqref{eq:zzsep}, except the sum is over $\cP^\mix$ instead of $\cP^\sep$.
\end{lemma}
\begin{proof}
    Comparing the coefficient of $\beta^{\ell_1}\gamma^{\ell_2}$ in \eqref{eq:zz} with the Taylor expansion of $\EE{\wh Z(\beta)\wh Z(\gamma)}$, it suffices to show that
    \begin{align}\label{eq:zzcoeff}
        &\frac{\EE{\tr(H^{\ell_1})\tr(H^{\ell_2})}}{\ell_1!\ell_2!}\nonumber\\
        &=
        \sum_{r=0}^\infty \frac{1}{r!}
        \sum_{\substack{P_1,\dots,P_r\in\cP^\mix\\ P_i\sim P_j\\
        \sum_{u=1}^r |\mu_u^{-1}(1)|=\ell_1,\,
        \sum_{u=1}^r |\mu_u^{-1}(2)|=\ell_2}}
        \prod_{u=1}^r
        \frac{\sigma^{|\eta_u^{-1}(0)|}\lambda^{|\eta_u^{-1}(1)|}}{|P_u|!}
        \left(\prod_{j:\eta_u(j)=1} c_{a_j^{(u)}}\right)
        \tr(W_1(P_u))\tr(W_2(P_u)),
    \end{align}
    where we decomposed polymers as $P_u = (\eta_u,\pi_u,\mu_u,\vec I^{(u)},\vec a^{(u)})$. By \eqref{eq:2rep-moment}, with $m=\ell_1+\ell_2$, we can rewrite the left-hand side as
    \begin{align}
        \frac{\EE{\tr(H^{\ell_1})\tr(H^{\ell_2})}}{\ell_1!\ell_2!}
        =
        \frac{1}{m!}
        \sum_{\tau\in T_m(\ell_1,\ell_2)}
        \sigma^{|\eta^{-1}(0)|}\lambda^{|\eta^{-1}(1)|}
        \left(\prod_{j:\eta(j)=1} c_{a_j}\right)
        \tr(W_1(\tau))\tr(W_2(\tau)).
    \end{align}
    We now apply Lemma~\ref{lem:decomp2} similarly to the proof of Lemma~\ref{lem:polymer}. Let $C_1,\dots,C_r$ be the connected components of $G^\mix(\tau)$, let $m_u=|C_u|$, and let $\rho_u:[m_u]\to C_u$ be the unique increasing bijection. Then
    \begin{align}\label{eq:hh}
        \frac{\EE{\tr(H^{\ell_1})\tr(H^{\ell_2})}}{\ell_1!\ell_2!} &= \frac{1}{m!} \sum_{\tau\in T_m(\ell_1,\ell_2)} \prod_{u=1}^r \sigma^{|\{j\in C_u:\eta(j)=0\}|} \lambda^{|\{j\in C_u:\eta(j)=1\}|} \left(\prod_{j\in C_u:\eta(j)=1} c_{a_j}\right)\nonumber\\
        &\quad \times 
            \tr(\prod_{a\in[m_u]:\,\mu(\rho_u(a))=1} X_{\rho_u(a)}(\tau))
            \tr(\prod_{a\in[m_u]:\,\mu(\rho_u(a))=2} X_{\rho_u(a)}(\tau)).
    \end{align}
    As before, each connected component determines a polymer $P_u = (\eta_u,\pi_u,\mu_u,\vec I^{(u)},\vec a^{(u)}) \in \cP^\mix$ defined by restricting to the connected component. Explicitly, we set
    \begin{align}
        \eta_u(a)=\eta(\rho_u(a)),
        \qquad
        \mu_u(a)=\mu(\rho_u(a)).
    \end{align}
    If $\eta_u(a)=0$, we set $I_a^{(u)}=I_{\rho_u(a)}$, while if $\eta_u(a)=1$, we set $a_a^{(u)}=a_{\rho_u(a)}$. Finally,
    \begin{align}
        \pi_u
        =
        \left\{
            \{\rho_u^{-1}(b),\rho_u^{-1}(c)\}
            \,:\,
            \{b,c\}\in\pi,\ b,c\in C_u
        \right\}.
    \end{align}
    Since $C_u$ is connected in $G^\mix(\tau)$, we have $P_u\in T_{m_u}^\mix\subseteq \cP^\mix$. Moreover, if $u\neq v$, then $P_u\sim P_v$, and thus an edge would join $C_u$ and $C_v$.

    To rewrite the sum over $T_m(\ell_1,\ell_2)$, we show that this construction is invertible up to the ordering of the $r$-tuple. Let $m_1,\dots,m_r\geq 1$ satisfy $\sum_{u=1}^r m_u=m$, let $C_1,\dots,C_r\subseteq [m]$ be disjoint with
    \begin{align}
        \bigsqcup_{u=1}^r C_u=[m],
        \qquad
        |C_u|=m_u,
    \end{align}
    and let $(P_1,\dots,P_r)$ satisfy $P_u\in T_{m_u}^\mix$ and $P_u\sim P_v$ for $u\neq v$. Write $P_u=(\eta_u,\pi_u,\mu_u,\vec I^{(u)},\vec a^{(u)})$ and let $\rho_u:[m_u]\to C_u$ be the unique increasing bijection. We construct $\tau=(\eta,\pi,\mu,\vec I,\vec a)\in T_m$ by setting $\eta(\rho_u(a))=\eta_u(a)$ and $\mu(\rho_u(a))=\mu_u(a)$; if $\eta_u(a)=0$, we set $I_{\rho_u(a)}=I_a^{(u)}$, while if $\eta_u(a)=1$, we set $a_{\rho_u(a)}=a_a^{(u)}$. 
    Finally, we define
    \begin{align}\label{eq:choices}
        \pi
        =
        \bigsqcup_{u=1}^r
        \left\{
            \{\rho_u(a),\rho_u(b)\}
            \,:\,
            \{a,b\}\in \pi_u
        \right\}.
    \end{align}
    By construction, whenever $\{x,y\}\in\pi$, both indices come from the same polymer $P_u$, and hence $I_x=I_y$. Since
    \begin{align}
        |\mu^{-1}(1)|=\sum_{u=1}^r |(\mu_u)^{-1}(1)|,
        \qquad
        |\mu^{-1}(2)|=\sum_{u=1}^r |(\mu_u)^{-1}(2)|,
    \end{align}
    we have that $\tau\in T_m(\ell_1,\ell_2)$.

    We next check that the connected components of $G^\mix(\tau)$ are exactly $C_1,\dots,C_r$. First, each $C_u$ is connected. Indeed, since $P_u\in T_{m_u}^\mix$, its mixed graph is connected; every edge in $G^\mix(P_u)$ induces an edge in $G^\mix(\tau)$ after applying $\rho_u$, because:
    \begin{itemize}
        \item if $\mu_u(a)=\mu_u(b)$ and $\supp(X_a(P_u))\cap \supp(X_b(P_u))\neq\emptyset$, then
        \begin{align}
            \mu(\rho_u(a))=\mu(\rho_u(b)),
            \qquad
            \supp(X_{\rho_u(a)}(\tau))\cap \supp(X_{\rho_u(b)}(\tau))\neq\emptyset,
        \end{align}
        so $\rho_u(a)$ and $\rho_u(b)$ are adjacent by the first edge condition;
        \item if $\{a,b\}\in \pi_u$, $\eta_u(a)=\eta_u(b)=0$, and $\mu_u(a)\neq\mu_u(b)$, then
        \begin{align}
            \{\rho_u(a),\rho_u(b)\}\in\pi,
            \qquad
            \eta(\rho_u(a))=\eta(\rho_u(b))=0,
            \qquad
            \mu(\rho_u(a))\neq\mu(\rho_u(b)),
        \end{align}
        so $\rho_u(a)$ and $\rho_u(b)$ are adjacent by the second edge condition.
    \end{itemize}
    Hence $C_u$ is connected in $G^\mix(\tau)$.

    Conversely, there are no edges between distinct $C_u$ and $C_v$. By \eqref{eq:choices}, no pair in $\pi$ joins indices from different components, so no edge of the second type can connect $C_u$ and $C_v$. For the first type, fix $b\in\{1,2\}$. Since $P_u\sim P_v$, we have
    \begin{align}
        \supp_b(P_u)\cap \supp_b(P_v)=\emptyset.
    \end{align}
    Therefore no $x\in C_u$ and $y\in C_v$ with $\mu(x)=\mu(y)=b$ can satisfy
    \begin{align}
        \supp(X_x(\tau))\cap \supp(X_y(\tau))\neq\emptyset.
    \end{align}
    So no edge of the first type connects distinct components. Thus $C_1,\dots,C_r$ are exactly the connected components of $G^\mix(\tau)$.

    This bijection, up to the ordering of the $r$-tuple, rewrites \eqref{eq:hh} as
    \begin{align}
        \frac{\EE{\tr(H^{\ell_1})\tr(H^{\ell_2})}}{\ell_1!\ell_2!}
        &=
        \frac{1}{m!}
        \sum_{r=0}^\infty \frac{1}{r!}
        \sum_{\substack{m_1+\cdots+m_r=m\\ m_u\geq 1}}
        \sum_{\substack{C_1,\dots,C_r\subseteq [m]\\ \bigsqcup_u C_u=[m]\\ |C_u|=m_u}}
        \sum_{\substack{P_1,\dots,P_r\in\cP^\mix\\ P_u\sim P_v,\ u\neq v\\ |P_u|=m_u\\
        \sum_u |(\mu_u)^{-1}(1)|=\ell_1,\,
        \sum_u |(\mu_u)^{-1}(2)|=\ell_2}}\nonumber\\
        &\qquad \prod_{u=1}^r \sigma^{|\eta_u^{-1}(0)|}\lambda^{|\eta_u^{-1}(1)|} \left(\prod_{j:\eta_u(j)=1} c_{a_j^{(u)}}\right)
        \tr(W_1(P_u))\tr(W_2(P_u)).
    \end{align}
    Replacing the sum over $C_1,\dots,C_r$ by the counting factor $\binom{m}{m_1,\dots,m_r}$ gives
    \begin{align}
        \frac{\EE{\tr(H^{\ell_1})\tr(H^{\ell_2})}}{\ell_1!\ell_2!}
        &=
        \sum_{r=0}^\infty \frac{1}{r!}
        \sum_{\substack{m_1+\cdots+m_r=m\\ m_u\geq 1}}
        \sum_{\substack{P_1,\dots,P_r\in\cP^\mix\\ P_u\sim P_v,\ u\neq v\\ |P_u|=m_u\\
        \sum_u |(\mu_u)^{-1}(1)|=\ell_1,\,
        \sum_u |(\mu_u)^{-1}(2)|=\ell_2}}
        \prod_{u=1}^r
        \frac{\sigma^{|\eta_u^{-1}(0)|}\lambda^{|\eta_u^{-1}(1)|}}{m_u!}\nonumber\\
        &\quad \times
        \left(\prod_{j:\eta_u(j)=1} c_{a_j^{(u)}}\right)
        \tr(W_1(P_u))\tr(W_2(P_u)).
    \end{align}
    Combining the sums over $m_1+\cdots+m_r=m$ and over the polymers $P_1,\dots,P_r$ gives \eqref{eq:zzcoeff}, proving \eqref{eq:zz}.
\end{proof}

\subsection{Koteck\'y-Preiss}
We now apply Koteck\'y-Preiss (Theorem~\ref{thm:kp}) using the quantity
\begin{align}
    a(P) = \frac{|\supp(P)|}{4\alpha}
\end{align}
for some constant $\alpha$ satisfying $\alpha \geq \max\{q,L\}$. We will apply it separately for the disconnected and connected contributions. To ensure that we can show a KP condition $\leq a(P)$ (Theorem~\ref{thm:kp}) when adding both contributions together, we will use a stronger bound on the disconnected piece that keeps it $\leq \frac{1}{2}a(P)$.

\begin{lemma}[Disconnected condition]\label{lem:disccrit}
    For $t = \max\{|\beta|, |\gamma|\}$, let
    \begin{align}
        \Xi^\sep(t) = \sup_{s\in\{1,2\}, \, x \in [n]} \sum_{\substack{P \in \cP^\sep\\ x \in \supp_s(P)}} |z_P(\beta, \gamma)| e^{a(P)}
    \end{align}
    If
    \begin{align}\label{eq:discond}
        \frac{\alpha e^{q/4\alpha}}{2^{q+1}}t^2 + \alpha e^{L/4\alpha}\Gamma|\lambda|t \leq \frac{1}{4e^{1/4}}
    \end{align}
    then $\Xi^\sep(t) \leq \frac{1}{8\alpha}$.
\end{lemma}
\begin{proof}
    The proof follows Lemmas~\ref{lem:annealcrit} and~\ref{lem:annealcritsat} exactly since the individual replicas do not interact. There are only two differences: first, we replace the exponential weight $e^{|\supp(P)|/2\alpha}$ by $e^{|\supp(P)|/4\alpha}$; second, we use an upper bound of $1/4e^{1/4}$ instead of $1/2\sqrt e$. By Fact~\ref{fac:tfunc}, this second modification gives an upper bound of $\Xi^\sep(t) \leq 1/8\alpha$ instead of $\Xi^\sep(t) \leq 1/4\alpha$.
\end{proof}

\begin{corollary}[Disconnected Koteck\'y-Preiss]\label{cor:kpdisc}
If \eqref{eq:discond} holds, then for every $P \in \cP^\mix$,
\begin{align}
    \sum_{\substack{Q \in \cP^\sep\\ Q \nsim P}} |z_Q(\beta,\gamma)|e^{a(Q)} \leq \frac{1}{2} a(P).
\end{align}
\end{corollary}
\begin{proof}
    If $Q \nsim P$, then $\supp(Q) \cap \supp(P) \neq \emptyset$, so
    \begin{align}
        \sum_{\substack{Q \in \cP^\sep\\ Q \nsim P}} |z_Q(\beta,\gamma)|e^{a(Q)} \leq \sum_{(s,x) \in \supp(P)} \sum_{\substack{Q \in \cP^\sep\\ x \in \supp_s(Q)}} |z_Q(\beta,\gamma)| e^{a(Q)} \leq |\supp(P)| \Xi^\sep(t) \leq \frac{1}{2} a(P)
    \end{align}
    by Lemma~\ref{lem:disccrit}.
\end{proof}

The nontrivial terms are on the remaining polymers. For $|\beta|, |\gamma| \leq R$, define
\begin{align}
    \cM(\beta,\gamma) = \sum_{P \in \cP^\mix\setminus \cP^\sep} |z_P(\beta, \gamma)| e^{a(P)} = \cM^\dist(\beta,\gamma) + \cM^\rep(\beta,\gamma).
\end{align}
A polymer $P\in \cP^\mix\setminus \cP^\sep$ belongs to $\cM^\dist$ if it contains at least one mixed Wick pair whose support appears nowhere else in the polymer, neither as another Gaussian support nor as one of the perturbation supports $A_{a_j}$. We can write this formally as follows. For a polymer $P=(\eta,\pi,\mu,\vec I,\vec a)$, let
\begin{align}
    \Pi^\mix(P) = \{\{u,v\}\in \pi : \mu(u)\neq \mu(v)\}
\end{align}
denote the set of mixed Wick pairs, and for $e=\{u,v\}\in \Pi^\mix(P)$ define its support
\begin{align}
    K_e = I_u = I_v \in \binom{[n]}{q}.
\end{align}
We then define $\cP^\dist$ as
\begin{align}\label{eq:pdist}
    \Big\{
        P \in \cP^\mix \setminus \cP^\sep
        \,:\,
        \exists\, e\in \Pi^\mix(P) \text{ s.t. }
        K_e \neq K_{e'} \ \forall\, e'\in \pi\setminus\{e\},
        \;
        K_e \neq A_{a_j}\ \forall\, j \text{ with }\eta(j)=1
    \Big\},
\end{align}
and $\cP^\rep$ contains the remaining polymers in $\cP^\mix \setminus \cP^\sep$.

\begin{lemma}[Distinct set bound]\label{lem:dist}
    For $t = \max\{|\beta|, |\gamma|\}$ and $\alpha \geq \max\{q,L\}$, if
    \begin{align}\label{eq:concond1}
        \frac{e^{1/2}\alpha q t^2}{2^q} + e^{1/4}\Gamma \alpha q |\lambda| t \leq \frac{1}{e}
    \end{align}
    then
    \begin{align}
        \cM^\dist(\beta, \gamma) \leq \frac{(q-1)!}{\alpha^2q} n^{1-q/2}.
    \end{align}
\end{lemma}
\begin{proof}
    We write out $\cM^\dist$ using Lemma~\ref{lem:polymer2} as
    \begin{align}
        &\cM^\dist(\beta, \gamma) = \sum_{P\in\cP^\dist}\abs{z_P(\beta,\gamma)}e^{a(P)}\\
        &\leq \sum_{P\in\cP^\dist} e^{a(P)} \frac{\sigma^{|\eta^{-1}(0)|} \abs{\lambda}^{|\eta^{-1}(1)|}}{\lr{|\eta^{-1}(0)|+|\eta^{-1}(1)|}!} \abs{\beta}^{|\mu^{-1}(1)|} \abs{\gamma}^{|\mu^{-1}(2)|} \lr{\prod_{j:\eta(j)=1} c_{a_j}} \abs{\tr(W_1(P)) \tr(W_2(P))},
    \end{align}
    where the polymer is decomposed as $P=(\eta, \pi, \mu, \vec I, \vec a)$. Suppose $e \in \Pi^\mix(P)$ is a Wick pair that certifies that $P \in \cP^\dist$ per \eqref{eq:pdist}. Similarly to before, set notation
    \begin{align}
        \abs{\eta^{-1}(0)} = 2k, \quad \abs{\eta^{-1}(1)} = r.
    \end{align}
    For each edge $\{u, v\} \in \pi$, recall that $I_u = I_v$; label these sets by $K_1, \dots, K_k \in \binom{[n]}{q}$ such that the edge $e$ identified above corresponds to $K_1$. Note that by cyclicity of trace, we can rewrite
    \begin{align}
        \tr(W_1(P)) = \tr(\psi_{K_1}O_1), \quad \tr(W_2(P)) = \tr(\psi_{K_1}O_2)
    \end{align}
    where $O_1, O_2$ depend only on $(\eta, \pi, \mu, e, K_2, \dots, K_k, \vec a)$ and not $K_1$.
    
    Let $G(K_1; K_2, \dots, K_k, \vec a)$ denote the support-intersection graph on $K_1, \dots, K_k, A_{a_1}, \dots, A_{a_r}$. Since this graph is connected in the polymer expansion, we have the bound for $t = \max\{|\beta|, |\gamma|\}$ that
    \begin{align}\label{eq:mbound1}
        \cM^\dist(\beta, \gamma) &\leq \sum_{k \geq 1} \sum_{r \geq 0} e^{(2qk+Lr)/4\alpha} \frac{(\sigma t)^{2k} (|\lambda| t)^r}{(2k+r)!} \sum_{\substack{(\eta, \pi, \mu, e)\\ \abs{\eta^{-1}(0)} = 2k, \abs{\eta^{-1}(1)}=r\\ e\in\Pi^\mix(\eta,\pi,\mu)}} \sum_{a_1,\dots,a_r \in [M]} \lr{\prod_{j=1}^r \abs{c_{a_j}}}\nonumber\\
        &\quad \sum_{K_2,\dots,K_k\in\binom{[n]}{q}} \sum_{K_1\in\binom{[n]}{q}} 1\{G(K_1;K_2,\dots,K_k,\vec a) \text{ connected}\} \abs{\tr(\psi_{K_1}O_1) \tr(\psi_{K_1}O_2)},
    \end{align}
    where we used
    \begin{align}
        a(P) = \frac{|\supp(P)|}{4\alpha} = \frac{|\supp_1(P)| + |\supp_2(P)|}{4\alpha} \leq \frac{2qk+Lr}{4\alpha}.
    \end{align}
    Let $G(\emptyset; K_2, \dots, K_k, \vec a)$ be the intersection graph on the remaining objects; denote its connected components by $R_1, \dots, R_c$. We claim that for every term with
    \begin{align}
        1\{G(K_1;K_2,\dots,K_k,\vec a) \text{ connected}\} \abs{\tr(\psi_{K_1}O_1) \tr(\psi_{K_1}O_2)}
    \end{align}
    the number of connected components satisfies $c \leq q/2$.

    This claim can be shown starting from two observations: first, $\tr(\psi_{K_1} O_b)$ is only nonzero when $O_b$ is supported only on $K_1$; second, the support of each $R_j$ must intersect $K_1$, since $G(K_1; K_2, \dots, K_k, \vec a)$ was connected. Fix $b \in \{1,2\}$. We decompose $O_b$ over its supports in the different connected components $R_j$; these supports are necessarily disjoint, otherwise the $R_j$s wouldn't be distinct connected components. Accordingly, for some scalars $\omega_{j,b}$, let
    \begin{align}
        O_b = \prod_{j=1}^c \omega_{j,b} \psi_{B_{j,b}}
    \end{align}
    where $B_{j,b} \cap B_{j',b} = \emptyset$ for all $j \neq j'$. Hence, we have that
    \begin{align}
        K_1 = \supp(O_b) = \bigsqcup_{j=1}^c B_{j,b}
    \end{align}
    by the first observation. Moreover, since the support of a product of Majorana strings is given by the symmetric difference of each string's support, and $|K_i|, |A_{a_j}|$ are both even, the support of $B_{j,b}$ also has even cardinality; hence, $|B_{j,b}| \geq 2$. Mutual disjointness thus gives
    \begin{align}
        q = |K_1| = \sum_{j=1}^c |B_{j,b}| \geq 2c.
    \end{align}
    We use this fact to loosen the indicator on the connectedness of $G$ into the connectedness of $R_j$. To notate the sum over $R_j$, we introduce a map $\phi$ that labels the connected components as $R_j = \phi^{-1}(j)$ for $j \in [c]$. Hence, the last sum in \eqref{eq:mbound1} is bounded by
    \begin{align}
        \sum_{K_1\in\binom{[n]}{q}} 1\{G(K_1;K_2,\dots,K_k,\vec a) \text{ connected}\} \abs{\tr(\psi_{K_1}O_1) \tr(\psi_{K_1}O_2)} \leq 2^{-qk} \sum_{c=1}^{q/2} \sum_\phi \prod_{j=1}^c 1\{R_j(\phi) \text{ connected}\}
    \end{align}
    due to the bound
    \begin{align}
        \abs{\tr(\psi_{K_1}O_1)\tr(\psi_{K_1}O_2)} \leq 2^{-q} \norm{O_1} \cdot \norm{O_2} \leq 2^{-qk},
    \end{align}
    applying the Majorana normalization convention and the fact that each $O_1$ is a string of $(k-1)q$ fermions. We now bound
    \begin{align}
        \cM^\dist(\beta, \gamma) &\leq \sum_{k \geq 1} \sum_{r \geq 0} e^{(2qk+Lr)/4\alpha} \frac{(\sigma t)^{2k} (|\lambda| t)^r}{(2k+r)!} \sum_{\substack{(\eta, \pi, \mu, e)\\ \abs{\eta^{-1}(0)} = 2k, \abs{\eta^{-1}(1)}=r\\ e\in\Pi^\mix(\eta,\pi,\mu)}} \sum_{a_1,\dots,a_r \in [M]} \lr{\prod_{j=1}^r \abs{c_{a_j}}}\nonumber\\
        &\quad \sum_{K_2,\dots,K_k\in\binom{[n]}{q}} 2^{-qk} \sum_{c=1}^{q/2} \sum_\phi \prod_{j=1}^c 1\{R_j(\phi) \text{ connected}\}.
    \end{align}
    For a given choice of $\phi$, let $g_u$ denote the number of $K$-sets in $R_u$ and let $h_u$ denote the number of $A$-sets in $R_u$. Since
    \begin{align}
        \sum_{u=1}^c g_u = k-1, \quad \sum_{u=1}^c h_u = r
    \end{align}
    we have by Lemma~\ref{lem:count}(1), summing over the $n$ choices for the root $x$ and weighting by $\Gamma$, that
    \begin{align}
        &\sum_{a_1,\dots,a_r} \lr{\prod_{j=1}^r |c_{a_j}|} \sum_{K_2,\dots,K_k} \prod_{u=1}^c 1\{R_u(\phi)\text{ connected}\} \nonumber\\
        &\leq n^c \Gamma^r \binom{n-1}{q-1}^{k-1} \prod_{u=1}^c (g_u q + h_u L)^{g_u + h_u - 1} \leq n^c \Gamma^r \binom{n-1}{q-1}^{k-1} [\alpha(k+r)]^{k+r-1-c}.
    \end{align}
    Since the number of surjections $\phi$ is at most $c^{k+r-1}$, this gives
    \begin{align}
        \cM^\dist(\beta, \gamma) &\leq \sum_{k \geq 1} \sum_{r \geq 0} e^{(2qk+Lr)/4\alpha} \frac{(\sigma t)^{2k} (|\lambda| t)^r}{(2k+r)!} \sum_{\substack{(\eta, \pi, \mu, e)\\ \abs{\eta^{-1}(0)} = 2k, \abs{\eta^{-1}(1)}=r\\ e\in\Pi^\mix(\eta,\pi,\mu)}} \nonumber\\
        &\quad 2^{-qk} \Gamma^r \binom{n-1}{q-1}^{k-1} [\alpha(k+r)]^{k+r-1} \sum_{c=1}^{q/2} c^{k+r-1} \frac{n}{\alpha(k+r)}^c.
    \end{align}
    Using upper bound
    \begin{align}
        \sum_{c=1}^{q/2} c^{k+r-1} \frac{n}{\alpha(k+r)}^c \leq \lr{\frac{q}{2}}^{k+r-1} \frac{n^{q/2}}{\alpha(k+r)}
    \end{align}
    and counting the $\binom{2k+r}{r}$ choices for $\eta$, $(2k-1)!!$ choices for $\pi$, and $2^{2k+r-1}$ choices for $\mu$, and $k$ choices for $e$, we have
    \begin{align}
        \cM^\dist(\beta, \gamma) &\leq \sum_{k \geq 1} \sum_{r \geq 0} e^{(2qk+Lr)/4\alpha} \frac{(\sigma t)^{2k} (|\lambda| t)^r}{(2k+r)!} \nonumber\\
        &\quad \times \binom{2k+r}{r} (2k-1)!! 2^{2k+r-1} k 2^{-qk} n^{q/2} \lr{\frac{q}{2}}^{k+r-1} \Gamma^r \binom{n-1}{q-1}^{k-1} [\alpha(k+r)]^{k+r-2}.
    \end{align}
    Using bounds
    \begin{align}
        \binom{2k+r}{r}(2k-1)!! \frac{2^{2k+r-1}}{(2k+r)!} = \frac{2^{k+r-1}}{k!r!}, \quad \binom{n-1}{q-1}\sigma^2 \leq 1, \quad (k+r)^{k+r-2}k \leq (k+r)^{k+r-1}
    \end{align}
    we have
    \begin{align}
        \cM^\dist(\beta, \gamma) &\leq \frac{\sigma^2 n^{q/2}}{\alpha^2q} \sum_{k \geq 1} \sum_{r \geq 0} e^{(2qk+Lr)/4\alpha} \frac{(k+r)^{k+r-1}}{k!r!} \lr{2^{-q} \alpha q t^2}^k \lr{\Gamma \alpha q |\lambda| t}^r.
    \end{align}
    We then use
    \begin{align}
        \sum_{\substack{k,r \geq 0\\ k+r \geq 1}} \frac{(k+r)^{k+r-1}}{k!r!} y_1^k y_2^r \leq T(y_1 + y_2)
    \end{align}
    and $\sigma^2 n^{q/2} = (q-1)!/n^{q/2-1}$ to obtain
    \begin{align}
        \cM^\dist(\beta, \gamma) &\leq \frac{(q-1)!}{\alpha^2q} n^{1-q/2} T\lr{e^{q/2\alpha}2^{-q} \alpha q t^2 + e^{L/4\alpha}\Gamma \alpha q |\lambda| t}.
    \end{align}
    Finally, Fact~\ref{fac:tfunc} gives
    \begin{align}
        \cM^\dist(\beta, \gamma) \leq \frac{(q-1)!}{\alpha^2q} n^{1-q/2}
    \end{align}
    for all
    \begin{align}
        \frac{e^{1/2}\alpha q t^2}{2^q} + e^{1/4}\Gamma \alpha q |\lambda| t \leq \frac{1}{e}.
    \end{align}
\end{proof}

\begin{lemma}[Repeated set bound]\label{lem:rep}
    For $t = \max\{|\beta|, |\gamma|\}$ and $\alpha \geq \max\{q,L\}$, if
    \begin{align}\label{eq:concond2}
        \frac{\alpha e^{1/2}}{2^{q-1}} t^2  + 2 \alpha e^{1/4} \Gamma |\lambda| t \leq \frac{1}{e},
    \end{align}
    then
    \begin{align}
        \cM^\rep(\beta, \gamma) \leq \frac{(q-1)!}{\alpha^2} n^{2-q}.
    \end{align}
\end{lemma}
\begin{proof}
    Let $K_1$ denote a mixed Wick pair support, i.e., $K_1 = I_u = I_v$ for $\{u, v\} \in \pi$ and $\mu(u) \neq \mu(v)$. As in the proof of Lemma~\ref{lem:dist}, we set $|\eta^{-1}(0)|=2k$ and $|\eta^{-1}(1)|=r$. In $\cM^\rep$, there exists another set $B \in \{K_2, \dots, K_k, A_{a_1}, \dots, A_{a_r}\}$ such that $K_1 = B$. We consider two cases: $B \in \{K_2, \dots, K_k\}$ is from the SYK terms, and $B \in \{A_{a_1}, \dots, A_{a_r}\}$ is from the perturbation terms. In both cases, it will suffice to use the bound
    \begin{align}\label{eq:wbound}
        \abs{\tr(W_1(P)) \tr(W_2(P))} \leq 2^{-qk},
    \end{align}
    which holds since $\norm{\psi_I} = 2^{-q/2}$ and $\norm{\psi_{A_{a_j}}} \leq 1$.
    
    In the first case, there are at most $k-1$ choices $e \in \{2, \dots, k\}$ to set $K_1 = K_e$. Lemma~\ref{lem:count}(3) with $K_1=K_e$ and fixed $e$ gives at most
    \begin{align}
        q \binom{n-1}{q-1}^{k-2} (qk+Lr)^{k+r-3}
    \end{align}
    choices for the remaining $K$ sets and $A_{a_j}$. As in Lemma~\ref{lem:dist}, there are $\binom{2k+r}{r}$ ways to choose $\eta$, $(2k-1)!!$ ways to choose $\pi$, and $2^{2k+r-1}$ ways to choose $\mu$ (since, as before, the pair $\{u,v\} \in \pi$ corresponding to $K_1$ must have $\mu(u) \neq \mu(v)$). Hence, the contribution of the first case to $\cM^\rep$ is at most
    \begin{align}
        &\sum_{k \geq 2} \sum_{r \geq 0} (k-1) \binom{n}{q} \cdot q \binom{n-1}{q-1}^{k-2} \Gamma^r (qk+Lr)^{k+r-3}\nonumber\\
        &\qquad\qquad\cdot \binom{2k+r}{r} (2k-1)!! 2^{2k+r-1} \cdot \frac{(\sigma t)^{2k}(|\lambda|t)^r}{(2k+r)!} 2^{-qk} e^{(2qk+Lr)/4\alpha}.
    \end{align}
    We further loosen this upper bound with $(k-1)(qk+Lr)^{k+r-3} \leq (qk+Lr)^{k+r-2}$.

    In the second case, there are at most $r$ choices for $f$ such that $A_{a_f} = K_1$. Lemma~\ref{lem:count}(4) bounds the number of remaining $K_2, \dots, K_k$ and $A_{a_j}$ with $j \neq f$ as
    \begin{align}
        q \binom{n-1}{q-1}^{k-1} (qk+Lr)^{k+r-3}.
    \end{align}
    We use \eqref{eq:wbound} once again to get a contribution of $\cM^\rep$ of at most
    \begin{align}
        \sum_{k \geq 1} \sum_{r \geq 1} r \cdot q \binom{n-1}{q-1}^{k-1} \Gamma^r (qk+Lr)^{k+r-3} \cdot \binom{2k+r}{r} (2k-1)!! 2^{2k+r-1} \cdot \frac{(\sigma t)^{2k}(|\lambda|t)^r}{(2k+r)!} 2^{-qk} e^{(2qk+Lr)/4\alpha}.
    \end{align}
    We further loosen this upper bound with $r(qk+Lr)^{k+r-3} \leq (qk+Lr)^{k+r-2}$ and $\binom{n-1}{q-1}^{k-1} \leq \binom{n}{q} \binom{n-1}{q-1}^{k-2}$.
    
    Combining the two cases (further loosening the upper bounds by extending the sums over $k$ and $r$), we find
    \begin{align}
        \cM^\rep &\leq 2\sum_{k \geq 1} \sum_{r \geq 0} q\binom{n}{q}\binom{n-1}{q-1}^{k-2} \Gamma^r (qk+Lr)^{k+r-2} \cdot \binom{2k+r}{r} (2k-1)!! 2^{2k+r-1} \nonumber\\
        &\quad \cdot \frac{(\sigma t)^{2k}(|\lambda|t)^r}{(2k+r)!} 2^{-qk} e^{(2qk+Lr)/4\alpha}.
    \end{align}
    Using for $k \geq 1$ that
    \begin{align}
        \binom{2k+r}{r}(2k-1)!! \frac{2^{2k+r-1}}{(2k+r)!} = \frac{2^{k+r-1}}{k!r!}, \quad q\binom{n}{q}\binom{n-1}{q-1}^{k-2}\sigma^{2k} \leq (q-1)!n^{2-q},
    \end{align}
    we obtain
    \begin{align}
        \cM^\rep \leq (q-1)!n^{2-q} \sum_{k \geq 1} \sum_{r \geq 0} \frac{(qk+Lr)^{k+r-2}}{k!r!} 2^{r-(q-1)k} \Gamma^r (|\lambda|t)^r t^{2k} e^{(2qk+Lr)/4\alpha}.
    \end{align}
    Using the assumed condition $qk+Lr \leq \alpha(k+r)$ and loosening $(k+r)^{k+r-2} \leq (k+r)^{k+r-1}$ gives by Fact~\ref{fac:tfunc}
    \begin{align}
        \cM^\rep &\leq \frac{(q-1)!}{\alpha^2}n^{2-q} \sum_{\substack{k,r \geq 0\\ k+r \geq 1}} \frac{(k+r)^{k+r-1}}{k!r!} \lr{\frac{\alpha e^{q/2\alpha}}{2^{q-1}}t^2}^k \lr{2\alpha e^{L/4\alpha}\Gamma |\lambda| t}^r\\ &= \frac{(q-1)!}{\alpha^2}n^{2-q} T\lr{\frac{\alpha e^{1/2}}{2^{q-1}}t^2 + 2\alpha e^{1/4}\Gamma |\lambda| t}.
    \end{align}
    We conclude by using $0 \leq T(x) \leq 1$ for $x \leq 1/e$.
\end{proof}

\begin{lemma}[Connected Koteck\'y-Preiss]\label{lem:kpconn}
    If \eqref{eq:discond}, \eqref{eq:concond1} and \eqref{eq:concond2} hold, then the KP condition
    \begin{align}
        \sum_{\substack{Q \in \cP^\mix\\ Q \nsim P}} |z_Q(\beta, \gamma)| e^{a(Q)} \leq a(P)
    \end{align}
    holds for every $P \in \cP^\mix$ for sufficiently large $n$.
\end{lemma}
\begin{proof}
    By Corollary~\ref{cor:kpdisc},
    \begin{align}
        \sum_{\substack{Q \in \cP^\sep\\ Q \nsim P}} |z_Q(\beta,\gamma)|e^{a(Q)} \leq \frac{1}{2} a(P)
    \end{align}
    for every $P \in \cP^\mix$. By Lemmas~\ref{lem:dist} and~\ref{lem:rep},
    \begin{align}\label{eq:cmineq}
        \cM(\beta, \gamma) = O\lr{n^{1-q/2} + n^{2-q}} = O\lr{n^{1-q/2}}
    \end{align}
    for $q \geq 4$.
    For the remaining polymers, we have
    \begin{align}
        \sum_{\substack{Q \in \cP^\mix \setminus \cP^\sep \\ Q \nsim P}} |z_Q(\beta, \gamma)| e^{a(Q)} &\leq \sum_{(s,x) \in \supp(P)} \sum_{\substack{Q \in \cP^\mix \setminus \cP^\sep\\ x \in \supp_s(Q)}} |z_Q(\beta,\gamma)| e^{a(Q)} \\
        &\leq |\supp(P)| \cM(\beta, \gamma) \\
        &= 4\alpha \cM(\beta, \gamma) a(P) \\
        &\leq a(P) \cdot O\lr{n^{1-q/2}},
    \end{align}
    where second-to-last line used the definition $a(P) = |\supp(P)|/4\alpha$, and the last inequality applied \eqref{eq:cmineq}.
    Combining this with the separated contribution gives the claimed KP condition, i.e.,
    \begin{align}
        \sum_{\substack{Q \in \cP^\mix\\ Q \nsim P}} |z_Q(\beta, \gamma)| e^{a(Q)} \leq \lr{\frac{1}{2} + O\lr{n^{1-q/2}}}a(P) \leq a(P).
    \end{align}
\end{proof}

\subsection{Proof of Theorem~\ref{thm:sa}}

Set $\gamma = \overline{\beta}$ and assume $|\beta|$ satisfies \eqref{eq:anncond}, \eqref{eq:discond}, \eqref{eq:concond1} and \eqref{eq:concond2}. By Lemma~\ref{lem:kpconn}, the KP condition holds on $\cP^\mix$ with $a(P) = |\supp(P)|/4\alpha$. Since $\cP^\sep \subset \cP^\mix$, we apply Corollary~\ref{cor:kp} with $\Lambda = \cP^\mix$ and $\Lambda_0 = \cP^\sep$ to obtain
\begin{align}
    \left|\log \frac{\E \left|\wh Z(\beta)\right|^2}{\left|\E\,\wh Z(\beta)\right|^2}\right| \leq \sum_{P \in \cP^\mix \setminus \cP^\sep} |z_P(\beta, \overline\beta)| e^{a(P)} = \cM(\beta, \overline\beta) = O\lr{n^{1-q/2}}.
\end{align}
Note that the denominator is nonzero by Theorem~\ref{thm:anneal} because we assume \eqref{eq:anncond}. Since $\left|\log \frac{\E \left|\wh Z(\beta)\right|^2}{\left|\E\,\wh Z(\beta)\right|^2}\right| < 1$ for sufficiently large $n$, the inequality $|x-1| \leq e^{|\log x|}-1 \leq 2|\log x|$ gives
\begin{align}
    \left|\frac{\E \left|\wh Z(\beta)\right|^2}{\left|\E\,\wh Z(\beta)\right|^2} - 1\right| \leq 2\left|\log \frac{\E \left|\wh Z(\beta)\right|^2}{\left|\E\,\wh Z(\beta)\right|^2}\right| = O\lr{n^{1-q/2}},
\end{align}
which produces
\begin{align}
    \frac{\E \left|\wh Z(\beta) - \E \wh Z(\beta)\right|^2}{\left|\E \wh Z(\beta)\right|^2} = O\lr{n^{1-q/2}}.
\end{align}
It remains to obtain a simpler sufficient condition on $|\beta|$ from \eqref{eq:anncond}, \eqref{eq:discond}, \eqref{eq:concond1} and \eqref{eq:concond2}. Set $\alpha = \max\{q,L\}$ and $t = |\beta|$. For $q \geq 4$, \eqref{eq:concond1} implies the other three conditions. Indeed,
\begin{align}
    \frac{\alpha e^{q/2\alpha}}{2^{q+1}}t^2 + \alpha e^{L/2\alpha}\Gamma|\lambda|t
    &\leq \frac{e^{1/4}}{q}\lr{\frac{e^{1/2}\alpha q}{2^q}t^2 + e^{1/4}\Gamma\alpha q |\lambda|t}
    \leq \frac{1}{2\sqrt e},\\
    \frac{\alpha e^{q/4\alpha}}{2^{q+1}}t^2 + \alpha e^{L/4\alpha}\Gamma|\lambda|t
    &\leq \frac{1}{q}\lr{\frac{e^{1/2}\alpha q}{2^q}t^2 + e^{1/4}\Gamma\alpha q |\lambda|t}
    \leq \frac{1}{4e^{1/4}},\\
    \frac{\alpha e^{1/2}}{2^{q-1}}t^2 + 2\alpha e^{1/4}\Gamma |\lambda|t
    &= \frac{2}{q}\lr{\frac{e^{1/2}\alpha q}{2^q}t^2 + e^{1/4}\Gamma\alpha q |\lambda|t}
    \leq \frac{1}{e}.
\end{align}
Hence it suffices to enforce \eqref{eq:concond1}. If $|\lambda|\Gamma \leq 2^{-q/2}(q\alpha)^{-1/2}$ and
\begin{align}
    t \leq \frac{C}{\sqrt{q\alpha}}, \qquad
    C = 2^{q/2-1}e^{-1/4}\lr{\sqrt{1+\frac{4}{e}}-1},
\end{align}
then
\begin{align}
    \frac{e^{1/2}\alpha q}{2^q}t^2 + e^{1/4}\Gamma\alpha q |\lambda|t
    \leq \frac{e^{1/2}}{2^q}C^2 + \frac{e^{1/4}}{2^{q/2}}C = \frac{1}{e}.
\end{align}
This proves that
\begin{align}
    \sup_{|\beta| \leq C/\sqrt{q\max\{q,L\}}} \frac{\E \left|\wh Z(\beta) - \E \wh Z(\beta)\right|^2}{\left|\E \wh Z(\beta)\right|^2} = O\lr{n^{1-q/2}}.
\end{align}

\bibliography{References}

\newcommand{\etalchar}[1]{$^{#1}$}
\begin{thebibliography}{MEAG{\etalchar{+}}20}

\bibitem[ACKK25]{anschuetz2025strongly}
Eric~R Anschuetz, Chi-Fang Chen, Bobak~T Kiani, and Robbie King.
\newblock Strongly interacting fermions are nontrivial yet nonglassy.
\newblock {\em Physical Review Letters}, 135(3):030602, 2025.

\bibitem[AJ19]{auffinger2019thouless}
Antonio Auffinger and Aukosh Jagannath.
\newblock {Thouless--Anderson--Palmer} equations for generic $p$-spin glasses.
\newblock {\em Ann. Probab.}, 47(4):2230--2256, July 2019.

\bibitem[ALR87]{aizenman1987some}
Michael Aizenman, Joel~L Lebowitz, and David Ruelle.
\newblock Some rigorous results on the sherrington-kirkpatrick spin glass
  model.
\newblock {\em Communications in mathematical physics}, 112(1):3--20, 1987.

\bibitem[AMS24]{almheiri2024universal}
Ahmed Almheiri, Alexey Milekhin, and Brian Swingle.
\newblock Universal constraints on energy flow and syk thermalization.
\newblock {\em Journal of High Energy Physics}, 2024(8):1--43, 2024.

\bibitem[Ans25]{anschuetz2025efficient}
Eric~R Anschuetz.
\newblock Efficient learning implies quantum glassiness.
\newblock {\em arXiv preprint arXiv:2505.00087}, 2025.

\bibitem[ART06]{achlioptas2006solution}
Dimitris Achlioptas and Federico Ricci-Tersenghi.
\newblock On the solution-space geometry of random constraint satisfaction
  problems.
\newblock In {\em Proceedings of the thirty-eighth annual ACM symposium on
  Theory of computing}, pages 130--139, 2006.

\bibitem[Bar14]{barvinok2014computing}
Alexander Barvinok.
\newblock Computing the partition function for cliques in a graph.
\newblock {\em arXiv preprint arXiv:1405.1974}, 2014.

\bibitem[Bar16a]{barvinok2016approximating}
Alexander Barvinok.
\newblock Approximating permanents and hafnians.
\newblock {\em arXiv preprint arXiv:1601.07518}, 2016.

\bibitem[Bar16b]{barvinok2016combinatorics}
Alexander Barvinok.
\newblock {\em Combinatorics and complexity of partition functions}, volume~30.
\newblock Springer, 2016.

\bibitem[Bar16c]{barvinok2016computing}
Alexander Barvinok.
\newblock Computing the permanent of (some) complex matrices.
\newblock {\em Foundations of Computational Mathematics}, 16(2):329--342, 2016.

\bibitem[Bar18]{barvinok2018approximating}
Alexander Barvinok.
\newblock Approximating real-rooted and stable polynomials, with combinatorial
  applications.
\newblock {\em arXiv preprint arXiv:1806.07404}, 2018.

\bibitem[BC00]{biroli2000quantum}
Giulio Biroli and Leticia~F Cugliandolo.
\newblock Quantum tap equations.
\newblock {\em arXiv preprint cond-mat/0011028}, 2000.

\bibitem[BCGW21]{bravyi2021complexity}
Sergey Bravyi, Anirban Chowdhury, David Gosset, and Pawel Wocjan.
\newblock On the complexity of quantum partition functions.
\newblock {\em arXiv preprint arXiv:2110.15466}, 2021.

\bibitem[BCL24]{bergamaschi2024quantum}
Thiago Bergamaschi, Chi-Fang Chen, and Yunchao Liu.
\newblock Quantum computational advantage with constant-temperature gibbs
  sampling.
\newblock In {\em 2024 IEEE 65th Annual Symposium on Foundations of Computer
  Science (FOCS)}, pages 1063--1085. IEEE, 2024.

\bibitem[BDOT08]{bravyi2006complexity}
Sergey Bravyi, David~P. Divincenzo, Roberto Oliveira, and Barbara~M. Terhal.
\newblock The complexity of stoquastic local hamiltonian problems.
\newblock {\em Quantum Info. Comput.}, 8(5):361–385, May 2008.

\bibitem[BDPR21]{bencs2021zero}
Ferenc Bencs, Ewan Davies, Viresh Patel, and Guus Regts.
\newblock On zero-free regions for the anti-ferromagnetic potts model on
  bounded-degree graphs.
\newblock {\em Annales de l’Institut Henri Poincar{\'e} D}, 8(3):459--489,
  2021.

\bibitem[BEAH{\etalchar{+}}22]{bandeira2022franz}
Afonso~S Bandeira, Ahmed El~Alaoui, Samuel Hopkins, Tselil Schramm, Alexander~S
  Wein, and Ilias Zadik.
\newblock The franz-parisi criterion and computational trade-offs in high
  dimensional statistics.
\newblock {\em Advances in Neural Information Processing Systems},
  35:33831--33844, 2022.

\bibitem[BFK{\etalchar{+}}13]{bapst2013quantum}
Victor Bapst, Laura Foini, Florent Krzakala, Guilhem Semerjian, and Francesco
  Zamponi.
\newblock The quantum adiabatic algorithm applied to random optimization
  problems: The quantum spin glass perspective.
\newblock {\em Physics Reports}, 523(3):127--205, 2013.

\bibitem[BFK24]{bunin2024fisher}
Guy Bunin, Laura Foini, and Jorge Kurchan.
\newblock Fisher zeroes and the fluctuations of the spectral form factor of
  chaotic systems.
\newblock {\em SciPost Physics}, 17(4):114, 2024.

\bibitem[BGG{\v{S}}18]{bezakova2018inapproximability}
Ivona Bez{\'a}kov{\'a}, Andreas Galanis, Leslie~Ann Goldberg, and Daniel
  {\v{S}}tefankovi{\v{c}}.
\newblock Inapproximability of the independent set polynomial in the complex
  plane.
\newblock In {\em Proceedings of the 50th annual ACM SIGACT symposium on theory
  of computing}, pages 1234--1240, 2018.

\bibitem[BGMZ22]{basso2022performance}
Joao Basso, David Gamarnik, Song Mei, and Leo Zhou.
\newblock Performance and limitations of the qaoa at constant levels on large
  sparse hypergraphs and spin glass models.
\newblock In {\em 2022 IEEE 63rd Annual Symposium on Foundations of Computer
  Science (FOCS)}, pages 335--343. IEEE, 2022.

\bibitem[BHK{\etalchar{+}}19]{barak2019nearly}
Boaz Barak, Samuel Hopkins, Jonathan Kelner, Pravesh~K Kothari, Ankur Moitra,
  and Aaron Potechin.
\newblock A nearly tight sum-of-squares lower bound for the planted clique
  problem.
\newblock {\em SIAM Journal on Computing}, 48(2):687--735, 2019.

\bibitem[BHL{\etalchar{+}}25]{bencs2025zeros}
Ferenc Bencs, Brice Huang, Daniel~Z Lee, Kuikui Liu, and Guus Regts.
\newblock On zeros and algorithms for disordered systems: mean-field spin
  glasses.
\newblock {\em arXiv preprint arXiv:2507.15616}, 2025.

\bibitem[BK19]{brandao2019finite}
Fernando~GSL Brandao and Michael~J Kastoryano.
\newblock Finite correlation length implies efficient preparation of quantum
  thermal states.
\newblock {\em Communications in Mathematical Physics}, 365(1):1--16, 2019.

\bibitem[BLMT24]{bakshi2024high}
Ainesh Bakshi, Allen Liu, Ankur Moitra, and Ewin Tang.
\newblock High-temperature gibbs states are unentangled and efficiently
  preparable.
\newblock In {\em 2024 IEEE 65th Annual Symposium on Foundations of Computer
  Science (FOCS)}, pages 1027--1036. IEEE, 2024.

\bibitem[Bol14]{bolthausen2014iterative}
Erwin Bolthausen.
\newblock An iterative construction of solutions of the tap equations for the
  sherrington--kirkpatrick model.
\newblock {\em Communications in Mathematical Physics}, 325(1):333--366, 2014.

\bibitem[BS20]{baldwin2020quenched}
CL~Baldwin and B~Swingle.
\newblock Quenched vs annealed: Glassiness from sk to syk.
\newblock {\em Physical Review X}, 10(3):031026, 2020.

\bibitem[BS26]{bettaque2026magic}
Val{\'e}rie Bettaque and Brian Swingle.
\newblock Magic and wormholes in the sachdev-ye-kitaev model.
\newblock {\em arXiv preprint arXiv:2602.12339}, 2026.

\bibitem[CBDL25]{chen2025quantum}
Zherui Chen, Joao Basso, Zhiyan Ding, and Lin Lin.
\newblock Quantum replica exchange.
\newblock {\em arXiv preprint arXiv:2510.07291}, 2025.

\bibitem[Cha10]{chatterjee2010spin}
Sourav Chatterjee.
\newblock Spin glasses and stein’s method.
\newblock {\em Probability theory and related fields}, 148(3):567--600, 2010.

\bibitem[Che13]{chen2013aizenman}
Wei-Kuo Chen.
\newblock The {Aizenman-Sims-Starr} scheme and parisi formula for mixed
  $p$-spin spherical models.
\newblock {\em Electron. J. Probab.}, 18(none), January 2013.

\bibitem[CKBG25]{chen2025efficient}
Chi-Fang Chen, Michael Kastoryano, Fernando~GSL Brand{\~a}o, and Andr{\'a}s
  Gily{\'e}n.
\newblock Efficient quantum thermal simulation.
\newblock {\em Nature}, 646(8085):561--566, 2025.

\bibitem[CO10]{coja2010better}
Amin Coja-Oghlan.
\newblock A better algorithm for random k-sat.
\newblock {\em SIAM Journal on Computing}, 39(7):2823--2864, 2010.

\bibitem[Cra07]{crawford2007thermodynamics}
Nicholas Crawford.
\newblock Thermodynamics and universality for mean field quantum spin glasses.
\newblock {\em Communications in mathematical physics}, 274(3):821--839, 2007.

\bibitem[CT21]{chen2021convergence}
Wei-Kuo Chen and Si~Tang.
\newblock On convergence of the cavity and bolthausen’s tap iterations to the
  local magnetization.
\newblock {\em Communications in Mathematical Physics}, 386(2):1209--1242,
  2021.

\bibitem[Del73]{delsarte1973algebraic}
Philippe Delsarte.
\newblock An algebraic approach to the association schemes of coding theory.
\newblock {\em Philips Res. Rep. Suppl.}, 10:vi+--97, 1973.

\bibitem[DLL25]{ding2025efficient}
Zhiyan Ding, Bowen Li, and Lin Lin.
\newblock Efficient quantum gibbs samplers with kubo--martin--schwinger
  detailed balance condition.
\newblock {\em Communications in Mathematical Physics}, 406(3):67, 2025.

\bibitem[Dob96]{dobrushin1996estimates}
Roland~L Dobrushin.
\newblock Estimates of semi-invariants for the ising model at low temperatures.
\newblock {\em Translations of the American Mathematical Society-Series 2},
  177:59--82, 1996.

\bibitem[DS85]{dobrushin1985completely}
Roland~L Dobrushin and Senya~B Shlosman.
\newblock Completely analytical gibbs fields.
\newblock In {\em Statistical Physics and Dynamical Systems: Rigorous Results},
  pages 371--403. Springer, 1985.

\bibitem[DZPL25]{ding2025end}
Zhiyan Ding, Yongtao Zhan, John Preskill, and Lin Lin.
\newblock End-to-end efficient quantum thermal and ground state preparation
  made simple.
\newblock {\em arXiv preprint arXiv:2508.05703}, 2025.

\bibitem[EAMS22]{el2022sampling}
Ahmed El~Alaoui, Andrea Montanari, and Mark Sellke.
\newblock Sampling from the sherrington-kirkpatrick gibbs measure via
  algorithmic stochastic localization.
\newblock In {\em 2022 IEEE 63rd Annual Symposium on Foundations of Computer
  Science (FOCS)}, pages 323--334. IEEE, 2022.

\bibitem[EAMS25]{el2025sampling}
Ahmed El~Alaoui, Andrea Montanari, and Mark Sellke.
\newblock Sampling from mean-field gibbs measures via diffusion processes.
\newblock {\em Probability and Mathematical Physics}, 6(3):961--1022, 2025.

\bibitem[EM18]{eldar2018approximating}
Lior Eldar and Saeed Mehraban.
\newblock Approximating the permanent of a random matrix with vanishing mean.
\newblock In {\em 2018 IEEE 59th Annual Symposium on Foundations of Computer
  Science (FOCS)}, pages 23--34. IEEE, 2018.

\bibitem[Fis65]{fisher1965nature}
Michael~E. Fisher.
\newblock The nature of critical points.
\newblock In {\em Lectures in Theoretical Physics}, volume~7C. University of
  Colorado Press, Boulder, 1965.

\bibitem[FP07]{fernandez2007cluster}
Roberto Fern{\'a}ndez and Aldo Procacci.
\newblock Cluster expansion for abstract polymer models. new bounds from an old
  approach.
\newblock {\em Communications in Mathematical Physics}, 274(1):123--140, 2007.

\bibitem[FTW19]{feng2019spectrum}
Renjie Feng, Gang Tian, and Dongyi Wei.
\newblock Spectrum of syk model.
\newblock {\em Peking Mathematical Journal}, 2(1):41--70, 2019.

\bibitem[Gam21]{gamarnik2021overlap}
David Gamarnik.
\newblock The overlap gap property: A topological barrier to optimizing over
  random structures.
\newblock {\em Proceedings of the National Academy of Sciences},
  118(41):e2108492118, 2021.

\bibitem[GCDK24]{gilyen2024quantum}
Andr{\'a}s Gily{\'e}n, Chi-Fang Chen, Joao~F Doriguello, and Michael~J
  Kastoryano.
\newblock Quantum generalizations of glauber and metropolis dynamics.
\newblock {\em arXiv preprint arXiv:2405.20322}, 2024.

\bibitem[GLL20]{guo2020zeros}
Heng Guo, Jingcheng Liu, and Pinyan Lu.
\newblock Zeros of ferromagnetic 2-spin systems.
\newblock In {\em Proceedings of the Fourteenth Annual ACM-SIAM Symposium on
  Discrete Algorithms}, pages 181--192. SIAM, 2020.

\bibitem[{Goo}25]{google2025observation}
{Google Quantum AI and Collaborators}.
\newblock Observation of constructive interference at the edge of quantum
  ergodicity.
\newblock {\em Nature}, 646(8086):825--830, Oct 2025.

\bibitem[GPS00]{georges2000mean}
Antoine Georges, Olivier Parcollet, and Subir Sachdev.
\newblock Mean field theory of a quantum heisenberg spin glass.
\newblock {\em Physical review letters}, 85(4):840, 2000.

\bibitem[GPS01]{georges2001quantum}
Antoine Georges, Olivier Parcollet, and Subir Sachdev.
\newblock Quantum fluctuations of a nearly critical heisenberg spin glass.
\newblock {\em Physical Review B}, 63(13):134406, 2001.

\bibitem[Gre69]{greenberg1969correlation}
William Greenberg.
\newblock Correlation functionals of infinite volume quantum spin systems.
\newblock {\em Communications in Mathematical Physics}, 11(4):314--320, 1969.

\bibitem[GT02]{guerra2002thermodynamic}
Francesco Guerra and Fabio~Lucio Toninelli.
\newblock The thermodynamic limit in mean field spin glass models.
\newblock {\em Communications in Mathematical Physics}, 230(1):71--79, 2002.

\bibitem[Gue03]{guerra2003broken}
Francesco Guerra.
\newblock Broken replica symmetry bounds in the mean field spin glass model.
\newblock {\em Communications in mathematical physics}, 233(1):1--12, 2003.

\bibitem[Has23]{hastings2023field}
Matthew~B Hastings.
\newblock Field theory and the sum-of-squares for quantum systems.
\newblock {\em arXiv preprint arXiv:2302.14006}, 2023.

\bibitem[HE23]{hangleiter2023computational}
Dominik Hangleiter and Jens Eisert.
\newblock Computational advantage of quantum random sampling.
\newblock {\em Reviews of Modern Physics}, 95(3):035001, 2023.

\bibitem[Het16]{hetterich2016analysing}
Samuel Hetterich.
\newblock Analysing survey propagation guided decimation on random formulas.
\newblock {\em arXiv preprint arXiv:1602.08519}, 2016.

\bibitem[HMS20]{harrow2020classical}
Aram~W Harrow, Saeed Mehraban, and Mehdi Soleimanifar.
\newblock Classical algorithms, correlation decay, and complex zeros of
  partition functions of quantum many-body systems.
\newblock In {\em Proceedings of the 52nd Annual ACM SIGACT Symposium on Theory
  of Computing}, pages 378--386, 2020.

\bibitem[HO22]{hastings2022optimizing}
Matthew~B Hastings and Ryan O'Donnell.
\newblock Optimizing strongly interacting fermionic hamiltonians.
\newblock In {\em Proceedings of the 54th annual ACM SIGACT symposium on theory
  of computing}, pages 776--789, 2022.

\bibitem[Hop18]{hopkins2018statistical}
Samuel Hopkins.
\newblock {\em Statistical inference and the sum of squares method}.
\newblock Cornell University, 2018.

\bibitem[HS17]{hopkins2017bayesian}
Samuel~B Hopkins and David Steurer.
\newblock Bayesian estimation from few samples: community detection and related
  problems.
\newblock {\em arXiv preprint arXiv:1710.00264}, 2017.

\bibitem[JCSH25]{jiang2025positive}
Jiaqing Jiang, Jielun Chen, Norbert Schuch, and Dominik Hangleiter.
\newblock Positive bias makes tensor-network contraction tractable.
\newblock In {\em Proceedings of the 57th Annual ACM Symposium on Theory of
  Computing}, pages 471--482, 2025.

\bibitem[JI24]{jiang2024quantum}
Jiaqing Jiang and Sandy Irani.
\newblock Quantum metropolis sampling via weak measurement.
\newblock {\em arXiv preprint arXiv:2406.16023}, 2024.

\bibitem[KGKB25]{king2025triply}
Robbie King, David Gosset, Robin Kothari, and Ryan Babbush.
\newblock Triply efficient shadow tomography.
\newblock In {\em Proceedings of the 2025 Annual ACM-SIAM Symposium on Discrete
  Algorithms (SODA)}, pages 914--946. SIAM, 2025.

\bibitem[Kit15a]{kitaev2015SYK}
Alexei Kitaev.
\newblock A simple model of quantum holography (part 1).
\newblock Talk at the Kavli Institute for Theoretical Physics program
  ``Entanglement in Strongly-Correlated Quantum Matter'', April 2015.
\newblock \url{https://online.kitp.ucsb.edu/online/entangled15/}.

\bibitem[Kit15b]{kitaev2015talks}
AY~Kitaev.
\newblock Talks at kitp, university of california, santa barbara.
\newblock {\em Entanglement in Strongly-Correlated Quantum Matter}, 2015.

\bibitem[KL21]{khramtsov2021spectral}
Mikhail Khramtsov and Elena Lanina.
\newblock Spectral form factor in the double-scaled syk model.
\newblock {\em Journal of High Energy Physics}, 2021(3):1--38, 2021.

\bibitem[KP86]{kotecky1986cluster}
Roman Koteck{\`y} and David Preiss.
\newblock Cluster expansion for abstract polymer models.
\newblock {\em Communications in Mathematical Physics}, 103(3):491--498, 1986.

\bibitem[KS94]{kirkpatrick1994critical}
Scott Kirkpatrick and Bart Selman.
\newblock Critical behavior in the satisfiability of random boolean
  expressions.
\newblock {\em Science}, 264(5163):1297--1301, 1994.

\bibitem[KS18]{kitaev2018soft}
Alexei Kitaev and S~Josephine Suh.
\newblock The soft mode in the {Sachdev}-{Ye}-{Kitaev} model and its gravity
  dual.
\newblock {\em J. High Energy Phys.}, 2018(5):1--68, 2018.

\bibitem[LCB18]{liu2018quantum}
Chunxiao Liu, Xiao Chen, and Leon Balents.
\newblock Quantum entanglement of the sachdev-ye-kitaev models.
\newblock {\em Physical Review B}, 97(24):245126, 2018.

\bibitem[LO69]{larkin1969quasiclassical}
Anatoly~I Larkin and Yu~N Ovchinnikov.
\newblock Quasiclassical method in the theory of superconductivity.
\newblock {\em Sov Phys JETP}, 28(6):1200--1205, 1969.

\bibitem[LRRS21]{leschke2021free}
Hajo Leschke, Sebastian Rothlauf, Rainer Ruder, and Wolfgang Spitzer.
\newblock The free energy of a quantum sherrington--kirkpatrick spin-glass
  model for weak disorder.
\newblock {\em Journal of Statistical Physics}, 182:1--41, 2021.

\bibitem[LSS19]{liu2019ising}
Jingcheng Liu, Alistair Sinclair, and Piyush Srivastava.
\newblock The ising partition function: Zeros and deterministic approximation:
  J. liu et al.
\newblock {\em Journal of Statistical Physics}, 174(2):287--315, 2019.

\bibitem[LSS25]{liu2025correlation}
Jingcheng Liu, Alistair Sinclair, and Piyush Srivastava.
\newblock Correlation decay and partition function zeros: Algorithms and phase
  transitions.
\newblock {\em SIAM Journal on Computing}, 54(4):FOCS19--200, 2025.

\bibitem[LY52]{lee1952statistical}
Tsung-Dao Lee and Chen-Ning Yang.
\newblock Statistical theory of equations of state and phase transitions. ii.
  lattice gas and ising model.
\newblock {\em Physical Review}, 87(3):410, 1952.

\bibitem[MEAG{\etalchar{+}}20]{mcardle2020quantum}
Sam McArdle, Suguru Endo, Al{\'a}n Aspuru-Guzik, Simon~C Benjamin, and Xiao
  Yuan.
\newblock Quantum computational chemistry.
\newblock {\em Reviews of Modern Physics}, 92(1):015003, 2020.

\bibitem[MH21]{mann2021efficient}
Ryan~L Mann and Tyler Helmuth.
\newblock Efficient algorithms for approximating quantum partition functions.
\newblock {\em Journal of Mathematical Physics}, 62(2), 2021.

\bibitem[MM21]{maldacena2021syk}
Juan Maldacena and Alexey Milekhin.
\newblock Syk wormhole formation in real time.
\newblock {\em Journal of High Energy Physics}, 2021(4):258, 2021.

\bibitem[MM24]{mann2024algorithmic}
Ryan~L Mann and Romy~M Minko.
\newblock Algorithmic cluster expansions for quantum problems.
\newblock {\em PRX Quantum}, 5(1):010305, 2024.

\bibitem[MMZ05]{mezard2005clustering}
Marc M{\'e}zard, Thierry Mora, and Riccardo Zecchina.
\newblock Clustering of solutions in the random satisfiability problem.
\newblock {\em Physical Review Letters}, 94(19):197205, 2005.

\bibitem[Mon25]{montanari2025optimization}
Andrea Montanari.
\newblock Optimization of the sherrington--kirkpatrick hamiltonian.
\newblock {\em SIAM Journal on Computing}, 54(4):FOCS19--1, 2025.

\bibitem[MPV87]{mezard1987spin}
Marc M{\'e}zard, Giorgio Parisi, and Miguel~Angel Virasoro.
\newblock {\em Spin glass theory and beyond: An Introduction to the Replica
  Method and Its Applications}, volume~9.
\newblock World Scientific Publishing Company, 1987.

\bibitem[MS16]{maldacena2016remarks}
Juan Maldacena and Douglas Stanford.
\newblock Remarks on the sachdev-ye-kitaev model.
\newblock {\em Physical Review D}, 94(10):106002, 2016.

\bibitem[MSS16]{maldacena2016bound}
Juan Maldacena, Stephen~H Shenker, and Douglas Stanford.
\newblock A bound on chaos.
\newblock {\em Journal of High Energy Physics}, 2016(8):106, 2016.

\bibitem[MZK{\etalchar{+}}99]{monasson1999determining}
R{\'e}mi Monasson, Riccardo Zecchina, Scott Kirkpatrick, Bart Selman, and
  Lidror Troyansky.
\newblock Determining computational complexity from characteristic ‘phase
  transitions’.
\newblock {\em Nature}, 400(6740):133--137, 1999.

\bibitem[Pan13]{panchenko2013sherrington}
Dmitry Panchenko.
\newblock {\em The sherrington-kirkpatrick model}.
\newblock Springer Science \& Business Media, 2013.

\bibitem[Pan14]{panchenko2014parisi}
Dmitry Panchenko.
\newblock The parisi formula for mixed $p$-spin models.
\newblock {\em Ann. Probab.}, 42(3):946--958, May 2014.

\bibitem[Par79]{parisi1979infinite}
Giorgio Parisi.
\newblock Infinite number of order parameters for spin-glasses.
\newblock {\em Physical Review Letters}, 43(23):1754, 1979.

\bibitem[Par82]{park1982cluster}
Yong~Moon Park.
\newblock The cluster expansion for classical and quantum lattice systems.
\newblock {\em Journal of Statistical Physics}, 27(3):553--576, 1982.

\bibitem[PG99]{parcollet1999non}
Olivier Parcollet and Antoine Georges.
\newblock Non-fermi-liquid regime of a doped mott insulator.
\newblock {\em Physical Review B}, 59(8):5341, 1999.

\bibitem[Ple82]{plefka1982convergence}
Timm Plefka.
\newblock Convergence condition of the tap equation for the infinite-ranged
  ising spin glass model.
\newblock {\em Journal of Physics A: Mathematical and general},
  15(6):1971--1978, 1982.

\bibitem[Ple02]{plefka2002modified}
T~Plefka.
\newblock Modified thouless-anderson-palmer equations for the
  sherrington-kirkpatrick spin glass: Numerical solutions.
\newblock {\em Physical Review B}, 65(22):224206, 2002.

\bibitem[PR17]{patel2017deterministic}
Viresh Patel and Guus Regts.
\newblock Deterministic polynomial-time approximation algorithms for partition
  functions and graph polynomials.
\newblock {\em SIAM Journal on Computing}, 46(6):1893--1919, 2017.

\bibitem[PR19]{peters2019conjecture}
Han Peters and Guus Regts.
\newblock On a conjecture of sokal concerning roots of the independence
  polynomial.
\newblock {\em Michigan Mathematical Journal}, 68(1):33--55, 2019.

\bibitem[PWD{\etalchar{+}}21]{pan2021yukawa}
Gaopei Pan, Wei Wang, Andrew Davis, Yuxuan Wang, and Zi~Yang Meng.
\newblock Yukawa-syk model and self-tuned quantum criticality.
\newblock {\em Physical Review Research}, 3(1):013250, 2021.

\bibitem[RCTJ25]{ramkumar2025high}
Akshar Ramkumar, Yiyi Cai, Yu~Tong, and Jiaqing Jiang.
\newblock High-temperature fermionic gibbs states are mixtures of gaussian
  states.
\newblock {\em arXiv preprint arXiv:2505.09730}, 2025.

\bibitem[Reg23]{regts2023absence}
Guus Regts.
\newblock Absence of zeros implies strong spatial mixing.
\newblock {\em Probability Theory and Related Fields}, 186(1):621--641, 2023.

\bibitem[RFA24]{rouze2024optimal}
Cambyse Rouz{\'e}, Daniel~Stilck Fran{\c{c}}a, and {\'A}lvaro~M Alhambra.
\newblock Optimal quantum algorithm for gibbs state preparation.
\newblock {\em arXiv preprint arXiv:2411.04885}, 2024.

\bibitem[RFA25]{rouze2025efficient}
Cambyse Rouz{\'e}, Daniel~Stilck Fran{\c{c}}a, and {\'A}lvaro~M Alhambra.
\newblock Efficient thermalization and universal quantum computing with quantum
  gibbs samplers.
\newblock In {\em Proceedings of the 57th Annual ACM Symposium on Theory of
  Computing}, pages 1488--1495, 2025.

\bibitem[RW24]{rajakumar2024gibbs}
Joel Rajakumar and James~D Watson.
\newblock Gibbs sampling gives quantum advantage at constant temperatures with
  $ o (1) $-local hamiltonians.
\newblock {\em arXiv preprint arXiv:2408.01516}, 2024.

\bibitem[RY17]{roberts2017chaos}
Daniel~A Roberts and Beni Yoshida.
\newblock Chaos and complexity by design.
\newblock {\em Journal of High Energy Physics}, 2017(4):1--64, 2017.

\bibitem[Sac23]{sachdev2023quantum}
Subir Sachdev.
\newblock {\em Quantum phases of matter}.
\newblock Cambridge University Press, 2023.

\bibitem[SBSSH16]{swingle2016measuring}
Brian Swingle, Gregory Bentsen, Monika Schleier-Smith, and Patrick Hayden.
\newblock Measuring the scrambling of quantum information.
\newblock {\em Physical Review A}, 94(4):040302, 2016.

\bibitem[SK75]{sherrington1975solvable}
David Sherrington and Scott Kirkpatrick.
\newblock Solvable model of a spin-glass.
\newblock {\em Physical review letters}, 35(26):1792, 1975.

\bibitem[SKS{\etalchar{+}}25]{schuster2025cooling}
Thomas Schuster, Bryce Kobrin, Vincent~P Su, Hugo Marrochio, and Norman~Y Yao.
\newblock Cooling the sachdev-ye-kitaev model using thermofield double states.
\newblock {\em arXiv preprint arXiv:2511.09620}, 2025.

\bibitem[SO12]{szabo2012modern}
Attila Szabo and Neil~S Ostlund.
\newblock {\em Modern quantum chemistry: introduction to advanced electronic
  structure theory}.
\newblock Courier Corporation, 2012.

\bibitem[SS21]{shao2021contraction}
Shuai Shao and Yuxin Sun.
\newblock Contraction: A unified perspective of correlation decay and
  zero-freeness of 2-spin systems.
\newblock {\em Journal of Statistical Physics}, 185(2):12, 2021.

\bibitem[SY93]{sachdev1993gapless}
Subir Sachdev and Jinwu Ye.
\newblock Gapless spin-fluid ground state in a random quantum heisenberg
  magnet.
\newblock {\em Physical Review Letters}, 70(21):3339--3342, 1993.

\bibitem[T{\etalchar{+}}15]{tropp2015introduction}
Joel~A Tropp et~al.
\newblock An introduction to matrix concentration inequalities.
\newblock {\em Foundations and Trends in Machine Learning}, 8(1-2):1--230,
  2015.

\bibitem[Tak11]{takahashi2011replica}
Kazutaka Takahashi.
\newblock Replica analysis of partition-function zeros in spin-glass models.
\newblock {\em Journal of Physics A: Mathematical and Theoretical},
  44(23):235001, 2011.

\bibitem[Tal03]{talagrand2003spin}
Michel Talagrand.
\newblock {\em Spin glasses: a challenge for mathematicians: cavity and mean
  field models}, volume~46.
\newblock Springer Science \& Business Media, 2003.

\bibitem[Tal06a]{talagrand2006free}
Michel Talagrand.
\newblock Free energy of the spherical mean field model.
\newblock {\em Probability theory and related fields}, 134(3):339--382, 2006.

\bibitem[Tal06b]{talagrand2006parisi}
Michel Talagrand.
\newblock The parisi formula.
\newblock {\em Annals of mathematics}, pages 221--263, 2006.

\bibitem[Tal10]{talagrand2010mean}
Michel Talagrand.
\newblock {\em Mean field models for spin glasses: Volume I: Basic examples},
  volume~54.
\newblock Springer Science \& Business Media, 2010.

\bibitem[TAP77]{thouless1977solution}
D.~J. Thouless, P.~W. Anderson, and R.~G. Palmer.
\newblock Solution of {`Solvable model of a spin glass'}.
\newblock {\em The Philosophical Magazine: A Journal of Theoretical
  Experimental and Applied Physics}, 35(3):593--601, 1977.

\bibitem[TO13]{takahashi2013zeros}
K~Takahashi and T~Obuchi.
\newblock Zeros of the partition function and dynamical singularities in
  spin-glass systems.
\newblock {\em Journal of Physics: Conference Series}, 473(1):012023, dec 2013.

\bibitem[TOV{\etalchar{+}}11]{temme2011quantum}
Kristan Temme, Tobias~J Osborne, Karl~G Vollbrecht, David Poulin, and Frank
  Verstraete.
\newblock Quantum metropolis sampling.
\newblock {\em Nature}, 471(7336):87--90, 2011.

\bibitem[TW05]{troyer2005computational}
Matthias Troyer and Uwe-Jens Wiese.
\newblock Computational complexity and fundamental limitations to fermionic
  quantum monte carlo simulations.
\newblock {\em Physical review letters}, 94(17):170201, 2005.

\bibitem[Vid03]{vidal2003efficient}
Guifr{\'e} Vidal.
\newblock Efficient classical simulation of slightly entangled quantum
  computations.
\newblock {\em Physical review letters}, 91(14):147902, 2003.

\bibitem[YYZ22]{yao2022polynomial}
Penghui Yao, Yitong Yin, and Xinyuan Zhang.
\newblock Polynomial-time approximation of zero-free partition functions.
\newblock {\em arXiv preprint arXiv:2201.12772}, 2022.

\bibitem[ZBC23]{zhang2023dissipative}
Daniel Zhang, Jan~Lukas Bosse, and Toby Cubitt.
\newblock Dissipative quantum gibbs sampling.
\newblock {\em arXiv preprint arXiv:2304.04526}, 2023.

\bibitem[Zha19]{zhang2019evaporation}
Pengfei Zhang.
\newblock Evaporation dynamics of the sachdev-ye-kitaev model.
\newblock {\em Physical Review B}, 100(24):245104, 2019.

\bibitem[ZK26]{zlokapa2026syk}
Alexander Zlokapa and Bobak~T Kiani.
\newblock Syk thermal expectations are classically easy at any temperature.
\newblock {\em arXiv preprint arXiv:2602.22619}, 2026.

\bibitem[ZKA25]{zlokapa2025average}
Alexander Zlokapa, Bobak~T Kiani, and Eric~R Anschuetz.
\newblock Average-case quantum complexity from glassiness.
\newblock {\em arXiv preprint arXiv:2510.08497}, 2025.

\end{thebibliography}
\bibliographystyle{alpha}
\newpage

\appendix

\section{Obstructions to classical algorithms for SYK}\label{sec:obstruct}

In this appendix we record three consequences of prior work that hold at any arbitrarily small constant $\beta > 0$. The first gives the purification complexity lower bound used in Remark~\ref{rem:obstruct}(1). The second shows that the Gibbs state cannot be written as a convex combination of superpositions of $\exp[o(n^{1/4})]$ orthogonal Gaussian states, as claimed in Remark~\ref{rem:obstruct}(2). For Remark~\ref{rem:obstruct}(3), we show that local observables fluctuate inverse polynomially with high probability over independent SYK instances. The first two results follow from small circuits and Gaussian states having poor energies compared to the typical energy of an eigenstate sampled at constant temperature. We formalize this using the following results from~\cite{anschuetz2025strongly} and using the notation introduced in Sec.~\ref{sec:prelim}, i.e., the Hamiltonian is
\begin{align}
    Z(\beta) = \Tr e^{-\beta H}, \quad H = \sum_{I \in \binom{[n]}{q}} J_I \psi_I, \quad J_I \sim \cN(0, \sigma^2), \quad \sigma^2 = \frac{(q-1)!}{n^{q-1}}, \quad \{\psi_i, \psi_j\} = \delta_{ij}, \quad \cJ = \sqrt{\frac{q}{2^q}}
\end{align}
and the commutation index is
\begin{align}
    \Delta = \sup_{\ket\phi} \binom{n}{q}^{-1} \sum_{I \in \binom{[n]}{q}} \bra{\phi}\psi_I\ket{\phi}^2.
\end{align}

\subsection{Negative Gibbs energy density at constant temperature}
\begin{theorem}[Commutation index and real self-averaging, Theorem I.2 of~\cite{anschuetz2025strongly}]\label{thm:ext-selfavg}
    For any constant $\beta \geq 0$,
    \begin{align}
        \left|\frac{1}{n}\E \log Z(\beta) - \frac{1}{n} \log \E \, Z(\beta)\right| = O(n^{-q/2}).
    \end{align}
\end{theorem}

\begin{theorem}[Annealed exponential lower bound, Lemma D.1 of~\cite{anschuetz2025strongly}]\label{thm:ext-annealed}
    Let
    \begin{align}
        \widetilde H = \frac{1}{\sqrt m}\sum_{i=1}^m g_i A_i,
        \qquad
        g_i \sim \cN(0,1),
    \end{align}
    where each $A_i$ is Hermitian and $A_i^2 = I$.  Define
    \begin{align}
        h_{\mathrm{comm}} = \frac12 \sup_i \sum_{j=1}^m \|[A_i,A_j]\|.
    \end{align}
    Then there is an absolute constant $c_1>0$ such that for every real $t$,
    \begin{align}
        \E \,\tr(e^{t\widetilde H})
        \ge
        \exp[
            \frac{t^2}{2}
            \left(
                1-\frac{c_1 t^2 h_{\mathrm{comm}}}{2m}
            \right)
        ].
    \end{align}
\end{theorem}
\begin{corollary}[SYK annealed exponential lower bound]\label{cor:annealed-syk}
    For all $0 \leq \beta \cJ \leq 1/\sqrt{c_1}$ and sufficiently large $n$,
    \begin{align}
        \log \E\, Z(\beta) \geq \frac{n}{2}\log 2 + \frac{\beta^2}{q2^{q+3}} n.
    \end{align}
\end{corollary}
\begin{proof}
    We set for $m = \binom{n}{q}$ and $g_a \sim \cN(0, 1)$
    \begin{align}
        \tilde H = \frac{1}{\sqrt m} \sum_{a=1}^m g_a A_{I_a}, \quad A_I = (2i)^{q/2}\psi_I
    \end{align}
    so $H = \lambda \tilde H$ for $\lambda = \sigma 2^{-q/2} \sqrt m$. Since $\norm{[A_I, A_J]} \leq 2$ if $I \cap J \neq \emptyset$ and is zero otherwise,
    \begin{align}
        h_\mathrm{comm} = \frac{1}{2} \sup_I \sum_J \norm{[A_I, A_J]} \leq \sup_I \#\{J\,:\,I\cap J \neq \emptyset\} \leq q \binom{n-1}{q-1}.
    \end{align}
    Theorem~\ref{thm:ext-annealed} thus gives, using $\binom{n-1}{q-1}\sigma^2 \leq 1$,
    \begin{align}
        \E\tr e^{-\beta H} \geq \exp[\frac{\beta^2 \lambda^2}{2}\lr{1 - \frac{c_1 \beta^2 \lambda^2 q}{2m} \binom{n-1}{q-1}}] \geq \exp[\frac{\beta^2 \lambda^2}{2}\lr{1 - \frac{c_1 q}{2^{q+1}} \beta^2}].
    \end{align}
    Hence, for all $\beta \leq \sqrt{\frac{2^q}{c_1 q}} = 1/\cJ\sqrt{c_1}$, we have using $\sigma^2 \binom{n}{q} \geq \frac{n}{2q}$ for sufficiently large $n$ that $\lambda^2 \geq n/q2^{q+1}$ and hence
    \begin{align}
        \E\tr e^{-\beta H} \geq \exp[\frac{\beta^2 \lambda^2}{4}] \geq \exp[\frac{\beta^2 n}{q2^{q+3}}].
    \end{align}
\end{proof}

\begin{lemma}[Negative Gibbs energy density]\label{lem:neg-energy}
    For any constant $\beta > 0$, there exists a constant $\kappa > 0$ such that
    \begin{align}
        \Tr(H\rho_\beta) \leq -\kappa n
    \end{align}
    with probability $1-o(1)$.
\end{lemma}
\begin{proof}
    By Corollary~\ref{cor:annealed-syk} and Theorem~\ref{thm:ext-selfavg}, we have for $0 \leq \beta \cJ \leq 1/\sqrt{c_1}$ that
    \begin{align}\label{eq:avg-logZ-lb}
        \E \log Z(\beta) \geq \frac{n}{2}\log 2 + \frac{\beta^2}{q2^{q+3}}n.
    \end{align}
    Differentiating with respect to the standard Gaussian disorder further gives
    \begin{align}
        \norm{\nabla \log Z(\beta)}_2^2 = \beta^2 \sigma^2 \sum_{I \in \binom{[n]}{q}} \Tr(\psi_I \rho_\beta)^2 = O(n^{1-q/2})
    \end{align}
    by Theorem~\ref{thm:comm}, so
    \begin{align}
        \log Z(\beta) \geq \frac{n}{2}\log 2 + \frac{\beta^2}{q2^{q+4}}n + o(n)
    \end{align}
    holds with probability $1-o(1)$ by Gaussian concentration. Finally, observe that
    \begin{align}
        \phi(\beta) = \log \Tr e^{-\beta H}
    \end{align}
    is convex since
    \begin{align}
        \phi'(\beta) = -\Tr(H\rho_\beta), \quad \phi''(\beta) = \Var_{\rho_\beta}(H) \geq 0
    \end{align}
    and thus with probability $1-o(1)$,
    \begin{align}
        -\Tr(H\rho_\beta) \geq \frac{\phi(\beta)-\phi(0)}{\beta} \geq \frac{\beta}{q2^{q+4}}n.
    \end{align}
    Since $-\Tr(H \rho_\beta)$ is nondecreasing, the lower bound $\frac{1/\cJ\sqrt{c_1}}{q2^{q+4}}n$ holds for any $\beta \cJ > 1/\sqrt{c_1}$.
\end{proof}

\subsection{Circuit complexity}

Remark~\ref{rem:obstruct}(1) follows from an argument similar to those presented in~\cite{anschuetz2025strongly}.

\begin{theorem}[Purification complexity of the SYK Gibbs state]\label{thm:purification-complexity}
    For any constant $\beta > 0$, there exists a constant $\epsilon > 0$ such that the minimum number of 2-local gates from a fixed finite universal gate set needed to prepare an $\epsilon$-approximation in trace distance of a purification of $\rho_\beta$, using arbitrary ancillas initialized in a product state, is $\Omega(n^{1+q/2}/\log n)$ with probability $1-o(1)$.
\end{theorem}
\begin{proof}
    Let $\ket{\Psi}$ be a purification of $\rho_\beta$. Since $\bra{\Psi}(-H \otimes I)\ket{\Psi}$ is a centered Gaussian variable with variance
    \begin{align}
        \Var(\bra{\Psi}(-H \otimes I)\ket{\Psi}) =  \sum_{I \in \binom{[n]}{q}} \sigma^2 \Tr(\psi_I\rho_\beta)^2 \leq \frac{n}{q} \Delta
    \end{align}
    by Theorem~\ref{thm:comm}, we also have for any constant $\delta$ that
    \begin{align}
        \pr{\bra{\Psi}(-H \otimes I)\ket{\Psi} \geq \delta n} \leq \exp[-\Omega(n^{1+q/2})].
    \end{align}
    For a fixed finite set of 2-local gates, the number of circuits with $G$ gates is at most $\exp[O(G \log(n+G))]$. By a union bound, the probability that a circuit with $G$ gates outputs a state $\ket{\Phi}$ with $\bra{\Phi}(-H \otimes I)\ket{\Phi} \geq \delta n$ is at most
    \begin{align}
        \exp[O(G \log(n+G)) - \Omega(n^{q/2+1})].
    \end{align}
    Hence, for some sufficiently small $c > 0$, $G \leq c n^{q/2+1}/\log n$ makes this probability $o(1)$. By Lemma~\ref{lem:neg-energy}, however, we know $\ket{\Psi}$ must satisfy for some constant $\kappa > 0$,
    \begin{align}
        \bra{\Psi}(-H \otimes I)\ket{\Psi} \geq \kappa n.
    \end{align}
    In particular, any state $\rho$ with $\Tr(H\rho) \leq -\delta n$ satisfies
    \begin{align}\label{eq:deltan}
        -\delta n \leq \Tr(H \rho) \leq \Tr(H \rho_\beta) + \norm{H}\norm{\rho-\rho_\beta}_1 \implies \norm{\rho-\rho_\beta}_1 \geq \frac{(\kappa-\delta) n}{\norm{H}}.
    \end{align}
    Since $\norm{H} = O(n)$ for the SYK model with probability $1-o(1)$ (see, e.g., \eqref{eq:syknorm} for a proof earlier in the text), we conclude that the trace distance can be made $\Omega(1)$ by choosing $\delta = \kappa/2$.
\end{proof}

\subsection{Gaussian states}

Remark~\ref{rem:obstruct}(2) follows from the following result of~\cite{hastings2023field} and a similar argument as in Theorem~\ref{thm:purification-complexity}.

\begin{theorem}[Few-Gaussian energy bound~\cite{hastings2023field}]\label{thm:ext-fewgauss}
    With probability $1-o(1)$ over the disorder, every state $\ket{\varphi}$ that is a superposition of at most $\exp[o(n^{1/4})]$ mutually orthogonal Gaussian states satisfies
    \begin{align}
        \abs{\bra{\varphi}H\ket{\varphi}} = o(1)
    \end{align}
    for the $q=4$ SYK Hamiltonian $H$.
\end{theorem}
\begin{corollary}
    For any $\beta > 0$, there exists a constant $\epsilon > 0$ such that any $\epsilon$-approximation in trace distance of the Gibbs state $\rho_\beta$ cannot be written as
    \begin{align}
        \rho_\beta = \sum_a p_a \ketbra{\varphi_a}, \quad p_a > 0, \quad \sum_a p_a = 1
    \end{align}
    where each $\ket{\varphi_a}$ is a superposition of at most $\exp[o(n^{1/4})]$ mutually orthogonal Gaussian states.
\end{corollary}
\begin{proof}
    Theorem~\ref{thm:ext-fewgauss} states that no constant $\delta > 0$ exists such that the states $\ket{\varphi_a}$ satisfy
    \begin{align}
        \sum_a p_a \bra{\varphi_a}H\ket{\varphi_a} \leq -\delta n,
    \end{align}
    contradicting Lemma~\ref{lem:neg-energy} for any constant $\beta > 0$. We union bound over the $o(1)$ bad events where the contradiction does not hold; the remaining proof follows exactly as \eqref{eq:deltan} in Theorem~\ref{thm:purification-complexity}.
\end{proof}

\subsection{Instance-to-instance fluctuations}
To show Remark~\ref{rem:obstruct}(3), we give a fairly simple probabilistic argument.

\begin{theorem}[Instance-to-instance inverse-polynomial fluctuations of local expectations]
    For any $\beta > 0$, let $\rho_\beta, \rho_\beta'$ be thermal states of two independent SYK instances. Then with probability $1-o(1)$, it holds that
    \begin{align}
        \max_{I \in \binom{[n]}{q}} \abs{\Tr(\psi_I \rho_\beta) - \Tr(\psi_I \rho_\beta')} = \Omega(n^{-(q-1)/2}).
    \end{align}
\end{theorem}
\begin{proof}
    Denote the length-$\binom{n}{q}$ disorder vectors of independent SYK Hamiltonians $H, H'$ by $J, J'$. Our starting point is Lemma~\ref{lem:neg-energy}, which gives a constant $\kappa > 0$ such that
    \begin{align}
        \kappa n &\leq -\Tr(H\rho_\beta) = -\sum_I J_I \Tr(\psi_I \rho_\beta) = -\sum_I J_I \Tr(\psi_I \rho'_\beta) - \sum_I J_I \lr{\Tr(\psi_I \rho_\beta)-\Tr(\psi_I \rho'_\beta)} \\
        &\leq \abs{\sum_I J_I \Tr(\psi_I \rho'_\beta)} + \norm{J}_1 \max_I \abs{\Tr(\psi_I \rho_\beta)-\Tr(\psi_I \rho'_\beta)}.
    \end{align}
    This gives the lower bound
    \begin{align}\label{eq:maxi}
        \max_I \abs{\Tr(\psi_I \rho_\beta)-\Tr(\psi_I \rho'_\beta)} \geq \frac{\kappa n - \abs{\sum_I J_I \Tr(\psi_I \rho'_\beta)}}{\norm{J}_1}.
    \end{align}
    We proceed to show bounds on $\norm{J}_1$ and $\abs{\sum_I J_I \Tr(\psi_I \rho'_\beta)}$, starting with the former.
    Since $J_I$ are independent Gaussians with variance $(q-1)!/n^{q-1}$, with probability $1-o(1)$ we have $\norm{J}_2 = O(\sqrt n)$ and hence
    \begin{align}\label{eq:j1}
        \norm{J}_1 \leq \sqrt{\binom{n}{q}} \norm{J}_2 = O(n^{(q+1)/2})
    \end{align}
    with probability $1-o(1)$. For the latter, we have by independence of $J, J'$ that conditioned on $J'$, the random variable
    \begin{align}
        \sum_I J_I \Tr(\psi_I \rho'_\beta) = \Tr(H \rho'_\beta)
    \end{align}
    is a centered Gaussian with variance
    \begin{align}
        \sigma^2 \sum_I \Tr(\psi_I \rho'_\beta)^2 \leq \sigma^2 \binom{n}{q} 2^{-q} \leq \frac{n}{q} 2^{-q} = O(n).
    \end{align}
    We conclude by applying the tail bound
    \begin{align}
        \pr{\abs{\Tr(H \rho'_\beta)} > n^{3/4} \,\big|\, J'} \leq \exp[-\Omega(n^{1/2})]
    \end{align}
    under the expectation over $J'$. Combining this with \eqref{eq:j1}, we find from \eqref{eq:maxi} that with probability $1-o(1)$,
    \begin{align}
        \max_I \abs{\Tr(\psi_I \rho_\beta)-\Tr(\psi_I \rho'_\beta)} \geq \frac{\kappa n - O(n^{3/4})}{O(n^{(q+1)/2})} \geq \Omega(n^{-(q-1)/2}).
    \end{align}
\end{proof}
\end{document}